\documentclass{article}
\usepackage{amsmath,amssymb,amsthm}
\usepackage{xypic}
\usepackage{authblk}
\usepackage{color}
\usepackage{graphicx}
\usepackage{float}
\usepackage{hyperref}
\usepackage{fancyhdr}
\hypersetup{colorlinks=true,linkcolor=blue,citecolor=blue}
\usepackage[titletoc]{appendix}
\usepackage{MnSymbol}
\usepackage[left=2cm,right=2cm,top=2.5cm,bottom=2.5cm]{geometry}

\newtheorem{theorem}{Theorem}
\newtheorem{proposition}[theorem]{Proposition}
\newtheorem{lemma}[theorem]{Lemma}
\newtheorem{definition}[theorem]{Definition}
\theoremstyle{remark}
\newtheorem{example}{Example}
\newtheorem{remark}{Remark}

\newcommand{\lvec}[1]{\overleftarrow{#1}}
\newcommand{\rvec}[1]{\overrightarrow{#1}}
\newcommand{\im}{\textrm{Im}}
\newcommand{\trqq}{T^*(\mathbb{R}\times Q\times Q)}
\newcommand{\trq}{\mathbb{R}\times T^*Q}
\newcommand{\lcf}{\lbrack\! \lbrack}
\newcommand{\rcf}{\rbrack\! \rbrack}
\newcommand{\pir}{\Pi_{\mathbb{R}}}
\newcommand{\rgroupoid}{\mathbb{R}\times G}
\newcommand{\rbase}{\mathbb{R}\times M}
\newcommand{\ralpha}{\alpha_{\mathbb{R}}}
\newcommand{\rbeta}{\beta_{\mathbb{R}}}
\newcommand{\repsilon}{\epsilon_{\mathbb{R}}}
\newcommand{\riota}{\iota_{\mathbb{R}}}
\newcommand{\rmultiplication}{m_{\mathbb{R}}}

\newcommand{\ra}{\mathbb{R}\times AG}
\newcommand{\rtau}{\tau_{\mathbb{R}}}
\newcommand{\rad}{\mathbb{R}\times A^*G}
\newcommand{\id}{\operatorname{id}}
\newcommand{\lag}{\mathcal{L}}
\newcommand{\ide}{\mathfrak{e}}

\author[1]{Sebasti\'an Ferraro\thanks{sferraro@uns.edu.ar}}
\author[2]{Manuel de Le\'on\thanks{ mdeleon@icmat.es}}
\author[3]{Juan Carlos Marrero\thanks{jcmarrer@ull.edu.es}}
\author[2]{David Mart\'in de Diego\thanks{david.martin@icmat.es}}
\author[2]{Miguel Vaquero\thanks{miguel.vaquero@icmat.es}}

\affil[1]{Universidad Nacional del Sur, CONICET, Departamento de Matem\'atica\authorcr Av. Alem 1253, 8000 Bah\'ia Blanca, Argentina \vspace*{0.2cm}}

\affil[2]{Instituto de Ciencias Matem\'aticas,  ICMAT\authorcr
c/ Nicol\'as Cabrera, n$^\textrm{o}$ 13-15, Campus Cantoblanco,UAM\authorcr
28049 Madrid, Spain \vspace*{0.2cm}}

\affil[3]{Unidad Asociada ULL-CSIC ``Geometr\'ia Diferencial y
  Mec\'anica Geom\'etrica'' \authorcr
 y Departamento de Matem\'aticas, Estad\'istica e IO,\authorcr
Facultad de Ciencias, ULL\authorcr
c/ Astrof\'isico Francisco S\'anchez, s/n\authorcr
38206 La Laguna - Tenerife, Canary Islands, Spain}

\date{}
\begin{document}

\title{On the Geometry of the Hamilton--Jacobi Equation\\  and\\
  Generating Functions}
\maketitle

\begin{abstract} 
In this paper we develop a geometric version of the
Hamil\-ton--Jacobi equation in the Poisson setting. Specifically, we ``geometrize'' what is
usually called a complete solution of the Hamilton--Jacobi equation. We use some well-known results about symplectic groupoids, in
particular cotangent groupoids, as a keystone for the construction
of our framework. Our methodology follows the
ambitious program proposed by A. Weinstein, \cite{Weinstein}, in order to develop geometric formulations of the dynamical behavior of Lagrangian and Hamiltonian systems on Lie algebroids
and Lie groupoids. This procedure allows us to take symmetries into
account, and, as a by-product, we recover results from
\cite{ChannellScovelII,GeHJ,GeMarsden}, but even in these situations
our approach is new. A theory of generating functions for the Poisson
structures considered here is also developed following the same
pattern, solving a longstanding problem of the area: how
to obtain a generating function for the identity tranformation  and
  the nearby Poisson automorphisms of Poisson
manifolds. A direct application of our results give
  the construction of a family of Poisson integrators, that is,
  integrators that conserve the underlying Poisson geometry. These
  integrators are implemented in the paper in benchmark problems. Some conclusions, current and future directions of research are shown at the end of the paper.
\end{abstract}

\vfill
\newpage
\tableofcontents

%%%%%%%%%%%%%%%%%%%%%%%%%%%%%%%%%%%%%%%%%%%%%%%%%%%%%%%%%%%%%%%%

\section{Introduction}\label{introduction}
\subsection{Motivation}\label{motivation}
 To find canonical changes of coordinates that reduce the
Hamiltonian function to a form such that the equations can be easily
integrated is a very useful procedure for the integration of the classical Hamilton's equations. As a consequence, this shows that the initial equations
are integrable. But, of course, the main problem  is to find these
particular canonical transformations. This problem is equivalent to
the determination of a large enough number of solutions of the
Hamilton--Jacobi equation.
This is the objective of the representation of canonical
transformations in terms of  generating functions and leads
to {\it complete solutions} of the Hamilton--Jacobi equations. The usefulness
of this method is highlighted in the following quote by V.I. Arnold, (see \cite{Arnoldmmcm}, p.~233):
\vspace{0.2cm}
\begin{quote}
{\it 
``The technique of generating functions for canonical transformations, developed by
Hamilton and Jacobi, is the most powerful method available for integrating
the differential equations of dynamics.''}

\vspace{0.1cm}

\hfill {\it -V.I. Arnold}
\end{quote}

\vspace{0.2cm}

The procedure described above is well-known in the classical case,
which geometrically corresponds to  the cotangent bundle of the
configuration manifold under consideration (we remark that a recent
geometric Hamilton-Jacobi theory, which includes Lagrangian and
Hamiltonian systems, was developed in \cite{Carinena}). Some research has
been also done in the Lie--Poisson
case (\cite{GeMarsden}) as well. The goal of the following
exposition is to introduce these two cases in order to motivate our
future constructions, which deepen and generalize these results in a
highly non-trivial way finding new and powerful applications.

\subsubsection{The Classical Case}\label{classical}

Let $Q$ be the $n$-dimensional configuration manifold of a mechanical system and let
$(T^*Q, \ \omega_Q, \ H)$ be a Hamiltonian system. In this
system $\omega_Q$ is the canonical symplectic structure of the
cotangent bundle, $T^*Q$. Along this paper $\pi_Q:T^*Q\rightarrow Q$ will be the natural projection of the cotangent bundle onto $Q$, and $H:T^*Q\rightarrow \mathbb{R}$ will denote the
Hamiltonian function. Associated with such a Hamiltonian there is a
Hamiltonian vector field, $X_H$, defined by $i_{X_H}\omega_Q=dH$. In
natural cotangent coordinates $(q^i, p_i)$ the symplectic structure reads
$\omega_Q=dq^i\wedge dp_i$ and the Hamiltonian vector field becomes \[X_H(q^i,p_i)=\displaystyle\frac{\partial
  H}{\partial p_ i}(q^i,p_i)\frac{\partial}{\partial q^i}-\frac{\partial
  H}{\partial q^i}(q^i,p_i)\frac{\partial}{\partial p_i}, \]
so the
equations of motion read
\begin{equation}\label{hamilton}
\begin{array}{rl}
\displaystyle\frac{dq^i}{dt}(t)=&\displaystyle\frac{\partial H}{\partial
  p_i}(q^i(t),p_i(t)),\\ \noalign{\bigskip}
\displaystyle\frac{dp_i}{dt}(t)=&-\displaystyle\frac{\partial H}{\partial q^i}(q^i(t),p_i(t)),
\end{array}
\end{equation}
for $i=1,\ldots,\ n$.

For the sake of simplicity, and in order to clarify the main ideas of
the paper, we start with a local coordinate description. Those ideas hold locally for any symplectic manifold in Darboux coordinates. The reader
interested in the details and proofs of the results presented here
is referred to \cite{AbrahamMarsden, Arnoldmmcm,Goldstein}. Assume that we
have found a function $S$ that depends on the time, $t$, the $(q^i)$-coordinates and
$n$ parameters, say $(x^i)$, $1\leq i\leq n$, so $S=S(t,q^i,x^i)$, satisfying the following two conditions:
\begin{enumerate}
\item \emph{Hamilton--Jacobi equation:} \begin{equation}\label{HamiltonJacobiK}\displaystyle\frac{\partial S}{\partial
    t}(t,q^i,x^i)+H(q^i,\frac{\partial S}{\partial q^i}(t,q^i,x^i))=K(t,x^i),\end{equation} where $K$ is a function that only depends on $t$ and $x^i$;

\item \emph{Non-degeneracy condition:} $\det(\displaystyle\frac{\partial^2S}{\partial
    q^i\partial x^j})\neq 0$,
\end{enumerate}
then, by the implicit function theorem we can make the following change
of coordinates $(t,q^i,p_i)\rightarrow (t,x^i,y_i)$ defined implicitly by
\begin{equation}\label{coordinatechange}
\begin{array}{cc}
\displaystyle\frac{\partial S}{\partial q^i}(t,q^i,x^i)=p_i, & -\displaystyle\frac{\partial S}{\partial x^i}(t,q^i,x^i)=y_i.
\end{array}
\end{equation}
After some brief computations, one can see that in the new coordinates
$(t,x^i,y_i)$ the equations of motion are again in Hamiltonian form, but
now the Hamiltonian is  the function $K(t,x^i)$, {\it i.e}. the
equations \eqref{hamilton} read now
\begin{equation}\label{hamilton2}
\begin{array}{rl}
\displaystyle\frac{dx^i}{dt}(t)=&\displaystyle\frac{\partial K}{\partial
  y_i}(t,x^i(t))=0,\\ \noalign{\bigskip}
\displaystyle\frac{dy_i}{dt}(t)=&-\displaystyle\frac{\partial K}{\partial x^i}(t,x^i(t)),
\end{array}
\end{equation}
for $i=1,\ldots,\ n$.

Since $K$ only depends on the time and the $(x^i)$-coordinates, these equations are
trivially integrable. Given an initial condition $(x^i_0,y_i^0)$ at time
$t_ 0$, the curve \[t\rightarrow \Big(x^i_0,y_i^0-\int_{t_0}^t\displaystyle\frac{\partial K}{\partial x^i}(t,x^i_0) dt\Big)\] is
the solution of equations \eqref{hamilton2} with initial condition  $(x^i_0,y_i^0)$.

\begin{remark}
The equation \begin{equation}\label{HamiltonJacobi}\displaystyle\frac{\partial S}{\partial
    t}(t,q^i,x^i)+H(q^i,\frac{\partial S}{\partial q^i}(t,q^i,x^i))=0\end{equation} appears frequently in the
  literature. If one is able to find the function $S$ satisfying 
  equation \eqref{HamiltonJacobi} and the above non-degeneracy condition, that means that
  $K=0$ and so the equations of motion become $\frac{dx^i}{dt}(t)=0$
  and $\frac{dy_i}{dt}(t)=0$. This means that in the new
  coordinates the system is in \textbf{``equilibrium''}; it does not
  evolve at all! The inverse of that change of variables gives the
  \textbf{flow}, up to an initial transformation given by
  $S(0,q^i,x^i)$, later on we will clarify this claim. In general, any $S$ satisfying the non-degeneracy condition will induce a canonical transformation implicitly by the rule described above, which implies that Hamilton's equations in the $(q^i,p_i)$ coordinates will remain as Hamilton's equations in the $(x^i,y_i)$ coordinates for a new Hamiltonian, say $K$, which is related to the Hamilton--Jacobi equation by the expression
	\begin{equation}\label{changehamiltonian}
	\displaystyle\frac{\partial S}{\partial t}+H(t,q^i,\frac{\partial S}{\partial q^i})=K(t,x^i,y_i).
	\end{equation}
	Observe that equations \eqref{HamiltonJacobiK} and \eqref{HamiltonJacobi} are particular instances of the last equation. We will elaborate on these and related issues in Section \ref{Generating functions for a class of Poisson Manifolds}.
\end{remark}

We proceed now to give a geometric framework for the previous
procedure. A
  nice exposition of Lagrangian submanifolds, generating functions and
  related topics can be found in \cite{CannasDaSilva}. The function $S$, satisfying \eqref{HamiltonJacobiK}, is
interpreted here as a function on the
product manifold $\mathbb{R}\times Q\times Q$ and so $\im(dS)$ is a
Lagrangian submanifold in $\trqq$. Notice that we are thinking about the
$(q^i)$ as coordinates on the first $Q$, and $(x^i)$ as coordinates on the
second factor $Q$. This interpretation is directly related to the fact
that we are describing here \textit{\textbf{type I}} generating functions in
the language of \cite{Goldstein}. Other types of generating functions
will be introduced along the next sections. On the other hand, consider the
projections $\pi_I:\trqq\rightarrow \mathbb{R}\times T^*Q$, $I=1,2$,
defined by $\pi_2(t,e,x^i,y_i,q^i,p_i)=(t,q^i,p_i)$ and $\pi_1(t,e,x^i,y_i,q^i,p_i)=(t,x^i,-y_i)$. With these
geometric tools the \emph{non-degeneracy condition} is equivalent to
saying that ${\pi_I}_{|\im(dS)}$ are  local diffeomorphisms for
$I=1,2$. We assume here for simplicity that they are global
diffeomorphisms, so we can consider the mapping $
 {\pi_1}_{|\im(dS)}\circ ({\pi_2}_{|\im(dS)}) ^{-1}:\trq\rightarrow
 \trq$. The local argument follows with the obvious restrictions to
 open sets. This mapping can be easily checked to be the geometric
 description of the change of variables introduced in
 \eqref{coordinatechange}. The \emph{Hamilton--Jacobi} equation \eqref{HamiltonJacobiK} can be
 understood as the fact that $dS^*(\pi_2^*H+e)$ must be equal to
 $pr_1^*K$ for some $K\in C^\infty(\mathbb{R}\times Q)$, where $pr_I:\mathbb{R}\times Q\times Q\rightarrow
 \mathbb{R}\times  Q , \ I=1,  \ 2$  are $pr_1(t,q^i,x^i)=(t,x^i), $ $ pr_2(t,q^i,x^i)=(t,q^i)$. The diagram below illustrates the situation.
\begin{figure}[H]\[
\xymatrix{
&&\im(dS)\subset T^*(\mathbb{R}\times Q\times
Q)\ar[r]^>>>{\pi_2^*H+e\qquad }\ar[dl]^{\pi_1}
\ar[dr]_{\pi_2}\ar[dd]^>>>{\pi_{\mathbb{R}\times Q\times Q}}&\mathbb{R}\\
&\mathbb{R}\times T^*Q\ar[dd]_{\id_{\mathbb{R}}\times \pi_Q}
&&\mathbb{R}\times T^*Q \ar@/_/[lu]_<<{({\pi_2}_{|\im(dS)}) ^{-1} }
\ar[dd]^{\id_{\mathbb{R}}\times \pi_Q}\ar[ll]^<<{ {\pi_1}\circ ({\pi_2}_{|\im(dS)}) ^{-1}\qquad \qquad \qquad\qquad}|\hole\ar[r]^-H&\mathbb{R}\\
&&\mathbb{R}\times Q\times Q \ar@/^/[uu]^>>>>>>>{dS}|\hole\ar[dl]_{pr_1}\ar[dr]^{pr_2}& \\
&\mathbb{R}\times Q&&\mathbb{R}\times Q
}\] 
\caption{Geometric interpretation.}\label{diagramZZZ}
\end{figure}
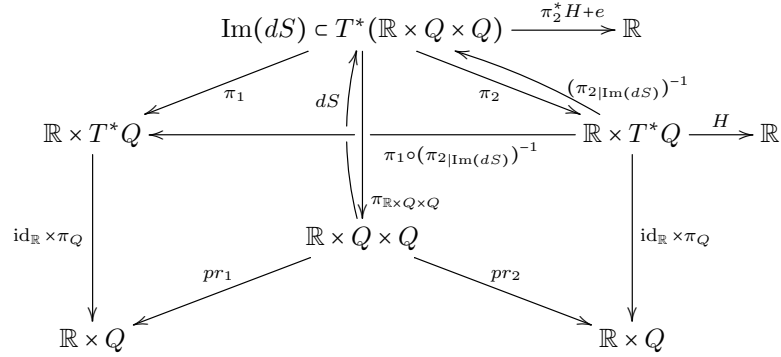

The transformation $ {\pi_1}_{|\im(dS)}\circ ({\pi_2}_{|\im(dS)})
^{-1}$ satisfies \[ \Big({\pi_1}_{|\im(dS)}\circ ({\pi_2}_{|\im(dS)})
^{-1}\Big)_* (\frac{\partial}{\partial t}+X_H)=\frac{\partial}{\partial
  t}+X_K,\] which is the geometric description of the transformation of
equations \eqref{hamilton} into \eqref{hamilton2}. 

\begin{remark}\label{Rfactor1}
We want to call the attention of the reader familiar with Lie groupoids
or discrete mechanics about the geometric structure needed to handle this
theory. In the above diagram, if one removes the $\mathbb{R}$ factor,
what we have is just a \textbf{pair groupoid} $Q\times Q$ and the
corresponding \textbf{cotangent groupoid} $T^*(Q\times Q)$, with base the dual of
its Lie algebroid $T^*Q$. We will show that the multiplication by the
$\mathbb{R}$ factor conserves the groupoid and cotangent groupoid
structures and for $\mathbb{R}\times Q\times Q$ the 
source and the target are exactly $pr_1$ and $pr_2$; furthermore, for the cotangent bundle
$\trqq$ the source and the target are just $\pi_1$ and $\pi_2$, introduced above.
\end{remark}

\begin{remark}
The inverse of the transformation induced above,  which happens to be
  $ {\pi_2}_{|\im(dS)}\circ ({\pi_1}_{|\im(dS)}) ^{-1}$, up to an initial
condition on $S$ at time $t=0$ gives the flow of the Hamiltonian
vector field $X_H$.
\end{remark}

\begin{remark} We used a time-independent Hamiltonian, but actually
  the theory is exactly the same for time-dependent systems.
\end{remark}

%%%%%%%%%%%%%%%%%%%%%%%%%%%%%%%%%%%%%%%%%%%%%%%%%%%%%%%%%The Lie-Poisson case
\subsubsection{The Lie--Poisson Case}
In this section we write in a geometric way the results about
Hamilton--Jacobi theory for Lie--Poisson systems, $(\mathfrak{g}^*,\
\Lambda, \ H)$, where 
\begin{enumerate}
\item $\mathfrak{g}^*$ is the dual of the Lie algebra $\mathfrak{g}$
  of a Lie group $G$.
\item $\Lambda$ is the canonical (-) Poisson structure on
  $\mathfrak{g}^*$, given by \[\{f,g\}(\mu)=-\mu([df(\mu),dg(\mu)])=\Lambda(df,dg)(\mu);\]
where $f, \ g: \ \mathfrak{g}^*\rightarrow\mathbb{R}$ and $\mu \in \mathfrak{g}^*$.
\item $H:\mathfrak{g}^*\rightarrow \mathbb{R}$ is a Hamiltonian function.
\end{enumerate}

These objects produce a dynamical system through the equation
\[\dot{\mu}=X_H(\mu)=\Lambda^{\sharp}(dH)(\mu),\] where $df(\Lambda^\sharp(dH))=\Lambda(df,dH)$, and $X_H$ is called the Hamiltonian vector
field. Detailed information about Lie--Poisson systems can be found, for instance, in
\cite{MarsdenRatiu}. Similar results to the ones  that we are going to
introduce now appeared for the first
time in \cite{GeMarsden}, related information can also be found in
\cite{Marsden}. Nonetheless, the approaches followed in those works
are very different and even in these situations we will give new
results. In order to
continue we need to define some mappings. The \emph{left} and \emph{right momentum} are the mappings $J_L:
T^*G\rightarrow \mathfrak{g}^*$ and $J_R:T^*G\rightarrow
\mathfrak{g}^*$ defined by $\langle J_L(\alpha_g),\xi\rangle=\langle\alpha_g,T_\ide R_g(\xi)\rangle$
and $\langle J_R(\alpha_g),\xi\rangle=\langle \alpha_g,T_\ide L_g(\xi)\rangle$. Let $S:\mathbb{R}\times G\rightarrow
\mathbb{R}$ be a function such that the following conditions hold:
\begin{enumerate}
\item {\it Hamilton--Jacobi equation:} $\displaystyle\frac{\partial
    S}{\partial t}(t,g)+H(J_R\circ dS_t)=k(t)$, where $S_t$ is defined
  by $S_t(g)=S(t,g)$.
\item {\it Non-degeneracy condition:} let $\xi_a$ be a basis of
  $\mathfrak{g}$. Then we assume that
  $\lvec{\xi}_a(\rvec{\xi}_b(S_t))$ is a regular matrix. Here
  $\lvec{\xi}_a $ and $\rvec{\xi}_b$ are the associated left-invariant
  and right-invariant vector fields respectively.
\end{enumerate}

With the function $S$ at hand we can define a transformation
analogous to \eqref{coordinatechange}. To make the exposition
easier, we introduce the following diagram analogous to diagram in
Figure \ref{diagramZZZ}:
\begin{figure}[H]
\[
\xymatrix{
&&\im(dS)\subset T^*(\mathbb{R}\times G)
\ar[dd]_>>>>>>>{\pi_{\mathbb{R}\times G}}\ar[rr]^-{(\id_{\mathbb{R}}\times
J_R)^*H+e\qquad }\ar[dl]^{\id_{\mathbb{R}}\times J_L}
\ar[dr]_{\id_{\mathbb{R}}\times J_R}&&\mathbb{R}\\
\mathbb{R}&\mathbb{R}\times \mathfrak{g}^*\ar[l]^{H}
&&\mathbb{R}\times \mathfrak{g}^*\\
&&\mathbb{R}\times G \ar@/_/[uu]_>>>>>>>{dS}&
}\]\caption{Lie-Poisson Setting.}\label{diagram2}
\end{figure}

\noindent where $(\id_{\mathbb{R}}\times J_L)(t,e,\alpha_g)=(t,J_L(\alpha_g))$ and,
in the same manner, we define $\id_{\mathbb{R}}$ $\times J_R$. The Hamilton--Jacobi
equation is equivalent to saying that $dS^*((\id_{\mathbb{R}}\times
J_R)^*(H+e))$ is equal to a time-dependent function $k(t)$.  It can also be checked
that the non-degeneracy condition implies that $\id_{\mathbb{R}}\times J_L$
restricted to $\im(dS)$ is a local diffeomorphism (it turns out that
this is equivalent to that $\id_{\mathbb{R}}\times J_R$ restricted to $\im(dS)$ to be a local
diffeomorphism). Now we can define the Poisson mapping
$(\id_{\mathbb{R}}\times J_L)\circ (\id_{\mathbb{R}}\times
{J_R}_{|\im(dS)})^{-1}:\mathbb{R}\times \mathfrak{g}^*\rightarrow
\mathbb{R}\times \mathfrak{g}^*$, henceforth denoted by $\hat{S}$. In
this way we complete the diagram in Figure \ref{diagram2} as follows:
\begin{figure}[H]\[
\xymatrix{
&&\im(dS)\subset T^*(\mathbb{R}\times G)
\ar[dd]_>>>>>{\pi_{\mathbb{R}\times G}}\ar[rr]^-{(\id_{\mathbb{R}}\times
J_R)^*H+e\qquad }\ar[dl]^{\id_{\mathbb{R}}\times J_L}
\ar[dr]_{\id_{\mathbb{R}}\times J_R}&&\mathbb{R}\\
\mathbb{R}&\mathbb{R}\times \mathfrak{g}^*\ar[l]^{H}
\ar@/^/[ru]^<<{(\id_{\mathbb{R}}\times
{J_L})_{|\im(dS)}^{-1}}&&\mathbb{R}\times \mathfrak{g}^* \hole
\ar@/_/[lu]_<<{\ \ (\id_{\mathbb{R}}\times
{J_R}_{|\im(dS)})^{-1}} \ar[ll]^<{\hat{S}\qquad \qquad \qquad}\\
&&\mathbb{R}\times G \ar@/_/[uu]_>>>>>>>{dS}|\hole&
}\]\caption{Geometric interpretation.}\label{diagrama222}
\end{figure}

The main fact is, again, that $\hat{S}:\mathbb{R}\times
\mathfrak{g}^*\rightarrow \mathbb{R}\times \mathfrak{g}^*$ verifies
\[
\hat{S}_*(\frac{\partial}{\partial t}+X_H)=\frac{\partial}{\partial t},
\]
which means that the Hamiltonian evolution is transformed into the trivial dynamics. With this method, we achieve a Poisson
transformation which completely integrates the dynamics.

\begin{remark}\label{Rfactor2}
If we think about the underlying geometric structure,
forgetting about the $\mathbb{R}$ factor, we have a \textbf{Lie group}, $G$,
and its \textbf{cotangent groupoid} $T^* G$. The source and
target of this cotangent groupoid are known to be the mappings $J_L$
and $J_R$. We will see that $T^*(\mathbb{R}\times G)$ is again a
(symplectic) groupoid and that $\id_{\mathbb{R}}\times J_L$ and
$\id_{\mathbb{R}}\times J_R$ are  its source and target.
\end{remark}
\begin{remark}
As in the standard case, the Hamiltonian can be time-dependent and
the same results hold.
\end{remark}

With these two examples at hand it seems clear that the geometry of
the theory can be described using \textbf{cotangent Lie
  groupoids}. The idea of using Lie groupoids to describe the
Hamilton--Jacobi equation appeared for the first time in \cite{GeHJ}, as
far as we know,
and there it is pointed out that A. Weinstein was the first to notice
that there may be a connection between generating functions and
symplectic groupoids. As opposed to that general approach, we focus on giving a complete
picture in the case of cotangent Lie groupoids. Cotangent Lie
groupoids seem to be general enough to include all the interesting
cases of mechanical systems, but at the same time they have
interesting features that make them very useful in practical
problems, which is our final goal. For instance, Darboux
coordinates in the cotangent groupoid are always available, but finding them 
is not such an easy task in the case of a general symplectic groupoid. Furthermore, cotangent
groupoids provide the natural framework to relate continuous and
discrete Hamiltonian and Lagrangian dynamics, as was pointed by
A. Weinstein in \cite{Weinstein}, following \cite{MoserVeselov}. This
viewpoint was
exploited by some of the authors in \cite{MMM,MMMIII, MMS}. A particular
case of cotangent groupoid, the cotangent groupoid of an action Lie
algebroid, was already suggested
as the correct setting for a Hamilton--Jacobi theory in
\cite{ScovelWeinstein}, although no more progress has been made in this direction so far. Moreover, in \cite{ScovelWeinstein} the authors
develop a finite-dimensional Poisson truncation, {\it i.e.} a finite-dimensional Poisson model of the Poisson--Vlasov
equation, an infinite-dimensional Poisson system. It turns out that
the truncation happens to be the dual algebroid of an action
groupoid, which is a particular case of the theory that we present
here. The importance of the {\it ``oid theory''} was already clear to
the authors
\vspace{0.2cm}
\begin{quote}
{\it ``Our first derivation used the theory of Lie groupoids
and Lie algebroids, and the Poisson structures on the duals of Lie
algebroids. We were then able to eliminate the ``oid'' theory in favor of
more well-known ideas on Poisson reduction. [...] The groupoid aspect of the theory also provides natural Poisson
  maps, useful in the application of Ruth type
integration techniques, which do not seem easily derivable from the general theory of Poisson
reduction''.}
\end{quote}
\vspace{0.1cm}
\hfill {\it -C. Scovel and A. Weinstein \cite{ScovelWeinstein}, p. $683$}.
\vspace{0.2cm}

Here we use that connection, symplectic groupoid--Poisson manifolds, to
develop a general theory for Hamiltonian systems on the dual of an integrable Lie algebroid, a framework large enough to study all
interesting Poisson Hamiltonian systems in classical mechanics. Using some well-known facts about
symplectic groupoids, like the fact that the cotangent bundle of a Lie
groupoid is a symplectic groupoid with base the dual bundle of the associated
Lie algebroid, then, a construction similar to the one outlined in the
two examples above leads to the desired Hamiton--Jacobi theory. This
theory allows us to seek for transformations which integrate
Hamilton's equations in the same way we described above.

We remak that, very recently, in \cite{GrilloPadron} the
authors develop a geometric Hamilton-Jacobi theory for dynamical
systems on the total space of a fibration. In the particular case
when the dynamical system is Hamiltonian with respect to an
(arbitrary) Poisson structure and the fibration is isotropic, they
discuss some applications of the theory to the integration by
quadratures of the system. It is clear that our approach in this paper
is different. In fact, we are not concerned with the integration by
quadratures of the system and, in addition, we focus on Hamiltonian
systems with respect to linear Poisson structures on vector
bundles. This fact, as we mentioned before, allows us to use the
theory of Lie algebroids and groupoids and to discuss the interesting
examples for Classical Mechanics, namely, Hamiltonian systems on
cotangent bundles, on the dual space of Lie algebras and the dual
bundles of action and Atiyah Lie algebroids. We also want to stress that besides the exact integration
of the Hamilton's equations our motivation comes from the applications
to numerical methods, which aims to develop, among other things, the
Ruth type integrators pointed in the last quote, taken from
\cite{ScovelWeinstein}. Our examples go into this direction,
although the analytical and dynamical applications should be exploited
as well (see the last section). It is well-known that the theory of generating
functions gives a family of symplectic numerical methods,
\cite{Kang,Hairer,McLachlanScoveI}. There are numerous
examples illustrating the superior preservation of phase-space
structures and qualitative dynamics by symplectic integrators. These
methods were extended to Lie groups in \cite{GeMarsden}, but our
approach is general enough to provide a general setting to develop new numerical
methods on the dual bundle of Lie algebroids. We also stress the
importance of this task, extending the symplectic integrators to the
Poisson world, by a Peter Lax's quote that can be found in \cite{Kang}.
\vspace{0.2cm}
\begin{quote}
{\it
``In the late 1980s Feng Kang proposed
and developed so-called symplectic
algorithms for solving equations in
Hamiltonian form. Combining theoretical
analysis and computer experimentation,
he showed that such methods, over
long times, are much superior to standard
methods. At the time of his death, he was
at work on extensions of this idea to other
structures.''}
\end{quote}

\vspace{0.1cm}
\hfill {\it - P. Lax}

\vspace{0.2cm}
Our approach also differs from the previous ones, \cite{BenzelGeScovel,ChannellScovelII,GeEquivariant,GeMarsden,PoissonIntegrators,McLachlanScoveII}, in that we focus on the
Lagrangian submanifolds instead of the generating functions
themselves. We understand  generating functions, as it has been done for many decades, as a device to describe
Lagrangian submanifolds which are ``horizontal'' regarding a certain
projection. That projection will only be defined locally in many
cases, but that is enough for our applications. Let us clarify a little bit
the situation. The most basic instance
of this setting is the cotangent bundle $(T^*Q,\ \omega_Q)$ endowed
with the canonical symplectic structure and the natural projection
$\pi_Q:T^*Q\rightarrow Q$. Then, we say that a Lagrangian submanifold, say $\lag$, 
is {\it horizontal} for $\pi_Q$ if $(\pi_Q)_{|\lag}:\lag\rightarrow Q$ is a
diffeomorphism, equivalently, if there exists a closed $1$-form $\gamma$ such
that $\im(\gamma)=\lag$. This fact is a straightforward application of
the implicit function theorem. By the Poincar\'e Lemma, at least
locally, there exists $S$ defined on an open set of $Q$ such that
$dS=\gamma$. In this sense, all the information of the Lagrangian
submanifold $\lag$ can be encoded in a function, much easier to
handle. Conditions on $\lag$ become PDEs in the unknown function
$S$. This simple fact allowed us to come with analogues of the  {\it non-free canonical transformations},
following the notation in \cite{Arnoldmmcm}. They
permit, in a very natural way, to generate the identity transformation
as generating functions. This was an open problem, unsolved as far as
we know:
\vspace*{0.2cm}
\begin{quote}
{\it ``Is there a generating function for Lie--Poisson maps which generates the identity map via the coadjoint action of the identity group element?''}
\end{quote}
\vspace{0.1cm}
\hfill {\it -R.I. McLachlan and C. Scovel. \cite{McLachlanScoveII}, p. $157$.}

\vspace*{0.2cm}

This result is quite useful, as the classical type I generating
functions are not enough even in the classical case for may practical
applications. In that regard, we would like to stress
  that the importance of our construction relies not on parametrizing
  the Lagrangian submanifolds of identities, but on giving a
  parametrization of all the Lagrangian submanifolds that are close to
the identity. Those Lagrangian submanifolds induce the Poisson
transformations which are near the identity {\i.e.} the Hamiltonian
flow for small enough time (see Section \ref{non-free}). In addition, the results in Section $3$ and $4$ are the
keystones to develop the geometric numerical methods in Section $5$. Furthermore, we give here results that show how Hamilton's
equations change under a canonical transformation induced by a
generating function, generalizing the equation
\eqref{changehamiltonian} which is needed to develop Ruth type
integrators. We also show local existence of solution and generalize
the classical result which claims that the action functional ($\int L\, dt$)
is a solution of the Hamilton--Jacobi equation. Our approach is new
even when applied to the known situations.

%%%%%%%%%%%%%%%%%%%%%%%%%%%%%%%%%%%%%%%%%%%%%%%%%%%%%%%%%%%%%%%%%%%%%%%%%%%%%SUMMARY

\subsection{Summary}

After the previous section, where we illustrated the motivation that
led us to our results, we
proceed in the following way.

Section $2$ is devoted to the definition of symplectic groupoids. We
include two appendices about Lie groupoids and Lie algebroids with all
the definitions and properties, so the reader
unfamiliar with them should find there the information needed to understand this paper. We
summarize some of the most important properties of symplectic
groupoids because we are going to use them in Section 3 to develop our
Hamilton--Jacobi theory. We provide
several examples of symplectic groupoids and describe carefully the
ones we are concerned with, cotangent groupoids.

In Section $3$ we develop the main results of this paper. The main goal
is to provide a theory of generating functions and a Hamilton--Jacobi theory that allow us to look for
transformations which integrate the Hamiltonian equations of
motion, or at least, approximate them through Poisson automorphisms. This is done using the symplectic groupoid structure of the
cotangent bundle of a groupoid, in particular, its dual pair structure. We also summarize some of the
results in \cite{GeMarsden} related to these generating functions.

In Section $4$ we give a couple of results about local existence of
solutions of the Hamilton--Jacobi equation introduced in the previous section. Since we are dealing with a PDE, it is our duty to show, at
least, local
existence of solutions. The first result generalizes the classical Jacobi's solution
of the Hamilton--Jacobi equation given by the action. The second one is the
method of characteristics, well-known in the literature.

Examples are given in Section $5$. Of course, our theory recovers
the classical situation via the pair groupoid and Ge--Marsden's
framework via the Lie group
case. Even in those cases, our approach clarifies the geometry and our
results are stronger. We present here a general procedure to apply in order to
develop numerical methods. The methods are very general and apply to any Poisson manifold under consideration, but they can be drastically improved when applied to concrete situations. Moreover, all the needed tools to design new numerical methods are also presented in this paper (see the last section). The
examples given in the paper are the rigid body, the heavy top and Elroy's beanie. 

%%%%%%%%%%%%%%%%%%%%%%%%%%%%%%%%%%%%%%%%%%%%%%%%%%%%%%%%%%
%% SYMPLECTIC GROUPOIDS
%%%%%%%%%%%%%%%%%%%%%%%%%%%%%%%%%%%%%%%%%%%%%%%%%%%%%%%%%%

\section{Symplectic Groupoids}

Along this section we introduce  the geometric
framework that allows us to ``geometrize'' the Hamilton--Jacobi equation
in the next section. The basic definitions and properties are given and
several examples are provided. The usual definitions and notation
about Lie groupoids and Lie algebroids are given in the appendices
\ref{liealgebroids} and \ref{liegroupoids}. The basic references for
this section are \cite{CannasDaSilvaWeinstein,CosteDazordWeinstein,lecturesontheintegrabilityofliebrackets}. We want to warn the reader, especially the one non experienced with groupoids, that in the end our constructions rely on the geometric structures of cotangent groupoids and not in  the algebraic ones. That is, in order to follow our results it is enough to understand how to obtain the source and target of the cotangent groupoid and the results of Section \ref{properties}. The main point is that cotangent groupoids provide a symplectic realization of the involved Poisson structures in a way that allows to control their Poisson automorphisms. To what extent the algebraic structures, multiplication and inversion, play a role here is still unknown to us.

\subsection{Definition of Symplectic Groupoids}
Let us now introduce the notion of a symplectic groupoid. For the definitions of groupoids, Lie groupoids or Lie algebroids and the corresponding notation see the appendices. 

\begin{definition}

A \emph{symplectic groupoid} is a Lie groupoid $G\rightrightarrows M$
equipped with a symplectic form $\omega$ on $G$ such that the graph of
the multiplication $m\ \colon \  G_2\to G$, that is, the set $\left\{ (g,h,gh)\,|\, (g,h)\in G_2 \right\}$, is a Lagrangian submanifold of $G\times G\times (-G)$ with the product symplectic form, where $-G$ denotes $G$ endowed with the symplectic form $-\omega$.
\end{definition}

\begin{remark}
It can be easily check that $(G, \omega)$ is a symplectic groupoid
in the above sense if and only if the $2$-form $\omega$ is multiplicative. We
say that the form $\omega$ is \emph{multiplicative} iff 
\[
m^*\omega=\pi_1^*\omega+\pi_2^*\omega
\]
where $\pi_i:G_2\rightarrow G$, $i=1,\ 2$ are the projections over the
first and second factor.
\end{remark}

\subsection{Cotangent Groupoids}
This section constitutes a source of examples of symplectic
groupoids. Since
cotangent groupoids play an important role in our theory they deserve
a whole subsection. Assume that $G\rightrightarrows M$ is a Lie groupoid with source and
target $\alpha$ and $\beta$ respectively and identity section,
inversion map and multiplication $\epsilon$, $\iota$ and $m$, then there is an induced Lie
groupoid structure $T^*G\rightrightarrows A^*G$ which we define below
(see Appendix \ref{liegroupoids} for notation). The cotangent bundle $T^*G$ turns out to be a symplectic groupoid with the canonical symplectic form $\omega_G$ (see 
\cite{CosteDazordWeinstein,lecturesontheintegrabilityofliebrackets}). Let us define the composition law on $T^*G$, which will be written as
$\oplus_{T^*G}$, and let $\widetilde \alpha$, $\widetilde \beta$,
$\widetilde \epsilon$ and $\widetilde \iota$ stand for the source, target,
identity section and inversion map of the groupoid structure on $T^*G$
defined below.  Each one of these maps will cover the corresponding structural map of $G$.

\begin{enumerate}
\item The \emph{source} is
 defined in a way that the following diagram is commutative
\begin{equation}\label{comm1}
\xymatrix{
T^*G \ar[r]^{\widetilde\alpha}\ar[d]_{\pi_G}\ar@{}[dr]|\circlearrowright
& A^*G\ar[d]^\tau\\
G\ar[r]^\alpha & M.
}
\end{equation}
Assume that $g\in G^x_y=\alpha^{-1}(x)\cap \beta^{-1}(y)$ and let $X\in A_xG$. So, by the definition of Lie algebroid
associated to a Lie groupoid, $X$ is a tangent vector at $\epsilon(x)$ that is tangent
to $G^x=\alpha^{-1}(x)$.  Then, since $\iota\ \colon\ G^x\rightarrow G_x$
(where $G_x=\beta^{-1}(x)$) is a
diffeomorphism between the $\alpha$-fibers and $\beta$-fibers, we
conclude that $-(T_{\epsilon(x)}\iota)(X)\in T_{\epsilon(x)}G_x$, \emph{i.e.}, it is
tangent to the $\beta$-fiber. Now,
recalling that right multiplication is a bijection $r_g:G_x\rightarrow
G_y$, with $g \in G^x_y$, we have that $T_{\epsilon(x)}r_g(-T_{\epsilon(x)}\iota(X))\in
T_gG$ and we can finally define
\[
\widetilde\alpha(\varpi_g)(X)=\varpi_g(-T_{\epsilon(x)}(r_g\circ\iota) (X)), \; \; \mbox{ for } \varpi_g \in T_g^*G.
\]
In particular, if we are dealing with a section $X\in\Gamma(\tau)$, the previous
construction leads us to right-invariant vector fields:
\begin{equation}\label{alpha}\widetilde\alpha(\varpi_g)(X)=\varpi_g(\rvec{X}(g)),\text{ for }X\in\Gamma(\tau),\end{equation}
where $\tau\ \colon AG \to M$ is the Lie algebroid associated to $G$ (see equation
\eqref{rinv} in Appendix \ref{liegroupoids}). 

Notice that $\widetilde \alpha\ \colon\ T^*G\to A^*G$ is a surjective submersion.

\item In an analogous way, the \emph{target}, $\widetilde \beta\ \colon\ T^*G\to A^*G$, is defined so that the following diagram commutes
\begin{equation}\label{comm2}
\xymatrix{
T^*G \ar[r]^{\widetilde\beta}\ar[d]_{\pi_G}\ar@{}[dr]|\circlearrowright
& A^*G\ar[d]^\tau\\
G\ar[r]^\beta & M.
}
\end{equation}
Now, given $\varpi_g \in T_g^*G$ and assuming that $g\in G^x_y$, since left
multiplication is defined on the $\alpha$-fibers, $l_g\ \colon G^y\rightarrow
G^x$, then $T_{\epsilon(x)}l_g(X)\in T_{g}G$ and we can define
\[
\widetilde\beta(\varpi_g)(X)=\varpi_g(T_{\epsilon(x)}l_g(X)).
\]
In other words, 
\begin{equation}\label{beta}\widetilde\beta(\varpi_g)(X)=\varpi_g(\lvec{X}(g)),\text{ for
}X\in\Gamma(\tau)\end{equation}from equation \eqref{linv}, Appendix \ref{liegroupoids}.

\item The \emph{identity map}, $\widetilde \epsilon : A^*G\to
  T^*G$, is defined so that
\[
\xymatrix{
A^*G \ar[r]^{\widetilde\epsilon}\ar[d]_{\tau}\ar@{}[dr]|\circlearrowright
& T^*G\ar[d]^{\pi_G}\\
M\ar[r]^\epsilon & G
}
\]
commutes. Take $v\in T_{\epsilon(x)}G$. We can obtain an element of $A_xG\subset T_{\epsilon(x)}G$ by computing $v-T(\epsilon\circ\alpha)(v)$. This tangent vector is indeed tangent to an $\alpha$-fiber since $T\alpha$ of it is $0$, using an argument analogous to that in the definition of $\widetilde \alpha$. Then, for $\mu_x\in A^*_xG$, we define
\[ \widetilde \epsilon(\mu_x)(v)=\mu_x(v-T(\epsilon\circ\alpha)(v)). \]
\item The \emph{inversion map}, $\widetilde \iota\colon\; T^*G\to T^*G$, is defined as a mapping from each $T^*_gG$ to $T^*_{g^{-1}}G$. If $X\in T_{g^{-1}}G$, then $T\iota(X)\in T_gG$, so we can define for $\varpi_g\in T^*_gG$
\[ \widetilde \iota (\varpi_g)(X)=-\varpi_g(T\iota(X)).\]

\item The \emph{groupoid operation} $\oplus_{T^*G}$ is defined for
  these pairs $(\varpi_g,\nu_h)\in T^*G\times T^*G$ satisfying the
  composability condition $\widetilde \beta(\varpi_g)=\widetilde
  \alpha(\nu_h)$, which in particular implies that $(g,h)\in G_2$
  using diagrams \eqref{comm1} and \eqref{comm2}. This condition can be rewritten as
\[ \varpi_g\circ Tl_g=-\nu_h\circ Tr_h\circ T\iota\quad\text{on }A_{\beta(g)}G\subset T_{\epsilon(\beta(g))}G. \]
We define 
\[
\left( \varpi_g\oplus_{T^*G}\nu_h \right)(T_{(g,h)}m(X_g,Y_h))=\varpi_g(X_g)+\nu_h(Y_h),
\]
for $(X_g, Y_h)\in T_{(g,h)}G_2$. An explicit expression for $\oplus_{T^*G}$ using local bisections in the Lie groupoid $G$ can be found in \cite{MMS}.
\end{enumerate}

\begin{remark}
Note from equation \eqref{alpha} (resp.\ \eqref{beta}) that the definition of
$\widetilde\alpha$ (resp.
 $\widetilde\beta$) is just given by
``translation'' via right-invariant vector fields (resp.\ left-invariant).
\end{remark}
\begin{remark}
Another interesting property is that the application $\pi_G\ \colon\  T^*G\to G$ is a
Lie groupoid morphism over the vector bundle projection $\tau_{A^*G}\colon A^*G\to M$.
\end{remark}

\begin{remark}
When the Lie groupoid $G$ is a Lie group, the Lie groupoid $T^*G$ is not in general a Lie group. The base $A^*G$ is identified with the dual of the Lie algebra $\mathfrak{g}^*$, and we have $\widetilde\alpha(\varpi_g)(\xi)=\varpi_g\left( Tr_g \xi\right)$ and $\widetilde\beta(\varpi_g)(\xi)=\varpi_g\left( Tl_g \xi\right)$, where $\varpi_g\in T^*_gG$ and $\xi\in {\mathfrak g}$.
\end{remark}
%%%%%%%%%%%%%%%%%%%%%%%%%%%%%%%%%%%%%%%%%%%%%%%%%%%%%%%%%%%%%%%%%%%%%%%%%%%%%%EXAMPLES

\subsection{Example: Cotangent Bundle of the Gauge Groupoid}\label{cotangentatiyah}
We present now the cotangent groupoid of a gauge groupoid as an
illustration of the previous constructions. The gauge groupoid is
described in Section \ref{atiyah}. We assume here that the principal $G$-bundle $\pi:P\rightarrow M$ is trivial, so $P=G\times M$. Then, it is easy to see that $(P\times P)/G\equiv M\times M\times G$, with the identification $[(m^1,g^1,m^2,g^2)]\rightarrow (m^1,m^2,(g^1)^{-1}g^2)$. Thus, $T^*\left((P\times P)/G\right)\equiv T^*(M\times M\times G)$, and given $(\lambda_{m^1},\mu_{m^2},\nu_{g})\in T^*(M\times M\times G)$, we have
\begin{enumerate}
 \item The {\it source}, $\tilde{\alpha}$, is the map $\tilde{\alpha}(\lambda_{m^1},\mu_{m^2},\nu_{g})=(-\lambda_{m^1},J_R(\nu_{g}))$.
\item The {\it target}, $\tilde{\beta}$, is the map $\tilde{\beta}(\lambda_{m^1},\mu_{m^2},\nu_{g})=(\mu_{m^2},J_L(\nu_{g}))$.
\item The  {\it identity map} is $\tilde\epsilon(\lambda_{m}, \nu_g)=(-\lambda_{m},\lambda_{m}, \nu_g)$.
\item The {\it inversion map} is $\tilde\iota(\lambda_{m^1},\mu_{m^2},\nu_{g})=(-\mu_{m^2},\lambda_{m^1},-\nu_{g}\circ T_{g^{-1}}\iota)$, where $\iota$ is the inversion in the Lie group $G$.
\item  The {\it multiplication} is $\tilde{m}\left((\lambda_{m^1},\mu_{m^2},\nu_{g}),  (-\mu_{m^2},\mu_{m^3},\nu_{h})\right)=(\lambda_{m^1},\mu_{m^3},\nu_{g}\circ T_{gh}r_{h^{-1}}=\nu_{h}\circ T_{gh}l_{g^{-1}})$.
\end{enumerate}

\begin{remark} Easier-to-handle expressions can be obtained by
  trivializing $T^*G\equiv G\times\mathfrak{g}^*$ in the product $T^*(M\times M\times G)\equiv T^*M\times T^*M\times T^*G$.
\end{remark}

\subsection{Properties}\label{properties}
The following theorem will be crucial in the next sections (see
\cite{CosteDazordWeinstein,lecturesontheintegrabilityofliebrackets} for a proof, as well as the references therein):

\begin{theorem}\label{theorem-sg}
Let $G\rightrightarrows M$ be a symplectic groupoid, with symplectic 2-form $\omega$. We have the following properties:
\begin{enumerate}
\item For any point $g\in G_y^x$, the subspaces $T_g G_y$ and $T_gG^x$ of the symplectic vector space $(T_gG, \omega_g)$ are mutually symplectic orthogonal. That is,
\begin{equation}\label{dualpair}
T_g G_y=(T_gG^x)^{\perp}.
\end{equation}
\item The submanifold $\epsilon (M)$ is a Lagrangian submanifold of the symplectic manifold $(G, \omega)$.
\item The inversion map $\iota\ \colon\  G\to G$ is an anti-symplectomorphism of $(G, \omega)$, that is, $\iota^*\omega=-\omega$.
\item There exists a unique Poisson structure $\Pi$ on $M$ for which $\beta\ \colon\  G\to M$ is a Poisson map, and $\alpha\ \colon\  G\to M$ is an anti-Poisson map (that is, $\alpha$ is a Poisson map when $M$ is equipped with the Poisson structure $-\Pi$).
\end{enumerate}
\end{theorem}

\begin{remark} The theorem above states that the $\alpha$-fibers and
  $\beta$-fibers are symplectically orthogonal. A pair of fibrations
  satisfying that property are called a \textbf{\emph{dual pair}}. This dual
  pair property will be the keystone of our construction.
\end{remark}

\begin{remark} When dealing with the symplectic groupoid $T^*G$, where
  $G$ is a groupoid, the Poisson structure on $A^*G$ is the
  (linear) Poisson structure of the dual of a Lie algebroid (see
  Appendix \ref{poisson}).
\end{remark}

The next theorem is the core of our Hamilton--Jacobi theory. It
basically says that Lagrangian bisections of a symplectic groupoid
induce Poisson transformations in the base.

\begin{theorem}\label{PoissonIsomorphisms}[see \cite{CosteDazordWeinstein}]\label{CDW} Let $G$ be a symplectic groupoid with
  source and target $\alpha$ and $\beta$ respectively. Let $\lag$ be a
  Lagrangian submanifold of $G$ such that $\alpha_{|\lag}$ is a (local)
  diffeomorphism. Then
\begin{enumerate}
\item $\beta_{|\lag}:\lag\rightarrow M$ is a (local) diffeomorphism as well.
\item The mapping
  $\hat{\lag}=\alpha\circ(\beta_{|\lag})^{-1}:M\rightarrow M$ and its
  inverse, $\beta\circ (\alpha_{|\lag})^{-1}$,
  are (local) Poisson isomorphisms.
\end{enumerate}
\end{theorem}

\begin{remark} We wrote in the above theorem the mappings in the
  global diffeomorphism case. Of course, when the diffeomorphism is
  local, these maps are restricted to the corresponding open sets.
\end{remark}

The Lagrangian submanifolds used above are usually
  called Lagrangian bisections. They are bisections that happen to be
  Lagrangian submanifolds at the same time. See the appendices
  included in this paper for a definition of bisection.

%%%%%%%%%%%%%%%%%%%%%%%%%%%%%%%%%%%%%%%%%%%%%%%%%%%%%%%%%%
%%Generating functions for Poisson manifolds
%%%%%%%%%%%%%%%%%%%%%%%%%%%%%%%%%%%%%%%%%%%%%%%%%%%%%%%%%%

\section{Generating functions and Hamilton-Jacobi theory for Poisson Manifolds}\label{Generating functions for a class of Poisson Manifolds}
\subsection{The Geometric Setting}
In this section we develop a theory of generating functions and a Hamilton--Jacobi theory using the geometric
structures introduced along the previous sections. We encourage the
reader, especially the reader not familiar with Lie groupoids and Lie
algebroids, to keep in mind the two instances of the Hamilton-Jacobi theory sketched in the introduction
of this paper (the classical and the Lie-Poisson cases). Our construction mimics what happens in those cases.
Before going to our theory of generating function  and Hamilton--Jacobi theory, we need to introduce one more
construction. It basically says that the $\mathbb{R}$ factor we
mentioned in Remark \ref{Rfactor1} and Remark \ref{Rfactor2} does not
modify the groupoid structure. This theorem was already sketched in \cite{GeHJ}.

\begin{proposition} Let $G$ be a Lie groupoid with source, target,
  identity map, inversion and multiplication $\alpha, \ \beta, \
  \epsilon, \ \iota$ and $m$. Then the manifold
  $\rgroupoid$ is a Lie groupoid with base $\mathbb{R}\times M$ and with structure
  mappings
\begin{equation}\label{rgroupoid}
\begin{array}{rccl}
\ralpha:& \rgroupoid&\longrightarrow & \rbase \\ 
& (t,g) &\to &\ralpha(t,g)=(t,\alpha(g)) \\ \noalign{\bigskip}

\rbeta:& \rgroupoid&\longrightarrow & \rbase \\ 
& (t,g) &\to &\rbeta(t,g)=(t,\beta(g)) \\ \noalign{\bigskip}

\repsilon:& \rbase&\longrightarrow & \rgroupoid \\ 
& (t,x) &\to &\repsilon(t,x)=(t,\epsilon(x)) \\ \noalign{\bigskip}

\riota:& \rgroupoid&\longrightarrow & \rgroupoid \\ 
& (t,g) &\to &\riota(t,g)=(t,\iota(g)) \\ \noalign{\bigskip}

\rmultiplication: &({\rgroupoid})_2 &\longrightarrow &\rgroupoid\\
&((t,g),(t,h))&\to&\rmultiplication ((t,g),(t,h))=(t,m(g,h)).
\end{array}
\end{equation}
Notice that the condition $\rbeta(t_1,g)=\ralpha(t_2,h)$ is
equivalent to $t_1=t_2$ and $\beta(g)=\alpha(h)$.
\end{proposition}

The proof is obvious and we omit it. If $AG$ is the Lie
algebroid associated to $G$, then the Lie algebroid associated to
$\rgroupoid$ is just $\rtau:\ra\rightarrow \rbase$, where
$\rtau(t,X)=(t,\tau(X))$ and the addition is just
$(t,X)+(t,Y)=(t,X+Y)$ for $t\in \mathbb{R}$ and $X, \ Y\in AG$. The
anchor and the Lie bracket on sections have analogous
expressions. Actually one can think of $\mathbb{R}\times G$ and $\ra$ as
``time-dependent'' structures, although we regard them just as groupoids. The dual of this Lie algebroid is
simply $A^*(\rgroupoid)=\mathbb{R}\times A^*G$. The
cotangent bundle of $\rgroupoid$ is a (symplectic) Lie groupoid with base
$\mathbb{R}\times A^*G$ and structure maps $\widetilde\ralpha, \
\widetilde\rbeta, \
\widetilde\repsilon, \ \widetilde\riota$ and
$\widetilde\rmultiplication$. The next diagram summarizes the
situation, note the commutativity:
\begin{figure}[H]\[
\xymatrix{
&T^*(\mathbb{R}\times G)\ar[dl]_{\tilde{\ralpha}}
\ar[dr]^{\tilde{\rbeta}}\ar[dd]^{\pi_{G}}&\\
\mathbb{R}\times A^*G\ar[dd]^{\rtau}&&\mathbb{R}\times A^*G \ar[dd]^{\rtau}\\
&\mathbb{R}\times G \ar@{}[ul]|\circlearrowright
\ar[dl]_{\ralpha}\ar[dr]^{\rbeta} \ar@{}[ur]|\circlearrowright& \\
\mathbb{R}\times M&&\mathbb{R}\times M.
}\]\caption{Geometric interpretation.} \label{diagram3}
\end{figure}

The reader should note that the diagram above is just the
generalization of diagrams in  Figures \ref{diagramZZZ} and \ref{diagrama222}. When
$G$ is the pair groupoid or a Lie group we recover these two
settings. A point of $T^*(\rgroupoid)$ is given by $(t,e,\alpha_g)$,
representing the cotangent vector $e\,dt+\alpha_g$, that is, $e$
is the time conjugate momentum. With this notation, the maps
$\widetilde\ralpha$ and $\widetilde\rbeta$ read as follows
\begin{equation}\label{diagram4}
\begin{array}{rccl}
\widetilde\ralpha:& T^*(\rgroupoid)&\longrightarrow & \rad \\ 
& (t,e,\varpi_g) &\to &\widetilde\ralpha(t,e,\varpi_g)=(t,\widetilde\alpha(\varpi_g)), \\ \noalign{\bigskip}

\widetilde\rbeta:& T^*(\rgroupoid)&\longrightarrow & \rad \\ 
& (t,e,\varpi_g) &\to &\widetilde\rbeta(t,e,\varpi_g)=(t,\widetilde\beta(\varpi_g)).
\end{array}
\end{equation}

\begin{remark} It is easy to see that the Poisson structure on
  $A^*(\rgroupoid)=\rad$ is given by the product of the $0$ Poisson
  structure on $\mathbb{R}$ and the natural Poisson structure on
  $A^*G$, say $\Pi$. So the Poisson structure on $\rad$ is just the one
  described at the beginning of the section, recall that we refer to
  this structure as $\pir$. Recall that by Theorem \ref{theorem-sg}
  we have a dual pair structure, considering the natural symplectic
  structure $\omega_{\rgroupoid}$ on $T^*(\rgroupoid)$ and the
  aforementioned Poisson structure $\pir$ on $\rad$.
\end{remark}

\subsection{Motivation}

Let us start with a (possibly time-dependent) Hamiltonian system on the dual
of an integrable Lie algebroid ($\mathbb{R}\times A^*G, \ \pir, \
H$). We clarify now the situation.
\begin{enumerate}
\item $\mathbb{R}\times A^*G$ is the \emph{extended phase space}, that is,
  the product of the phase space $A^*G$ and the time, $\mathbb{R}$. We
  are assuming that our (possibly time-dependent) Hamiltonian system
  evolves on the dual of a Lie algebroid. This framework is enough to cover
  all interesting cases:
\begin{enumerate}
\item cotangent bundles;
\item duals of Lie algebras (Lie--Poisson systems);
\item duals of action Lie algebroids (semi-direct products: heavy top,
  truncated Vlasov--Poisson equation...);
\item reduced systems (dual of the Atiyah algebroid),
\end{enumerate}
among others. Notice that this is the space where
A. Weinstein \cite{Weinstein} points out that the dynamics of Hamiltonian systems
should take place.

\item $\pir$ is the \emph{Poisson structure} in $\mathbb{R}\times A^*G$ given by the product of the $0$
  Poisson structure on $\mathbb{R}$ and the natural linear Poisson
  structure $\Pi$ on $A^*G$ described in Section \ref{poisson} of
  Appendix \ref{liealgebroids}. So $\pir=0\times \Pi$.

\item $H:\mathbb{R}\times A^*G\rightarrow \mathbb{R}$ is the
  \emph{Hamiltonian function}. Once that function is given, the dynamics is
  produced by the suspension of the Hamiltonian vector field $X_H=\pir^{\sharp}(dH)$, i.e.,
  $\frac{\partial}{\partial t}+X_H$, 
  which is a vector field on $\mathbb{R}\times A^*G$.
\end{enumerate}

Assume that we have the same Hamiltonian system
as in the previous section, ($\mathbb{R}\times A^*G, \ \pir, \
H$). Since any Lagrangian bisection $\lag$ of the symplectic
groupoid $T^*(\rgroupoid)$ produces a Poisson isomorphism (defined in Theorem \ref{PoissonIsomorphisms}) on the base,
which we will denote by $\hat{\lag}:\rad\rightarrow\rad$, then the
equations of motion $\frac{\partial}{\partial t}+X_H$ transform to
another Hamiltonian equations (see Theorem
\ref{hamiltoniantransformation}). In this section we give an explicit
formula for the new Hamiltonian, and therefore for the transformed Hamilton's
equations. This result is useful to compute the local error of the numerical methods developed in the next sections. After that, we look for explicit coordinates where we can
compute the Lagrangian submanifolds of interest. We are mainly
interested in computing the perturbations, nearby Lagrangian submanifolds, of a very important
Lagrangian submanifold, the submanifold of identities,
$\tilde\epsilon(A^*G)$, of the cotangent groupoid under consideration, $T^*G$. These Lagrangian submanifolds happen to be
generally non-horizontal, even in the classical or Lie--Poisson settings,
for the canonical projection $\pi_G$. We propose a local solution that
allows us to develop our Poisson integrators under the situations
described in the previous sections.

%%%%%%%%%%%%%%%%%%%%%%%%%%%%%%%%%%%%%%%%%%%%%%%%%%%%%%%%%%%%%%%%%%%
\subsubsection{ The Classical Case}
The reader familiar with the classical Hamilton--Jacobi theory should
 notice that we are we are going to describe in what follows the geometric counterpart of the
 usually stated fact that Hamilton's equations for the Hamiltonian $H$
in the coordinates $(t,q^i,p_i)$, say
\begin{equation}\label{hamilton3}
\begin{array}{rl}
\displaystyle\frac{d q^i}{dt}(t)=&\displaystyle\frac{\partial H}{\partial
  p_i}(q^i(t),p_i(t)),\\ \noalign{\medskip}
\displaystyle\frac{dp_i}{dt}(t)=&-\displaystyle\frac{\partial H}{\partial q^i}(q^i(t),p_i(t)),
\end{array}
\end{equation}
transform via a generating function
 $S$ (see equation \eqref{coordinatechange}) to another system of Hamilton's equations in coordinates $(t,x^i,y_i)$ for
 a new Hamiltonian function $K$, given by the expression
\begin{equation}\label{relation}
\displaystyle\frac{\partial S}{\partial t}+H(t,q^i,p_i)=K(t,x^i,y_i).
\end{equation}

This could be useful in several cases. For instance, if the new
Hamiltonian $K$ happens to depend only on the $(x^i)$ coordinates,
then the Hamilton's equations become
\begin{equation}\label{hamilton4}
\begin{array}{rl}
\displaystyle\frac{d x^i}{dt}(t)=&\displaystyle\frac{\partial K}{\partial
  y_i}(x^i(t))=0,\\ \noalign{\medskip}
\displaystyle\frac{d y_i}{dt}(t)=&-\displaystyle\frac{\partial K}{\partial x^i}(x^i(t)),
\end{array}
\end{equation}
which are trivially integrable. The Hamilton--Jacobi theory is understood later on as the particular case $K=0$, that is, the new dynamics are trivial. These results will follow from our theory of generating functions.
%%%%%%%%%%%%%%%%%%%%%%%%%%%%%%%%%%%%%%%%%%%%%%%%%%%%%%%%%%

%%%%%%%%%%%%%%%%%%%%%%%%%%%%%%%%%%%%%%%%%%%%%%%%%%%%%%%%%%

\subsection{Fundamental Lemma}

Recall that given a Lagrangian bisection $\lag$ in $T^*(\mathbb{R}\times G)$ then, using Theorem \ref{PoissonIsomorphisms},
we have the induced Poisson isomorphism
\[
\hat{\lag}=\widetilde\ralpha\circ(\widetilde\rbeta_{|\lag})^{-1}:
\mathbb{R}\times A^*G\longrightarrow \mathbb{R}\times A^*G.
\]
Note that the induced transformation by the Lagrangian bisection $\lag$ is obtained as follows:
\begin{enumerate}
\item Take $(t_0,\mu_x)\in \mathbb{R}\times A^*G$.
\item Find the unique point $(t_1,e_1,\varpi_g)$ in $\lag$ such that $\widetilde\rbeta(t_1,e_1,\varpi_g)=(t_0,\mu_x)$. That is equivalent to 
\[
\begin{array}{l}
t_0=t_1,\\ \noalign{\medskip}
\widetilde\beta(\varpi_g)=\mu_x.
\end{array}
\]
\item Finally, $\hat\lag(t_0,\mu_x)=\widetilde\ralpha(t_1,e_1,\varpi_g)$, and, by definition of $\ralpha$ we conclude 
\[
\hat\lag(t_0,\mu_x)=(t_0,\widetilde\alpha(\varpi_g)).
\]
\end{enumerate}
In that sense, Lagrangian bisections in $T^*(\rgroupoid)$ respect the product structure in $\mathbb{R}\times A^*G$ and parametrize time--dependent Poisson automorphisms of $A^*G$.
The next lemma says how the time evolution,
$\frac{\partial}{\partial t}$, transforms under the
mapping $\hat{\lag}$. This is the hard part of the proof of the Main Theorem, since the
evolution of the Hamiltonian vector field $X_H$ is known because
$\hat{\lag}$ is a Poisson mapping. Once more, we assume that the
induced map is a global diffeomorphism, since the local case follows
changing the domains and range of definition by the appropriate open
sets. We will denote by $e$ the coordinate function  $e:T^*(\mathbb{R}\times G)\rightarrow
\mathbb{R}$ given by $e(t,e,\varpi_g)=e$.

\begin{lemma}\label{fundamental} Let $\lag$ be a Lagrangian bisection of
  the symplectic groupoid $T^*(\rgroupoid)$. Then the induced Poisson
  mapping, $\hat{\lag}:\rad\rightarrow \rad$, satisfies
\begin{equation}\label{time}
\hat{\lag}_*(\displaystyle\frac{\partial}{\partial
  t})=\displaystyle\frac{\partial}{\partial t}+X_{\sigma_2^*e}
\end{equation}
where $\sigma_2=(\widetilde\ralpha_{|\lag})^{-1}$.
\end{lemma}
\begin{proof}
For the sake of clarity of the exposition we are going
to divide the proof into three steps. In the first one we show that
the Lagrangian bisection $\lag$ leaves in the $0$-level set of a certain
function $F$, in that sense, it satisfies certain Hamilton--Jacobi
equation. In the second step, using that $\lag$ is Lagrangian, we show
that the Hamiltonian vector of $F$ is tangent to $\lag$ and that is enough to compute the
push--forward of the time evolution plus certain Hamiltonian vector
field. In the third step we use that the push--forward of the
Hamiltonian vector fields is easy to treat under Poisson isomorphisms
to conclude the desired result. 

\vspace*{0.2cm}
\textbf{First step:} From now on we will write
$\sigma_1=(\widetilde\rbeta_{|\lag})^{-1}:\rad\rightarrow
\lag\subset T^*(\rgroupoid)$ and
$\sigma_2=(\widetilde\ralpha_{|\lag})^{-1}:\rad\rightarrow \lag\subset
T^*(\rgroupoid)$. Consider the Hamiltonian
$-\sigma_1^*{e}:\rad\rightarrow \mathbb{R}$. Then the associated
extended Hamiltonian is
$(\sigma_1^*e)^{ext}=-\widetilde\rbeta^*(\sigma_1^*e)+e$. Notice that the
Lagrangian submanifold $\lag$ satisfies the Hamilton--Jacobi equation,
that is, $(-\widetilde\rbeta^*(\sigma_1^*e)+e)_{|\lag}=0$. Equivalently,
$\lag$ lives in the $0$--level set of the function
$(-\widetilde\rbeta^*(\sigma_1^*e)+e)$. For simplicity we will usually
write $(-\widetilde\rbeta^*(\sigma_1^*e)+e)=F$.

 \vspace*{0.2cm}
\textbf{Second step:} Let $X_F$ be the Hamiltonian vector field of $F$
in $T^*(\mathbb{R}\times G)$, that is, the unique vector field
satisfying $i_{X_{F}}\omega_{\mathbb{R}\times G}=dF$. Since $\lag$ is
Lagrangian and $F_{|\lag}=0$ then $X_F$ is tangent to $\lag$ along the points
of $\lag$.  Let us analyze $X_F$, since $F=-\widetilde\rbeta^*(\sigma_1^*e)+e$, then
$X_F=X_{-\widetilde\rbeta^*(\sigma_1^*e)}+X_e$. Now we make two
observations:
\begin{enumerate}
\item $X_e=\frac{\partial}{\partial t}$ and
  $\widetilde\rbeta_*(X_e)=\frac{\partial}{\partial t}$.
\item Since $\widetilde\rbeta$ is a Poisson mapping (Theorem
  \ref{theorem-sg}), $\widetilde\rbeta
  _*(X_{-\widetilde\rbeta^*(\sigma_1^*e)})=X_{-\sigma_1^*e}$. Summarizing,
\end{enumerate}

\begin{equation}\label{projection}
\widetilde\rbeta _*(X_F)=\widetilde\rbeta _*(X_{-\widetilde\rbeta^*(\sigma_1^*e)+e})=\displaystyle\frac{\partial}{\partial t}+X_{-\sigma_1^*e}.
\end{equation}
Nevertheless, since $\widetilde\rbeta$ and $\widetilde\ralpha$
are a dual pair, i.e., the $\widetilde\rbeta$-fibers and $\widetilde\ralpha$-fibers are
symplectically orthogonal (Theorem
  \ref{theorem-sg}) or, equivalently,
$\ker(T\widetilde\rbeta)^\perp=\ker(T\widetilde\ralpha)$ and
since $d(-\widetilde\rbeta^*(\sigma_1^*e))\in \ker(T\widetilde\rbeta)^\circ$ by
construction, then
\[X_{(-\widetilde\rbeta^*(\sigma_1^*e))}=(\omega_{\rgroupoid})^\sharp(\widetilde d(-\widetilde\rbeta^*(\sigma_1^*e)))\in
\ker(T\widetilde\ralpha)\footnote{Recall that by definition
$\ker(T\widetilde\rbeta)^\perp=(\omega_{\rgroupoid})^{\sharp}(\ker(T\widetilde\rbeta)^\circ)$.}.\] This
last equation implies that
$\widetilde\ralpha_*(X_{(-\widetilde\rbeta^*(\sigma_1^*e))})=0$ and thus 
\begin{equation}\label{pushforward2}
\widetilde\ralpha_*(X_{F})=\displaystyle\frac{\partial}{\partial t}.
\end{equation}
Now we proceed to compute $\hat{\lag}_*(\frac{\partial}{\partial
  t}+X_{-\widetilde\rbeta^*(\sigma_1^*e)})$. By equation \eqref{projection} and since $X_{F}$ is
tangent to $\lag$, we have that
$(\widetilde\rbeta_{|\lag})_*(X_{F})=\frac{\partial}{\partial
  t}+X_{-\sigma_1^*e}$ or equivalently, reversing the arguments
\begin{equation}\label{equ1}
(\widetilde\rbeta_{|\lag})^{-1}_*(\frac{\partial}{\partial
  t}+X_{-\sigma_1^*e})=({X_{F}})_{|\lag}.
\end{equation}
Combining equation \eqref{equ1} and equation \eqref{pushforward2} we
obtain in a straightforward way
\begin{equation}\label{result}
\hat{\lag}_*(\displaystyle\frac{\partial}{\partial t}+X_{-\sigma_1^*e})=\displaystyle\frac{\partial}{\partial t}
\end{equation}
and we conclude the second step of the proof.

\vspace*{0.2cm}
\textbf{Third step:}
Now, since $\hat{\lag}$ is a Poisson automorphism, we have
   $\hat{\lag}_*(X_{\sigma_1^*e})=X_{(\hat{\lag}^{-1})^*(\sigma_1^*e)}$. It
is easy to see that $\sigma_1^*e\circ \hat{\lag}^{-1}=\sigma_2^*e$
and thus 
\begin{equation}\label{two}
\hat{\lag}_*(X_{\sigma_1^*e})=X_{\sigma_2^*e}. 
\end{equation}
Finally, adding equations \eqref{result} and \eqref{two} we get
\[
\hat{\lag}_*(\frac{\partial}{\partial
  t})=\hat{\lag}_*\Big(\displaystyle\frac{\partial}{\partial
  t}+X_{-\sigma_1^*e}+X_{\sigma_1^*e}\Big)=\frac{\partial}{\partial t}+X_{\sigma_2^*e}.\qedhere
\]
\end{proof}
\begin{remark}
In the previous Lemma, when $\lag=\im(dS)$, with $S \in C^{\infty}(\mathbb{R} \times G)$, then
\[\sigma_2^*e=\sigma_2^*\left(\displaystyle\frac{\partial S}{\partial
  t}(t,g)\right).\] That is, the new Hamiltonian is the function $\displaystyle\frac{\partial S}{\partial
  t}(t,g)$ read in $A^*G$ after the transformation  induced by $\im(dS)$.
\end{remark}

%%%%%%%%%%%%%%%%%%%%%%%%%%%%%%%%%%%%%%%%%%%%%%%
\subsection{Main Result}
With the previous lemma at hand, we can easily prove the following
theorem. We maintain the notations used in the previous lemma.

\begin{theorem}\label{hamiltoniantransformation} Let
  $H:\rad\rightarrow \mathbb{R}$ be a Hamiltonian function and
  $\lag\subset T^*(\rgroupoid)$ a Lagrangian bisection. The induced
  mapping $\hat{\lag}:\rad\rightarrow\rad$ satisfies
\[
\hat{\lag}_*(\displaystyle\frac{\partial}{\partial
  t}+X_H)=\displaystyle\frac{\partial}{\partial t}+X_K
\]
where $K=H\circ{\hat{\lag}^{-1}}+\sigma_2^*e$.
 \end{theorem}

\begin{proof} The proof is a consequence of Lemma
\ref{fundamental}. First notice that since $\hat{\lag}$ is Poisson,
we have $\hat{\lag}_*(X_H)=X_{H\circ\hat{\lag}^{-1}}$. Combining
this result with Lemma \ref{fundamental} we obtain
\[
\hat{\lag}_*(\displaystyle\frac{\partial}{\partial
  t}+X_H)=\hat{\lag}_*(\frac{\partial}{\partial
  t})+\hat{\lag}_*(X_H)=\frac{\partial}{\partial
  t}+X_{\sigma_2^*e}+X_{H\circ\hat{\lag}^{-1}}=\frac{\partial}{\partial
t}+X_K.\qedhere
\]
\end{proof}

\begin{remark}
Of course, the resemblances between the equations $\frac{\partial
  S}{\partial t}+H(t,q^i,p_i)=K(t,x^i,y_i)$ and
$H\circ{\hat{\lag}^{-1}}+\sigma_2^*e=K$ are obvious. When
$\lag=\im(dS)$ the last equation reduces to a generalization of the
first one
\[
K=H\circ{\hat{\lag}^{-1}}+\displaystyle\frac{\partial S}{\partial t}.
\]

In the pair groupoid case we recover the classical results.
\end{remark}

\begin{remark} The Poisson isomorphism described above is just the one given by Theorem \ref{PoissonIsomorphisms}. It is natural to wonder which Poisson transformation on $\rad$ are of
the form $\hat{\lag}$ for some Lagrangian bisection $\lag\subset
T^*(\rgroupoid)$. It is shown in the reference
\cite{GeHJ} that under some conditions ($\widetilde\rbeta^{-1}(\mu_x)$
be connected for all $\mu_x\in \rad$) one can find any Poisson
transformation preserving the symplectic leaves of $\rad$ by that
procedure. 

\end{remark}

%%%%%%%%%%%%%%%%%%%%%%%%%%%%%%%%%%%%%%%%%%%%%%%%%
\subsection{Hamilton--Jacobi theory for Poisson Manifolds}
The next theorem is the generalization of the results shown in
Section \ref{motivation}. The Hamilton--Jacobi theorem below, as it
happens in the classical case, is a consequence of the theory of
generating functions developed in the previous section. The main idea
is to look for a transformation that makes the new dynamics trivial or
easy to integrate, now that we have a formula relating the
transfomation $\hat{\lag}$, the original Hamiltonian $H$ and the new
Hamiltonian function $K$. Along this section $G$ will be a groupoid over a manifold $M$ with source and target $\alpha$ and $\beta$ respectively. The corresponding algebroid will be denoted by $AG$ and its dual by $A^*G$.

Now let us construct a function
  $H^{ext}:T^*(\mathbb{R\times G})\rightarrow \mathbb{R}$ out of the
  Hamiltonian $H$ that we call the extended Hamiltonian associated to
  $H$. Define $H^{ext} = \widetilde\rbeta^*H+e$. Finally, we state the
  Hamilton-–Jacobi theorem that allows the integration of the
  Hamilton's equations once a solution of a certain PDE has been
  found. The appropriate coordinates where the PDE has solution will
  be given in the next sections.

\begin{theorem}[Hamilton--Jacobi]\label{hj} Let $\lag$ be a Lagrangian submanifold of $T^*(\mathbb{R}\times G)$ and assume that
\begin{enumerate}
\item \emph{Hamilton--Jacobi equation}, $H^{ext}_{|\lag}=\widetilde\ralpha^*K$, where
  $K=\ralpha^*k$ for some function $k:\rbase\rightarrow
  \mathbb{R}$. That means that $K$ is a $\ralpha$-basic function. 
\item \emph{Non-degeneracy condition}, $\lag$ is a bisection, or equivalently,
  $\widetilde\ralpha_{|\lag}:\lag\rightarrow \rad$ is a (local) diffeomorphism.
\end{enumerate}
Then the induced Poisson isomorphism $\widetilde\ralpha\circ
(\widetilde\rbeta_{|\lag})^{-1}:\rad\rightarrow \rad$ (henceforth
denoted $\hat{\lag}$) satisfies
\begin{equation}\label{transformation}
\hat{\lag}_*(\displaystyle\frac{\partial}{\partial t}+X_H)=\displaystyle\frac{\partial}{\partial t}+X_{\hat{k}},
\end{equation}
where $\hat{k}=\rtau^*k$.
\end{theorem}

\begin{proof} The proof is an obvious application of Theorem \ref{hamiltoniantransformation}.

\end{proof}

\begin{remark}
\begin{enumerate}
\item When $K=0$ the new dynamics are $\frac{\partial}{\partial t}$, that is, the system after the transformation is in {\bf equilibrium}. The inverse of the this trasformation is, up to an initial condition, the flow of the Hamiltonian system given by the Hamiltonian $H$.
\item The Hamilton--Jacobi equation is easily seen to recover the two examples discussed
  in the introduction. In the classical case it reads as 
  \eqref{HamiltonJacobiK}. In the general case, if $\lag=\im(dS)$ then it becomes 
\[\displaystyle\frac{\partial
    S}{\partial t}+H(t,\widetilde\beta(\frac{\partial S}{\partial
    g}))=k(t,\alpha(g))\]
 where $dS=\frac{\partial
    S}{\partial t}dt+\frac{\partial S}{\partial
    g}dg$ comes from the identification
 $T^*(\rgroupoid)=T^*\mathbb{R}\times T^*G$, so $\frac{\partial
    S}{\partial t}dt\in T^*\mathbb{R}$ and $\frac{\partial S}{\partial
    g}dg\in T^*G$. The reader should compare this expression with  the
  Hamilton--Jacobi equations in the classical and Lie group cases. 

\item The non-degeneracy condition is the geometric description of the
  non-degeneracy previously used ($\det(\frac{\partial^2S}{\partial
    q^i\partial Q^j})\neq 0$ and $\lvec{\xi}_a(\rvec{\xi}_b(S_t))$ be regular). It can be stated saying that
  $\im(dS)$ is a  bisection.

\item Obviously, this theorem will only be useful if the new
  Hamiltonian dynamics  $\displaystyle\frac{\partial}{\partial
    t}+X_{\hat{k}}$ is trivially integrable, as before. Take
  coordinates adapted to the fibration $\alpha$, say $t,x^i,y_j$, $i=1,\ldots,n$,
  $j=1,\ldots,m$, where $t$ is the $\mathbb{R}-$coordinate, $n$ the
  dimension of $M$ and $m$ the dimension of the $\alpha$-fibers. Thus,
  the equations of motion read
 \begin{equation}\label{trivial}
 \begin{array}{l}
 \displaystyle\frac{d x^i}{dt}=\{x^i,\hat{k}(t,x^i)\}=0,\\
 \noalign{\bigskip}
 \displaystyle\frac{d
   y_j}{dt}=\{y_j,\hat{k}(t,x^i)\}=\rho^k_j(x^i)\frac{\partial \hat{k}}{\partial x^k}(t,x^i),
 \end{array}
 \end{equation}
 where $i=1,\ldots,n$, $j=1,\ldots,m$. So if $\hat{k}$ is time-independent, given an initial condition
 $(x^i_0,y_j^0)$ at time $t_0$ the solution of Hamilton's equations is
 \[
 t\rightarrow \Big(t,x^0_i, y_j^0+\rho^k_j(x_i^0)\displaystyle\frac{\partial
 \hat{k}}{\partial x^k}(x_i^0)(t-t_0)\Big).
 \]

Otherwise, in the time-dependent setting the solution is just given by
integration
\[
t\rightarrow \Big(t,x^0_i,\int_{t_0}^t \rho^k_j(x^i_0)\frac{\partial \hat{k}}{\partial x^k_0}(s,x^i_0)ds+y_i^0\Big).
\]
\end{enumerate}
\end{remark}

\begin{remark}\label{lagrangian}
Notice that we are working with Lagrangian bisections but 
along the next section we will provide the necessary method to obtain these Lagrangian submanifolds, that is, our theory of generating functions. The main
difficulty is to describe Lagrangian submanifolds which are not horizontal
(i.e. $\im(\gamma)$ for a closed $1$-form $\gamma$), which amount to deal with
generating functions which are not of the ``first type''. In the
classical theory, it means that one can allow $S$ to be dependent not
only on the variables $(t,q^i,x^i)$, but on ``mixed'' variables like
$(t,q^i,y_i)$. The previous statements hold only locally. Lagrangian
submanifolds are the geometric objects behind generating
functions. Even when the generating functions develop singularities
(\emph {caustics}), these objects are well-defined, as it is pointed
out in \cite{AbrahamMarsden}.
\end{remark}

%%%%%%%%%%%%%%%%%%%%%%%%%%%%%%%%%%%%%%%%%%%%%%%%%%%%%%%%%%
\subsection{Generating the identity and nearby transformations: {\it non-free} canonical transformations}\label{non-free}
\subsubsection{Motivation: Classical Case}\label{nonfreeclassical}

We start this subsection showing that even in the classical situation,
described in 
Section \ref{classical}, the identity transformation
$\id_{T^*Q}:T^*Q\rightarrow T^*Q$ can not be obtained through generating
functions of the first type, {\it i.e.} $S(x^i,q^i)$, where $(q^i)$
and $(x^i)$ are coordinates in $Q$, the configuration
manifold. Equivalently, the Lagrangian submanifold given by the
identities is not horizontal for $\pi_{Q\times Q}$. In the language of
\cite{Arnoldmmcm} a canonical transformation is called {\it
  free} if it has a generating function of type I and {\it non-free}
otherwise. The local
coordinate description is all that we are going to need and, thus, we assume that $Q=\mathbb{R}^n$
and so $T^*Q=\mathbb{R}^{2n}$ and consider global coordinates
$(q^i,p_i)$. Doubling these coordinates we get a coordinate system for
$T^*(Q\times Q)=\mathbb{R}^{4n}$, say $(x^i,y_i,q^i,p_i)$. It is easy to see that the
Lagrangian submanifold of the identities $\lag=\widetilde\epsilon(T^*Q)$ is given by
the set of points of the form
\[
\lag=\{(q^i,-p_i,q^i,p_i)\textrm{ such that } q^i, \ p_i\in\mathbb{R}\},
\]
or equivalently
\[
\lag=\{(x^i,y_i,q^i,p_i)\textrm{ such that } q^i=x^i, \ -p_i=y_i\in\mathbb{R}\}.
\]
Using the description of the source and the target it can be
checked that $\lag$ induces the identity transformation. Obviously, there
does not exist $S(x^i,q^i)$ such that $\im(dS)=\lag$. In order to solve that, we introduce the following symplectomorphism, 
\[
\begin{array}{rccl}
F:& (\mathbb{R}^{4n},\ dx^i\wedge dy_i+\ dq^i\wedge dp_i)&\longrightarrow & (\mathbb{R}^{4n},\ dy_i\wedge dx^i+\ dq^i\wedge dp_i)\\ \noalign{\medskip}
&(x^i,y_i,q^i,p_i)&\rightarrow & (\tilde{x}^i=x^i,\tilde{y}_i=-y_i,\tilde{q}^i=q^i,\tilde{p}_i=p_i).
\end{array}
\]
The main idea is that we interchanged the role of $x^i$ and $y_i$. Since symplectomorphisms conserve Lagrangian submanifolds, $\tilde{\lag}=F(\lag)$ is a Lagrangian submanifold in $(\mathbb{R}^{4n},\ dy_i\wedge dx^i+dq^i\wedge dp_i)$. The submanifold $\tilde{\lag}$ is explicitly given by
\[
\tilde{\lag}=\{(q^i,p_i,q^i,p_i)\textrm{ such that } q^i, \ p_i\in\mathbb{R}\}.
\]
We introduce a projection analogous to $\pi_Q$, which in
general is only defined locally and in a non-canonical way
\begin{equation}
\begin{array}{rccl}
\pi :& \mathbb{R}^{4n}&\longrightarrow & \mathbb{R}^{2n}\\ \noalign{\medskip}
& (x^i,y_i ,q^i,p_i)&\rightarrow & \pi (x^i,y_i ,q^i,p_i)=(y_i,q^i).
\end{array}
\end{equation}
Now, $\pi:\mathbb{R}^{4n}\rightarrow \mathbb{R}^{2n}$ with the symplectic structure in $\mathbb{R}^{4n},\ dq^i\wedge dp_i+dy_i\wedge dx^i$ is ``the same as a cotangent bundle'', but there $\tilde{L}$ is horizontal. We can consider the function $S(y_i,q^i)=q^i\cdot y_i$ and $\tilde{L}=\im(dS)$. It is now obvious that we can use $F$ to turn non-horizontal Lagrangian submanifolds into horizontal ones. Using $F^{-1}$, once a generating function is obtained for the desired Lagrangian submanifold, we can describe a canonical transformation implicitly following the same pattern we used in \eqref{coordinatechange}. That is 
\[
\tilde{\lag}=\{(\frac{\partial S}{\partial y_i}(y_i,q^i),y_i ,q^i,\displaystyle\frac{\partial S}{\partial q^i}(y_i,q^i))\textrm{ such that } q^i, \ y_i\in\mathbb{R}\}
\]
and 
\[
\lag=F^{-1}(\tilde{\lag})=\{(\frac{\partial S}{\partial y_i}(y_i,q^i),-y_i ,q^i,\displaystyle\frac{\partial S}{\partial q^i}(y_i,q^i))\textrm{ such that } q^i, \ y_i\in\mathbb{R}\}
\]
which in general will induce the transformation analogous to \eqref{coordinatechange}
\begin{equation}\label{coordinatechangenonfree}
\begin{array}{cc}
\displaystyle\frac{\partial S}{\partial y^i}(t,y_i,q^i)=x_i, &\displaystyle\frac{\partial S}{\partial q^i}(t,y_i,q^i)=p_i,
\end{array}
\end{equation}
which is usually called a type II generating function. The expression
above is found in most classical mechanics books, like
\cite{Arnoldmmcm,Goldstein}. In the particular case $S(y_i,q^i)=q^i\cdot y_i$, \eqref{coordinatechangenonfree} reads
\[
\begin{array}{cc}
y_i=p_i, & q^i=x^i.
\end{array}
\]
that is, we generated the identity.

\begin{remark} The transformations \eqref{coordinatechange}, \eqref{coordinatechangenonfree} are just the transformations defined by the corresponding Lagrangian bisections following Theorem \ref{CDW}.
\end{remark}

The previously described situation illustrates the importance of
``projecting'' Lagrangian submanifolds in the appropriate
setting. While their geometry is very rich and powerful, whenever one
is interested in computations these Lagrangian submanifolds should be
described by functions. Conditions on the Lagrangian submanifold, like
the Hamilton--Jacobi equation $H^{ext}_{|\lag}=0$, become PDE's that can be
treated by the techniques of analysis. In the rest of this section we
will describe analogous procedures to generate the identity in the
three cases of interest: Lie algebras, action  and Atiyah Lie
algebroids.

\subsubsection{Lie Algebras}\label{generatingalgebras}
Let $\mathfrak{g}$ a Lie algebra and $G$ a Lie group integrating it (not necessarily simply-connected). Then the cotangent groupoid $T^*G\rightrightarrows\mathfrak{g}^*$ is endowed with the following two natural projections
\begin{equation}\label{liealgebra}
\begin{array}{rccl}
\tilde{\alpha}:& T^*G&\longrightarrow &  \mathfrak{g}^* \\ 
& (\mu_g) &\to &\tilde{\alpha}(\mu_g)=J_L(\mu_g) ,\\ \noalign{\bigskip}

\tilde{\beta}:& T^*G&\longrightarrow &  \mathfrak{g}^* \\ 
& (\mu_g) &\to &\tilde{\beta}(\mu_g)=J_R(\mu_g).
 \end{array}
\end{equation}
These are the only structural mappings needed, we will not use the algebraic properties of groupoids in this section. It is easy to check that the submanifold of identities $\lag=\widetilde\epsilon(\mathfrak{g}^*)$ is given by the fiber $\pi_G^{-1}(\ide)$, where $\ide$ is the identity element of the Lie group $G$. It is quite illustrative to notice that we have the opposite situation of an horizontal Lagrangian submanifold for $\pi_G$, the projection of $\lag$ onto $G$ by $\pi_G:T^*G\rightarrow G$ is just $\{\ide\}$. This suggests that we should try to ``project $\lag$ onto the fibers'', not into $G$. In order to do that, take local coordinates in $G$ around the identity $\ide$, say $(g^i)$, $i=1,\ldots,n$ and let $(g^i,p_i)$ be the associated natural coordinates on $T^*G$. Assume for the sake of simplicity that $g^i(\ide)=0$. In those coordinates $\lag$ is given by
\[
\lag=\{(0,p_i)\textrm{ such that }p_i\in\mathbb{R}\}.
\]
 We introduce the, only locally defined, projection
\begin{equation}\label{}
\pi: (g^i,p_i)  \longrightarrow   (p_i).
\end{equation}
We think about the space $(g^i,p_i)$ with the symplectic form $dg^i\wedge dp_i$ and projection $\pi$ as a ``cotangent bundle''. Then we can consider $S(p^i)=0$ and $\lag$ is identified as the graph of the differential of this function
\[
\lag=\{(0,p_i)\textrm{ such that }p_i\in\mathbb{R}\}=\{(\displaystyle\frac{\partial S}{\partial p_i}(p_i),p_i)\textrm{ such that }p_i\in\mathbb{R}\}.
\]
The point of this construction is that nearby Lagrangian submanifolds to $\widetilde\epsilon(\mathfrak{g}^*)$ will be described by this kind of functions as well. This is the keystone for the construction of our numerical methods in Section \ref{examples} and it solves a question by McLachlan and Scovel (\cite{McLachlanScoveII}) that we mentioned in the introduction. Although these constructions are local, they permit us to parametrize the requested Lagrangian submanifolds, as it can be shown in our examples.

To finish this subsection, we illustrate how to deal with an example taken from \cite{GeEquivariant} where it is shown that there is no way to find a generating function depending on coordinates on the base $S(g^i)$ for the identity $\id_{\mathfrak{g}^*}:\mathfrak{g}^*\rightarrow \mathfrak{g}^*$. How to generate the identity transformation was only known for quadratic Lie algebras, \cite{ChannellScovel}.

\begin{example} Consider the group
\[
\textrm{Trian}(2)=\left\{ \left(\begin{array}{lr} a_{11} & a_{12}\\ \noalign{\medskip} 0& a_{22} \end{array}\right)\in \textrm{GL}(2,\mathbb{R})\right\}.
\]
We introduce coordinates
\[
(g^1,g^2,g^3)\rightarrow \left(\begin{array}{lr} a_{11}=g^1 & a_{12}=g^2\\ \noalign{\medskip} 0& a_{22}=g^3 \end{array}\right)
\]
defined for $(g^1,g^2,g^3)\in (\mathbb{R}-\{0\})\times (\mathbb{R}-\{0\})\times \mathbb{R}$. The product is given in local coordinates by $(g^1,g^2,g^3)\cdot (\tilde{g}^1,\tilde{g}^2,\tilde{g}^3)=(g^1\tilde{g}^1,g^1\tilde{g}^2+g^2\tilde{g}^3,g^3\tilde{g}^3)$. Let $(g^1,g^2,g^3,p_1,p_2,p_3)$ be natural coordinates in $T^*\textrm{Trian}(2)$ and a straightforward computation shows that 
\[
\begin{array}{l}
\tilde{\alpha}(g^1,g^2,g^3,p_1,p_2,p_3)=J_L(g^1,g^2,g^3,p_1,p_2,p_3)=(p_1g^1+p_2g^2,p_2g^3,p_3g^3)\\\noalign{\medskip}
\tilde{\beta}(g^1,g^2,g^3,p_1,p_2,p_3)=J_R(g^1,g^2,g^3,p_1,p_2,p_3)=(p_1g^1,p_2g^1,p_2g^2+p_3g^3).
\end{array}
\]
In this case in the coordinates $(g^1,g^2,g^3)$ the identity element $ \ide$
reads as $(1,0,1)$, so taking $S(p_1,p_2,p_3)=p_1+p_3$ implies that 
\[
\begin{array}{rl}
\im(dS)&=\{(\displaystyle\frac{\partial S}{\partial p_1}(p),\displaystyle\frac{\partial S}{\partial p_2}(p),\displaystyle\frac{\partial S}{\partial p_3}(p),p_1,p_2,p_3)\textrm{ such that }p_i\in\mathbb{R}\}\\\noalign{\medskip} &=\{(1,0,1,p_1,p_2,p_3)\textrm{ such that }p_i\in\mathbb{R}\}.
\end{array}
\]
Now we have 
\begin{equation}\label{equation}
\begin{array}{l}
\tilde{\alpha}(dS(p_1,p_2,p_3))=(p_1,p_2,p_3),\\ \noalign{\medskip}
\tilde{\beta}(dS(p_1,p_2,p_3))=(p_1,p_2,p_3),
\end{array}
\end{equation}
and equation \eqref{equation} implies that the transformation induced by $\im(dS)$ is the identity.
\end{example}

\subsubsection{Action Lie Algebroids}\label{generatingaction}
Action Lie algebroids are discussed in Appendix
\ref{liegroupoids}. The connection of these algebroids with semi-direct
products is very remarkable, as these structures have  proved their
importance describing some of the dynamical systems relevant for geometric mechanics, see \cite{SemiDirect}. Assume that our action Lie algebroid, say $M\times \mathfrak{g}$, is integrable with the action Lie groupoid $M\times G$ integrating it. Then $T^*(M\times G)\rightrightarrows M\times\mathfrak{g}^*$ is a symplectic Lie groupoid with source and target given by
\begin{equation}\label{actionliealgebroid}
\begin{array}{rccl}
\tilde{\alpha}:& T^*(M\times G)&\longrightarrow &  M\times \mathfrak{g}^* \\ 
& (\mu^1_q,\mu^2_g) &\to &\tilde{\alpha}(\mu^1_q,\mu^2_g)=(q,-J(\mu^1_q)+J_L(\mu^2_g)), \\ \noalign{\bigskip}

\tilde{\beta}:& T^*(M\times G)&\longrightarrow &  M\times \mathfrak{g}^* \\ 
& (\mu^1_q,\mu^2_g) &\to &\tilde{\beta}(\mu^1_q,\mu^2_g)=(qg,J_R(\mu^2_g)).
 \end{array}
\end{equation}
Here, $J$ is the standard momentum mapping for lifted
actions, \[J(\mu_q)(\xi)=\mu_q(\xi_M),\]$\xi_M$ is the infinitesimal
generator associated to $\xi\in\mathfrak{g}$. On the other hand, the Lagrangian submanifold of identities is given by \[\lag=\widetilde\epsilon(M\times \mathfrak{g}^*)=\{(0_m,\mu)\textrm{ such that }m\in M, \ \mu\in\mathfrak{g}^*=T_{\frak{e}}^*G\}. \]
Take local coordinates $(x^i)$ on an open set $U\subset M$, and $(g^j)$ coordinates on a neighbourhood of the identity on $G$, moreover, we assume that $g^i(\frak{e})=0$. We consider natural coordinates $(x^i,y_i,g^j,p_j)$ on $T^*(M\times G)$. In these coordinates $\lag$ is given by
\[
\{(x^i,0,0,p_j)\textrm{ such that }x^i\in U,\ p_j\in\mathbb{R}\}.
\] 

Arguing like in the previous sections, we get that we have to
``project'' onto the $(x^i,p_j)$-coordinates. We get a generating
function $S(x^i,p_j)$ that produces the Lagrangian submanifold, after a change of sign using a mapping analogous to $F$ in Section \ref{nonfreeclassical}
\[
\{(x^i,\displaystyle\frac{\partial S}{\partial x^i}(x^i,p_j),\displaystyle\frac{\partial S}{\partial p_j}(x^i,p_j),p_j),\textrm{ such that }x^i\in U,\ p_j\in\mathbb{R}\}.
\]
So taking $S(x^i,p_j)=\text{\it constant\/}$ gives us the desired generating function.

\subsubsection{Atiyah Algebroids}\label{Gen-func-Atiyah}
Atiyah algebroids are discussed in Appendix \ref{liegroupoids}, their cotangent groupoids were already introduced in Section \ref{cotangentatiyah}. We assume that $P\rightarrow M$ is trivial from the very beginning, as we are interested in the local description, so we have $T^*(M\times M\times G)$, and recall that 
\begin{equation}\label{actionliealgebroidz}
\begin{array}{rccl}
\tilde{\alpha}:& T^*(M\times M\times G)&\longrightarrow &  T^*M\times \mathfrak{g}^* \\ 
& (\lambda_{m^1},\mu_{m^2},\nu_{g}) &\to &\tilde{\alpha}(\lambda_{m^1},\mu_{m^2},\nu_{g})=(-\lambda_{m^1},J_L(\nu_{g})), \\ \noalign{\bigskip}

\tilde{\beta}:& T^*(M\times M\times G)&\longrightarrow & T^*M\times\mathfrak{g}^* \\ 
& (\lambda_{m^1},\mu_{m^2},\nu_{g}) &\to &\tilde{\beta}(\lambda_{m^1},\mu_{m^2},\nu_{g})=(\mu_{m^2},J_R(\nu_{g})).
 \end{array}
\end{equation}
The Lagrangian submanifold of identities is given by 
\[
\{(-\lambda_{m},\lambda_{m},\nu)\textrm{ such that } \lambda_m\in T^*M, \ \nu\in\mathfrak{g}^*\}.
\]

It should be clear, looking at the expressions above, that this case is a combination of Section \ref{nonfreeclassical} and Section \ref{generatingalgebras}. Taking the union of the $(x^i,q^i)$ and $g^j$ introduced in the aforementioned sections, we get the coordinate system $(x^i,q^i,g^j)$ on $M\times M\times G$. Combining the arguments exposed there, a generating function of the form $S(y^i,q^i,p^g_j)$ is obtained in a straightforward manner. The $p^g_j$ in the previous expression are the momenta associated to the $g^j$.

\subsubsection{Transitive Lie algebroids}
A Lie groupoid $G \rightrightarrows M$ with source $\alpha: G \to M$ and target $\beta: G \to M$ is said to be transitive if the anchor map
\[
(\alpha, \beta): G \to M \times M, \; \; g \in G \to (\alpha(g), \beta(g))
\]
is a surjective submersion (see \cite{Mackenzie}). Examples of transitive Lie groupoids are the pair or banal groupoid associated with a manifold, Lie algebras or Atiyah algebroids associated with principal bundles.

The Lie algebroid $AG$ of a transitive Lie groupoid $G \rightrightarrows M$ is transitive, that is, the anchor map of $AG$ is a vector bundle epimorphism on $TM$ (see \cite{Mackenzie}). On the other hand, one may prove that a transitive Lie algebroid is locally isomorphic to the Atiyah algebroid associated with a trivial principal bundle over $M$ (see Chapter $8$ in \cite{Mackenzie}).

So, using the results in Section \ref{Gen-func-Atiyah}, we can find local generating functions for the Lagrangian submanifold of the identities in the cotangent bundle $T^*G$ of a transitive Lie groupoid $G\rightrightarrows M$.

To conclude this section we would like to mention that the important thing here is not the parametrization of the Lagrangian submanifold of identities $\lag$, which induces the identity transformation and leaves everything unchanged. The main contribution of the previous subsections is that one can parametrize all the Lagrangian submanifolds, which are close to $\lag$, as the differentials of the appropriate functions. With that, one obtains a method to parametrize all the Hamiltonian flows, at least in the appropriate subsets, for small enough time which is our object of study.

%%%%%%%%%%%%%%%%%%%%%%%%%%%%%%%%%%%%%%%%%%%%%%%%%%%%%%%%%%
%% LOCAL EXISTENCE OF SOLUTIONS
%%%%%%%%%%%%%%%%%%%%%%%%%%%%%%%%%%%%%%%%%%%%%%%%%%%%%%%%%%
\section{Local Existence of Solution}

Since we are dealing with a PDE one of our main tasks is to show, at least, local
existence of solutions of the corresponding Hamilton--Jacobi equation. Our proof
follows the construction of Ge in \cite{GeHJ} to build a Lagrangian
submanifold which satisfies the Hamilton--Jacobi equation, according to Remark
\ref{lagrangian}. More precisely, to us, a solution of the Hamilton--Jacobi equation will be a Lagrangian bisection $\lag$ satisfying $H^{ext}_{|\lag}=0$ and a generating function generating $\lag$. We go further because our cotangent groupoid structure
allows us to show that, under certain initial value conditions, the
aforementioned Lagrangian submanifold is horizontal, and then it happens
to be the image of a closed $1$-form. Thus, locally the graph of that $1$-form is the
differential of a function, $S$, which satisfies the Hamilton--Jacobi
equation. The result follows after applying basic stability theory results. Our construction works for arbitrary Hamiltonian functions
$H:\rad\rightarrow \mathbb{R}$. 

Before proving the results stated in the paragraph above, we generalize a classical and very popular result, the action given by a regular Lagrangian gives a type I solution of the Hamilton--Jacobi equation. This result establishes the link of the Hamilton-Jacobi theory with the variational integrators. 

\subsection{Existence of type I solutions for Lagrangian systems}

When the Hamiltonian system under consideration comes from a
(hyper-)regular Lagrangian function $L: TQ \to \mathbb{R}$, there exists a local solution to the Hamilton--Jacobi equation, given by
the action functional
\[
S(t,q^i,x^i)=\int_0^tL(\dot{c}(s))ds
\]
for $t$ sufficiently small and $(q^{i})$ sufficiently close to $(x^{i})$. Here, 
$c(t)$ is the unique curve satisfying the Euler--Lagrange
equations with the boundary conditions 
\[
\begin{array}{l}
c(0)=(q^i), \\ \noalign{\medskip}
c(t)=(x^i).
\end{array}
\]
The details of the construction can be found in the references
\cite{Arnoldmmcm,Goldstein,MarsdenRatiu}. For useful information about
Lagrangian dynamics on Lie algebroids we recommend
\cite{CoLeMaMaMa,LeMaMa,Ma,EduardoM,Weinstein} (see also \cite{MaMeSa,SaMeMa}). We extend this result to
our setting, following an auxiliary construction introduced in
\cite{MMMIII,Weinstein}. Along this section $G$ will be a Lie groupoid with source and target $\alpha$ and $\beta$ respectively and $AG$ its corresponding Lie algebroid with dual vector bundle $A^*G$. The vector bundle
$\pi:\ker(T\alpha)\rightarrow G$, where $\pi$ is the natural projection, has a Lie algebroid
structure given by

\begin{enumerate}
\item  The {\it anchor} is just the inclusion
  $i:\ker(T\alpha)\rightarrow TG$.
\item    The {\it Lie bracket} is just the restriction of the Lie
  bracket on vector fields. 
\end{enumerate}
This Lie algebroid is related to $\tau: AG\rightarrow M$ through the
mapping below, which happens to be a Lie algebroid morphism
\[
\begin{array}{rccl}
\Psi:&\ker(T\alpha)&\longrightarrow & AG \\ \noalign{\medskip}
& X_g&\rightarrow &\Psi(X_g)=dl_{g^{-1}}(X_g).
\end{array}
\]

Here, $l_{g^{-1}}: G^{\alpha(g)} \to G^{\beta(g)}$ is the left translation by $g^{-1}$ (see Appendix \ref{liegroupoids}). Note that $\Psi$ is an isomorphism on each fiber. Moreover, given a Lagrangian function on the Lie algebroid $L:AG\rightarrow \mathbb{R}$, we can consider the induced Lagrangian dynamics in $\ker(T\alpha)$ by $L\circ\Psi$. The relation between the two systems is given by the following theorem.
\begin{theorem}[See \cite{Weinstein}, Theorem $4.5$]  Let $\Psi:
  A\rightarrow B$ be a morphism of Lie algebroids, and let $L$ be a
  regular Lagrangian on $B$. If $\Psi$ is an isomorphism on each
  fiber, then the image under $\Psi$ of any solution of Lagrange's
  equations for $L\circ \Psi$ is a solution of Lagrange's equations for $L$.
\end{theorem}

The main observation is that the Lagrangian dynamics can be
interpreted as a ``parametrized'' version of the usual Lagrangian
theory on tangent bundles (see \cite{MMMIII}). An easy computation shows that for any initial condition $X_g\in \ker(T\alpha)$ the source does not evolve, that is
$\alpha\circ \pi(\varphi_t^{X_{ L\circ \Psi}}(X_g))=\alpha(g)$, where
$\varphi_t^{X_{L\circ\Psi}}$ is the flow of the vector field
$X_{L\circ\Psi}$ whose trajectories are the solutions of the
Euler--Lagrange equations for $L \circ \Psi$. That
explains why we can think about the Lagrangian system $(\ker(T\alpha),\
L \circ\Psi)$ as a ``bundle of Lagrangian systems''
$(T\alpha^{-1}(m),(L\circ \Psi)_{|\alpha^{-1}(m)})$ parametrized by
$m\in M$. From now on we use $L_m=(L \circ \Psi)_{|\alpha^{-1}(m)}$. Of course, given $X_g$ the flow produced by the two
Lagrangian systems,  $\varphi_t^{X_{L\circ \Psi}}$ and $\varphi_t^{X_{L_m}}$ with initial condition $X_g$ coincides, and that
concludes our observation. We can reduce now our problem to the
classical tangent bundle case. Given $g$, close enough to the identity
$\epsilon(\alpha(g))$ and a time $t$ small enough,  there exists an initial condition $X\in
\ker(T\alpha)_{\epsilon(\alpha(g))}=T\alpha^{-1}(\alpha(g))$ such that
$\varphi_t^{X_{L\circ \Psi}}(X)=g$, just by
direct application of the classical result on tangent bundles, to the
system $(T\alpha^{-1}(\alpha(g)),L_{\alpha(g)})$. Since these
constructions only hold locally, close enough in this section can be
understood with the usual euclidean norm in local coordinates, for instance. The classical theory also
says that we can define 
\begin{equation}\label{actionsolution}
S(t,g)=\int_0^tL(\varphi_s^{X_{L\circ \Psi}}(X)) ds
\end{equation}
in an appropriate open set, that we do not describe here, such that $g$ is close
to $\epsilon(\alpha(g))$ (for more details, see \cite{MMMIII}).
\begin{remark} The classical situation corresponds to the pair
  groupoid case. There $G=Q\times Q$ and $\ker(T\alpha)\equiv
  \{0\}\times TQ$ and $\Psi(0,X_{q})=(X_{q})$. So for any $q^1\in Q$ we
  have the Lagrangian system $L_{q^1}(X_{q^2})=L(X_{q^2})$. Thus, given $q^1$ close
  to $q^2$, which is equivalent to saying $(q^1,q^2)$ ($g$ in the
  previous discussion) close to $(q^1,q^1)$ ($\epsilon(\alpha(g))$ in
  the previous discussion) there exists a unique path satisfying the
  Euler--Lagrange equations joining them. Notice that given an initial
  condition $(0_{q^1},X_{q^2})$ the source, $q^1$, does not evolve.
\end{remark}

The Legendre transformations , $\mathbb{F}(L\circ \Psi)$ and
$\mathbb{F}L$, make the diagram below commutative
\[
\xymatrix{
\ker(T\alpha)^* \ar[r]^{\tilde{\beta}} & AG^* \\
\ker(T\alpha) \ar[r]^{\Psi} \ar[u]^{\mathbb{F}(L\circ \Psi)}\ar@{}[ur]|\circlearrowright& AG \ar[u]_{\mathbb{F}L}
}
\]
The corresponding Hamiltonians, say $H_{L\circ \Psi}:\ker(T\alpha)^*\rightarrow \mathbb{R}$ and $H_L:A^*G\rightarrow\mathbb{R}$
satisfy $H_{L\circ \Psi}= H_L\circ\tilde{\beta}$. Once more, direct application of the tangent bundle case shows that
$S$ satisfies
\[
\displaystyle\frac{\partial S}{\partial t}+H_L(\tilde{\beta}\circ
\frac{\partial S}{\partial g})=0.
\]
\begin{remark} The previous construction shows the link of our theory
  with \textbf{variational integrators}, in regard of equation \eqref{actionsolution}. The Hamilton--Jacobi equation can be
  understood in this way as a computation of the exact discrete Lagrangian (see \cite{MMMIII}).
\end{remark}

\begin{remark}
The same proof with the obvious modifications holds in the
time-dependent case. 
\end{remark}
%%%%%%%%%%%%%%%%%%%%%%%%%%%%%%%%%%%%%%%%%%%%%%%%%%%%%%%%%%%%%%%%%%%%%%%%%%%

\subsection{Method of Characteristics}\label{characteristics}
We give the procedure below to show existence of solutions of our
Hamilton--Jacobi equation. Actually, what we use is the classical
method of characteristics adapted to this situation. 

Assume that we are dealing with the general setting, that is, our
Hamiltonian is defined in the dual bundle of an integrable Lie
algebroid, $H:\mathbb{R}\times A^*G\rightarrow\mathbb{R}$. Let
$\lag_0=\widetilde\epsilon(A^*G)\subset T^*G$ be the Lagrangian submanifold given
by the units of the cotangent groupoid, and $\lag_1=\{(0,0)\}\times
\lag_0\subset T^*(\mathbb{R}\times G)$. Let $X_{H^{ext}}$ be the
Hamiltonian vector field of the associated extended Hamiltonian,
defined in previous sections and $\varphi^{X_{H^{ext}}}(t,p)$ its flow. Consider the immersion below
\[
\begin{array}{rccl}
i:& \mathbb{R}\times \lag_1&\longrightarrow & T^*(\mathbb{R}\times G)\\ \noalign{\medskip}
&(t,p)&\rightarrow & \varphi^{X_{H^{ext}}}(t,p),
\end{array}
\]
a simple computation shows that its image, $\lag$, is a Lagrangian submanifold,
embedded Lagrangian submanifold when restricted to small open
sets. The fact that, by construction, this Lagrangian submanifold is
tangent to $X_{H^{ext}}$ implies that $dH^{ext}$ restricted to this manifold
is $0$, then $H^{ext}$ is constant in any connected
component. For small $t$, it is obvious that $\lag_t=\varphi_t^{X_{H^{ext}}}(\lag_1)$ will be a bisection, al least restricting ourselves to an open subset of $\lag_t$. Summarizing, we get a Lagrangian bisection such that $H^{ext}_{|\lag}$ is constant.

\subsubsection{Local Solution}
In this section we give an abstract argument which justifies our
choice of local coordinates (see Section \ref{non-free}). We just sketch the proof, as it
is straightforward. Imagine that we are in the setting of the previous
Section \ref{characteristics}, and that there exist a projection,
which may be only locally defined, $\pi:T^*
  G\rightarrow B$ such that $\lag_0$ is horizontal. We are assuming that
  $\pi:T^*G\rightarrow B$ can be, at least locally, identified with a
  cotangent bundle, see Section \ref{non-free} for an illustration. Then, maybe restricting to an open set of $\lag$, for
  $t\in \mathbb{R}$ small enough, $\lag$ is horizontal for $(\id_{\mathbb{R}},\pi):T^*(\mathbb{R}\times
  G)\rightarrow \mathbb{R}\times B$, defined as $(\id_{\mathbb{R}},\pi)(t,e,\mu_g)=(t,g)$, as an
  immediate application of stability of submersions. Since the
  argument is local, this is a consequence of the implicit function theorem. Making some
  abuse of notation, let us call this horizontal submanifold  (restriction of $\lag$ to an open set) $\lag$
as
  well. This implies that $\lag$ is a Lagrangian submanifold in
  $T^*(\mathbb{R}\times G)$ and so, locally again, there exists
  $S$, defined using the mappings introduced in previous sections, such that
  $\im(dS)=\lag$, and so we can guarantee local existence of
  solution. The projection $\pi$ has to be replaced by the
  appropriate mapping in each setting. For Lie algebras, action Lie
  algebroids and Atiyah algebroids, the existence of this map comes
  from Section \ref{non-free}.
	
	\begin{remark} When the Hamiltonian comes from
            an hyper-regular Lagrangian, then the Lagrangian
            submanifolds above, solutions of the Hamilton--Jacobi
            equation, are also horizontal with respect to the
            projection of the cotangent bundle
            $\pi_{\rgroupoid}:T^*(\rgroupoid)\rightarrow
            \rgroupoid$. A type I generating function for that
            Lagrangian submanifold is given by the action ($\int L$)
            as introduced in the previous section. Nonetheless, the
            approach above works for any Hamiltonian system and
            provides Lagrangian submanifolds that are solution of the
            Hamilton--Jacobi equation but are not horizontal with respect
            to the $\pi_{\rgroupoid}$ projection.
\end{remark}

%%%%%%%%%%%%%%%%%%%%%%%%%%%%%%%%%%%%%%%%%%%%%%%%%%%%%%%%%%%%%%%%%%%%%%%%%%%
%%       EXAMPLES
%%%%%%%%%%%%%%%%%%%%%%%%%%%%%%%%%%%%%%%%%%%%%%%%%%%%%%%%%%%%%%%%%%%%%%%%%%%

\section{Numerical Methods}\label{examples}

Here we develop our applications to numerical methods. The general
procedure is an improvement of the one outlined in
\cite{ChannellScovelII}. It works for Lie algebras, action Lie
algebroids and Atiyah Lie algebroids in a straightforward way, but it
can be generalized to other settings. We want to mention here that the
methods presented here can be improved in several ways, see item
\ref{improve} in Section \ref{conclusions}. Nonetheless, the numerical methods presented  already have the nice properties that one would expect from a Poisson integrator.

\subsection{General Procedure}

We describe here one of the most basic truncations that can be carried out in order to obtain numerical methods. However, we remark that other truncations can be elaborated based on our theoretical framework. In particular, {\it ad hoc} truncations can be obtained for any case, taking into account the features of the system under consideration.

\begin{enumerate}
\item Start with a Hamiltonian system in the dual of an integrable Lie
  algebroid $(\mathbb{R}\times A^*G,\Pi_\mathbb{R},H)$.

\item Construct coordinates where the identity can be generated, using
  Section \ref{non-free}, and obtain the Hamilton--Jacobi equation
\[
\displaystyle\frac{\partial S}{\partial t}+H(t,\tilde{\beta}\circ
\frac{\partial S}{\partial g})=0
\]
in those coordinates.

\item Approximate the solution taking the Taylor series in $t$ of $S$
  up to order $k$, $S(t,g)=\sum\limits_{i=0}^k S_i(t,g)t^i/i!+\mathcal{O}(t^{k+1})$, where
  $S_0$ is the generating function of the identity.

\item The equations for the $S_i$, $i\geq 1$ can be solved
  recursively. For instance, in the case of Lie algebras we get for
  the three first terms
\begin{itemize}
\item $\mathbf{S_0}(t, p^i)=0$.
\item $\mathbf{S_1}(t, p^i)+H(t,\displaystyle\frac{\partial S_0}{\partial
    p^i},p^i)=0$.
\item $\mathbf{S_2}(t, p^i)+\displaystyle\frac{\partial H}{\partial
    t}(t,\frac{\partial S_0}{\partial p^i},p^i)+\frac{\partial
    H}{\partial g^i}(t,\frac{\partial S_0}{\partial
    p^j},p^i)\frac{\partial S_1}{\partial p^i}=0$.
\end{itemize}
Each term can be obtained from the previous one by differentiating with respect
to $t$ and evaluating at $t=0$.

\item Collecting all the terms obtained up to order $k$,
  $S^k=\sum\limits_{i=0}^kS_{i}t^i/i!$, we get an approximation of the
  solution of the Hamilton--Jacobi equation. It is easy to see that the
  transformation induced implicitly $\mu^1_{m^1}\rightarrow
  \mu^2_{m^2}$, $\mu^i_{m^i}\in A^*G$ by 
\[
\begin{array}{cc}
\tilde{\alpha}\circ\displaystyle dS^k(t,g)=\mu^1_{m^1} & \tilde{\beta}\circ dS^k(t,g)=\mu^2_{m^2}
\end{array}
\] 
transforms the system to equilibrium up to order $k$. That is, it satisfies \[\displaystyle\frac{\partial S^k}{\partial t}+H(t,\tilde{\beta}\circ
\frac{\partial S^k}{\partial g})=\mathcal{O}(t^{k+1}).\]Application of Theorem \ref{hamiltoniantransformation} gives that the transformation is an approximation of order $k$ of the flow and so the
numerical method obtained by fixing a time-step, say $h$, is of order
$k$.
\end{enumerate}

We present now the results obtained after applying this scheme to
different situations.

\subsection{Examples}
\subsubsection{Rigid Body}

The rigid body is one of the classic benchmarks for Lie--Poisson integrators. We
show here the performance of the integrators described above. The configuration space for the rigid body (disregarding translations) is the Lie group $G=SO(3)$. The matrix $R\in SO(3)$ gives the configuration of the sphere as a rotation with respect to a reference configuration where its principal axes of inertia are aligned with the coordinate axes of an inertial system. Consider also a second system of coordinates fixed to the rotating body and aligned with the principal axes of inertia. We identify the Lie algebra $\mathfrak{so}(3)$ with $\mathbb{R}^3$ using the isomorphism $\hat{\cdot}\ \colon\ \mathbb{R}^3 \to \mathfrak{so}(3)$ given by
\[
\hat{x}=\begin{pmatrix}
  0&-x_3&x_2\\
x_3&0&-x_1\\
-x_2&x_1&0
\end{pmatrix}
\]
As is usual, the angular velocity in body coordinates is denoted by $\Omega$, and $\hat\Omega=R ^{-1}\dot R$.
The moment of inertia tensor in body coordinates is $\mathbb{I}=\operatorname{diag}(I_1,I_2,I_3)$.
Also, the body angular momentum is $\Pi=\mathbb{I}\Omega$, so in principal axes, $\Pi=(\Pi_1,\Pi_2,\Pi_3)=(I_1\Omega_1,I_2\Omega_2,I_3\Omega_3)$. The Hamiltonian in these variables is
\[
H= \frac{1}{2}\left( \frac{\Pi_1^2}{I_1}+\frac{\Pi_2^2}{I_2}+\frac{\Pi_3^2}{I_3}\right),
\]
which we regard as a function $H\ \colon\ \mathfrak{so}(3)^* \to
\mathbb{R}$. Note that if we regard $SO(3)$ as a Lie groupoid, then the dual Lie algebroid $A^*G$ is just $\mathfrak{so}^*(3)$. We use the Cayley map $\operatorname{cay}\ \colon\ \mathfrak{so}(3) \to SO(3)$,
\[
\operatorname{cay}(\hat\omega)=I_3+\frac{4}{4+\|\omega\|^{2}}\left(\hat\omega+\frac{\hat\omega^{2}}{2}\right),
\]
where $I_{3}$ is the $3\times 3$ identity matrix (see \cite{FerraroJimenezMartin}) to define local coordinates near the group identity. This gives a trivialization of the cotangent Lie groupoid $T^*SO(3)$ near the identity. The corresponding source and target maps, which map $T^*SO(3)$ into $\mathfrak{so}(3)^*$, become locally $\tilde\alpha_{\operatorname{cay}},\ \tilde\beta_{\operatorname{cay}}\ \colon\ \mathbb{R}^3\times (\mathbb{R}^3)^* \to \mathbb{R}^3$, and are given by 
\[
\tilde\alpha_{\operatorname{cay}}(x,y,z,p_x,p_y,p_z)= \left( \begin {array}{c}  \left( \frac{x^2}{4}+1 \right) { p_x}+
 \left(\frac{xy}{4}+ \frac{z}{2} \right) { p_y}+ \left( \frac{xz}{4}-\frac{y}{2}
 \right) { p_z}\\ \left( \frac{xy}{4}-\frac{z}{2} \right) {
 p_x}+ \left( \frac{y^2}{4}+1 \right) { p_y}+ \left( \frac{yz}{4}+\frac{x}{2}
 \right) { p_z}\\ \left( \frac{xz}{4}+\frac{y}{2} \right) {
 p_x}+ \left( \frac{yz}{4}-\frac{x}{2}
 \right) { p_y}+ \left( \frac{z^2}{4}+1 \right) { p_z}\end {array} \right),
\]
\[
\tilde\beta_{\operatorname{cay}}(x,y,z,p_x,p_y,p_z)= \left( \begin {array}{c}  \left( \frac{{x}^{2}}{4}+1 \right) { p_x}+
 \left( \frac{xy}{4}-\frac{z}{2} \right) { p_y}+ \left( \frac{xz}{4}+\frac{y}{2}
 \right) { p_z}\\ \left(\frac{xy}{4}+ \frac{z}{2} \right) {
 p_x}+ \left( \frac{y^2}{4}+1 \right) { p_y}+ \left( \frac{yz}{4}-\frac{x}{2} \right) { p_z}\\ \left( \frac{xz}{4}-\frac{y}{2}
 \right) { p_x}+ \left( \frac{yz}{4}+ \frac{x}{2} \right) { p_y}+ \left( \frac{z^2}{4} +1\right) { p_z}\end {array} \right). 
\]

The numerical method corresponding to a given truncation of the
function $S$ gives the evolution of $\Pi\in A^*G=\mathfrak{so}^* (3)$. We have run simulations using truncations of $S$ up to order 8, for a rigid body with $\mathbb{I}=(0.81,1,0.21)$. We have used $\Pi_0=(1.5,0.1,0)$ as the initial value, which makes the body rotate near the middle (unstable) axis.
The total run time was $T=5$, in which the body makes one ``tumbling''
motion, with decreasing values for the time-step $h$ (encoded as the
variable $t$ in the function $S$).  In Figure \ref{fig:rigidbody-errors} we plot the norm of the global error, as a distance in $\mathbb{R}^3$, with respect to a Runge-Kutta simulation (Matlab's ode45, variable step size) of the Euler equation $\dot\Pi=\Pi\times\Omega$ (\cite{Marsden}). Error values below $10^{-12}$ are not plotted due to inaccuracies caused by roundoff errors. For the rigid body, the terms with even orders in the expansion of $S$ are zero.
\begin{figure}[H]
\begin{center}
%\setlength\fboxsep{0pt}
%\setlength\fboxrule{0.5pt}
%\fbox{
\includegraphics[trim = 30mm 7mm 37mm 10mm, clip, scale=.52]{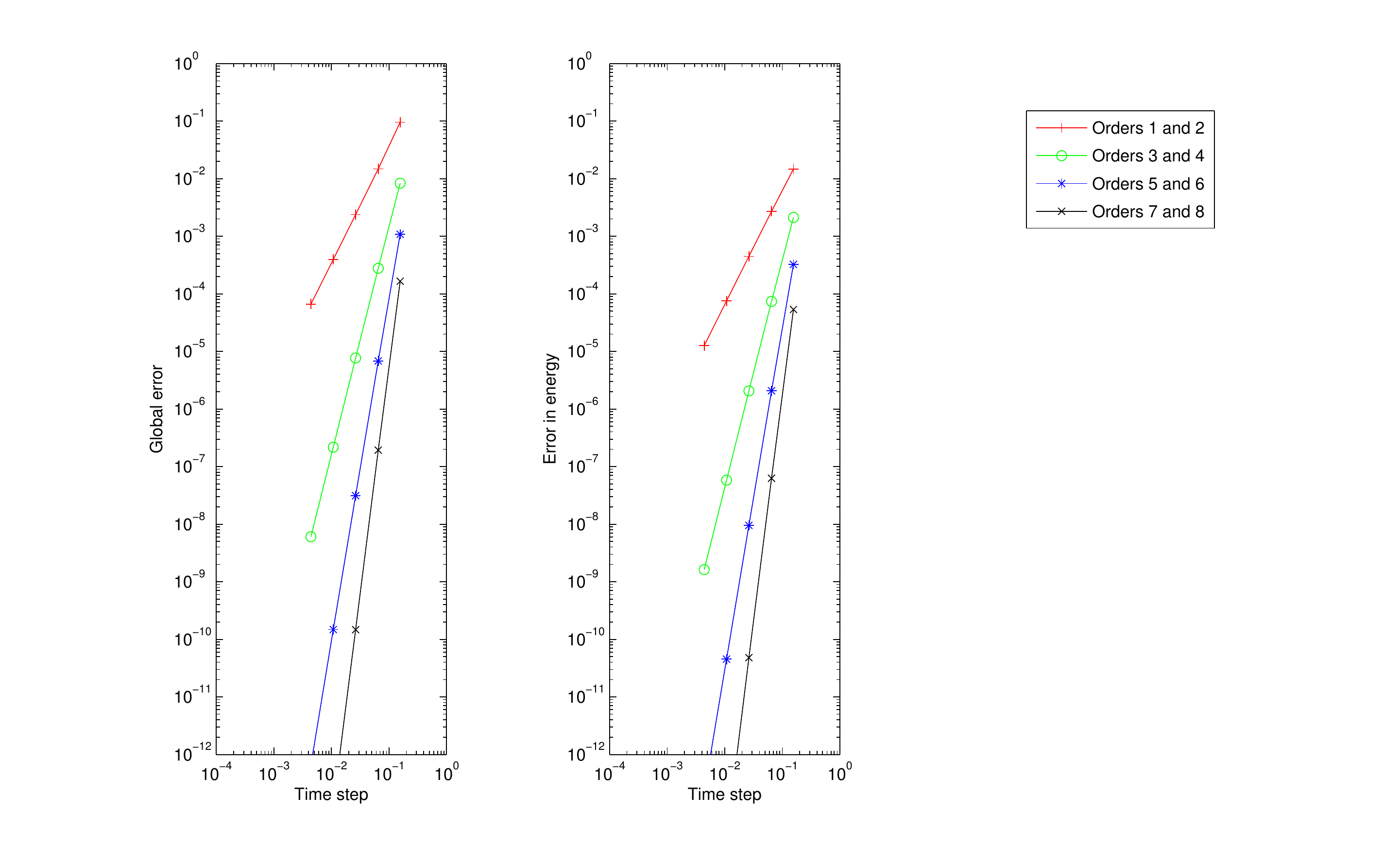}
%}
\end{center}
\setlength\abovecaptionskip{-5pt}
\caption{Global errors for the rigid body simulations after $5$ seconds, HJ method.}
\label{fig:rigidbody-errors}
\end{figure}

The next figures show the nice behavior of the energy for the orders $2$, $4$, $6$ and $8$, for $h=0.05$.  In the first three cases the behavior is essentially periodic, while in the last case the roundoff errors dominate and the evolution of the energy becomes a random walk.

\begin{figure}[H]
\begin{center}
%\setlength\fboxsep{0pt}
%\setlength\fboxrule{0.5pt}
%\fbox{
\includegraphics[trim = 10mm 5mm 10mm 0mm, clip, scale=.6]{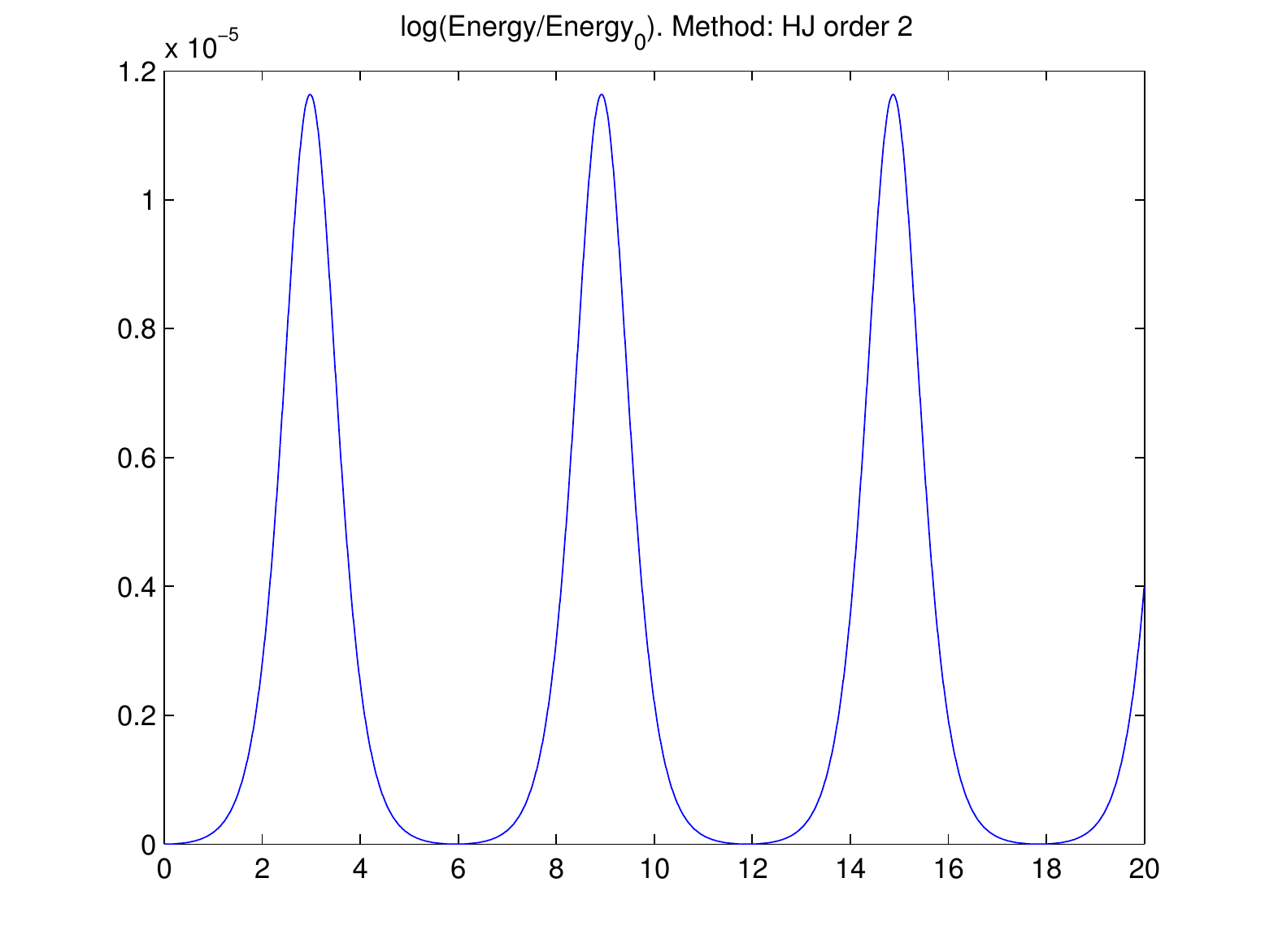}
\includegraphics[trim = 10mm 5mm 10mm 0mm, clip, scale=.6]{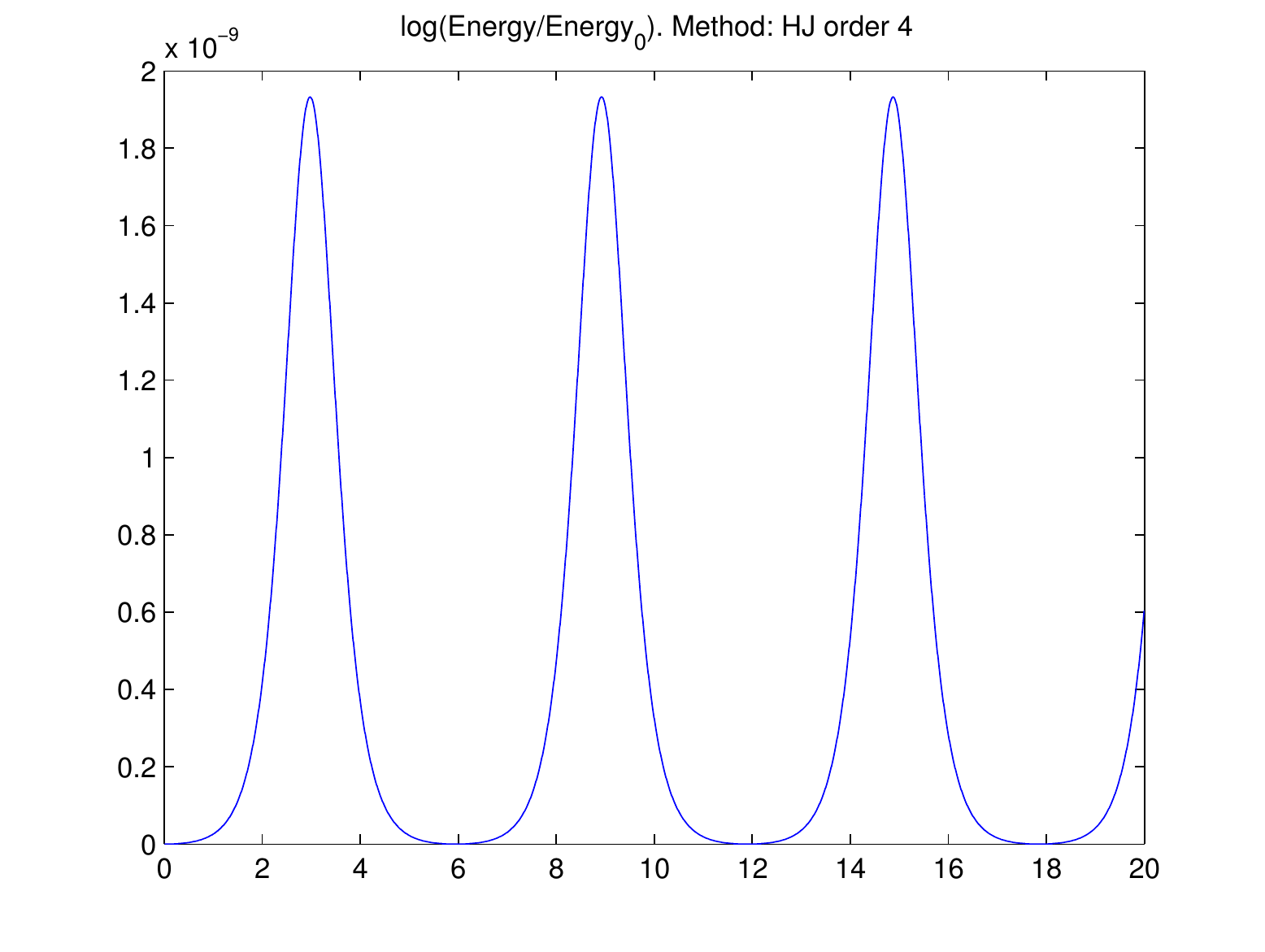}
%}
\end{center}
%\end{figure}
%
%\begin{figure}[H]
\begin{center}
%\setlength\fboxsep{0pt}
%\setlength\fboxrule{0.5pt}
%\fbox{
\includegraphics[trim = 10mm 5mm 10mm 0mm, clip, scale=.6]{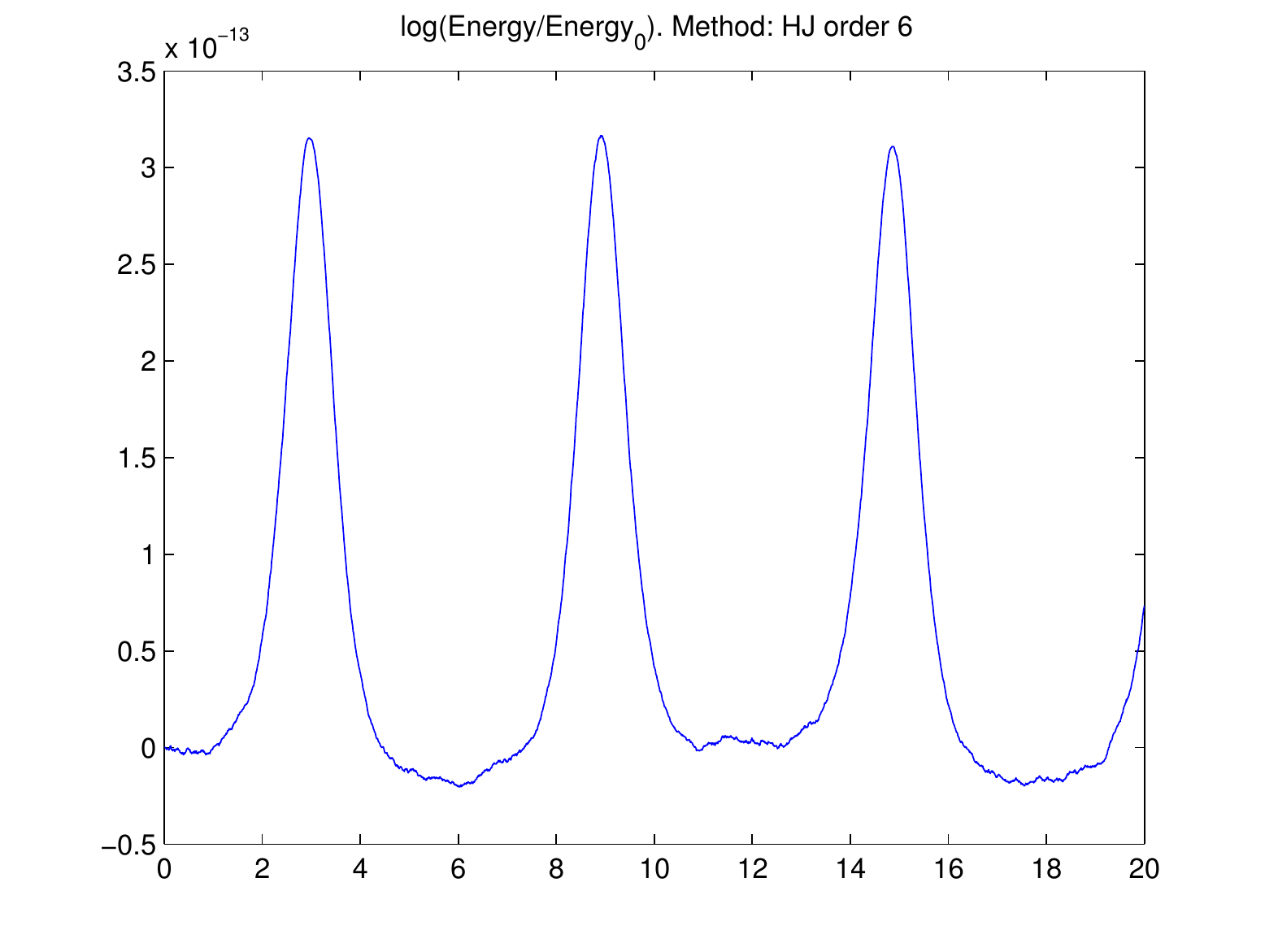}
\includegraphics[trim = 10mm 5mm 10mm 0mm, clip, scale=.6]{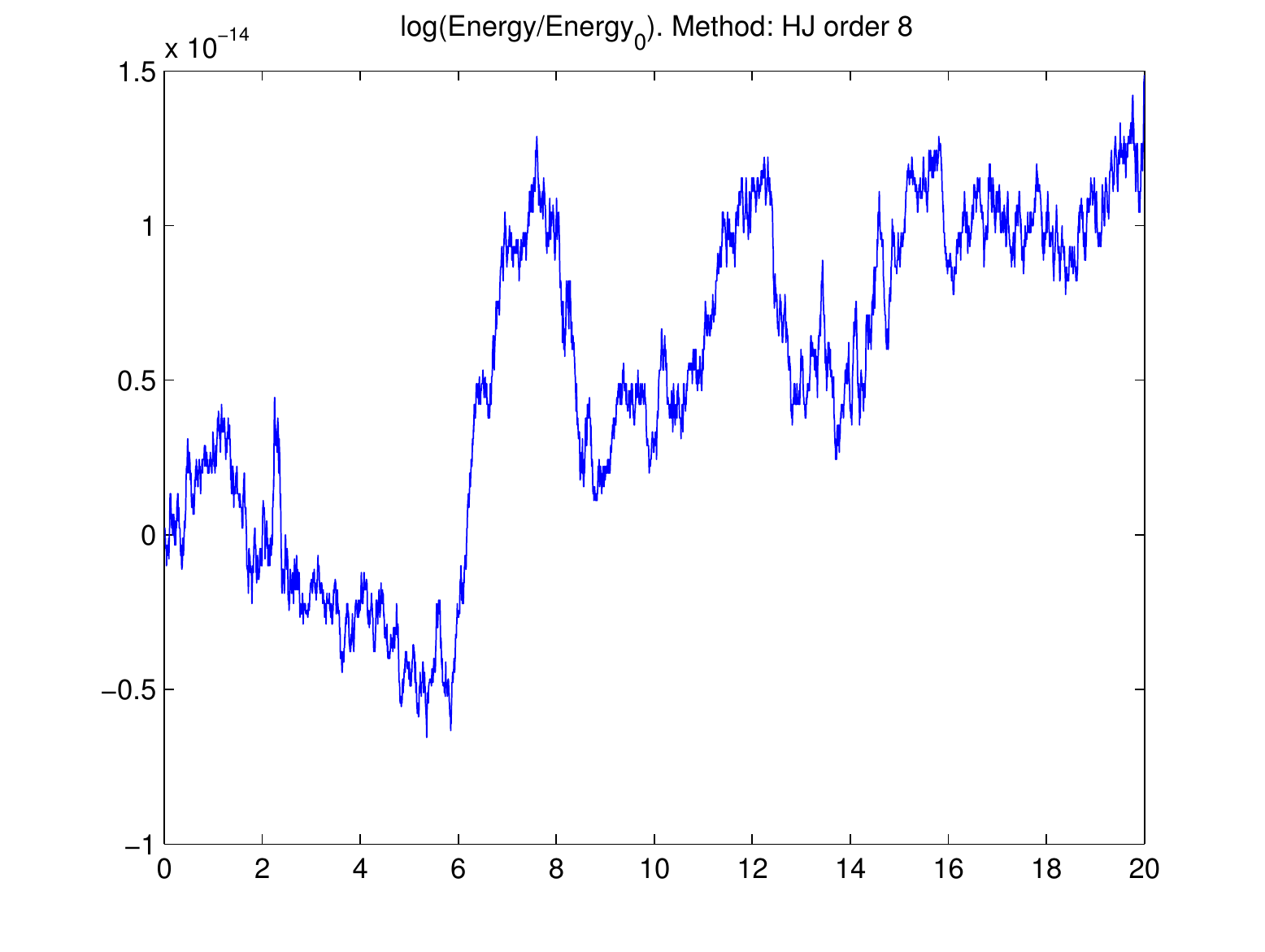}
%}
\end{center}
\setlength\abovecaptionskip{-10pt}
\caption{Energy conservation for the HJ method.}
\label{fig:HJ-energy}
\end{figure}

Notice that here we simulated the first $20$ seconds of the system, but the behavior of the energy is very stable, {\it i.e.}, the energy seems to oscillate periodically for all time, as shown in Figure \ref{energylong} where the first $25000$ seconds are simulated for the order $2$ method.

\begin{figure}[H]
\begin{center}
%\setlength\fboxsep{0pt}
%\setlength\fboxrule{0.5pt}
%\fbox{
\includegraphics[trim = 15mm 5mm 10mm 0mm, clip, scale=.6]{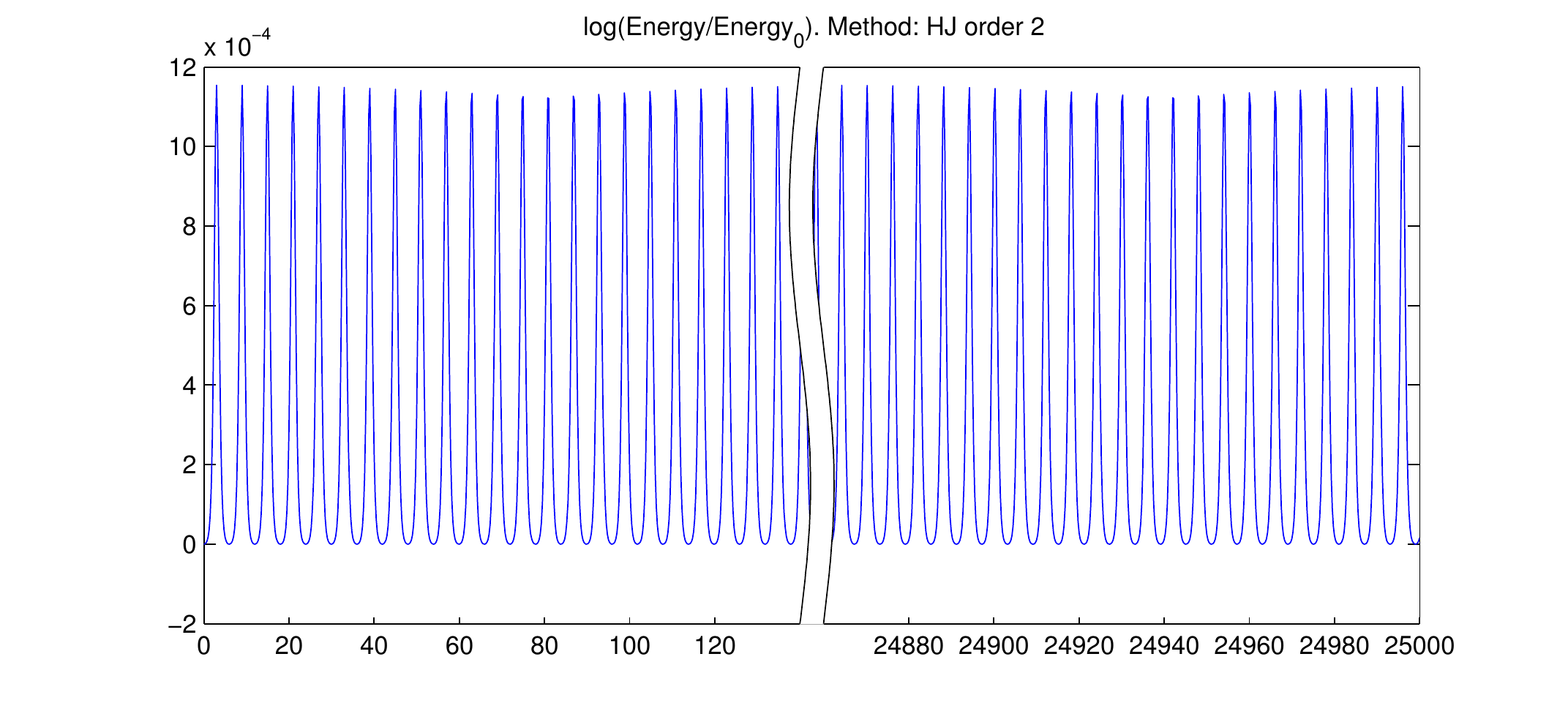}
%}
\end{center}
\setlength\abovecaptionskip{-10pt}
\caption{Long time energy conservation, order $2$ HJ method. }\label{energylong}
\end{figure}

The Casimirs have an exceptional conservation even for order $2$ methods, as shown in the figures below.

\begin{figure}[H]
\begin{center}
%\setlength\fboxsep{0pt}
%\setlength\fboxrule{0.5pt}
%\fbox{
\includegraphics[trim = 10mm 5mm 10mm 0mm, clip, scale=.6]{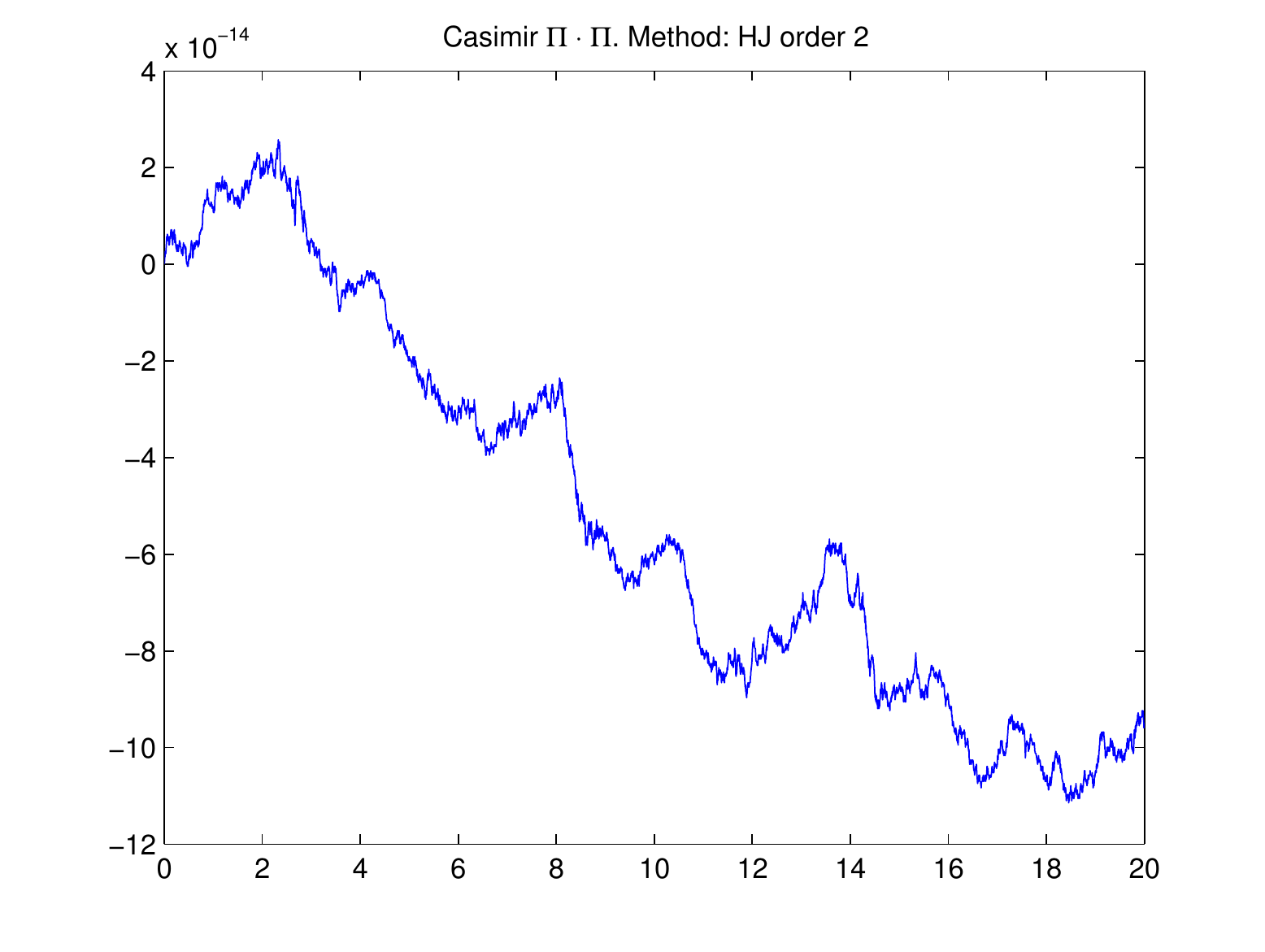}
\includegraphics[trim = 10mm 5mm 10mm 0mm, clip, scale=.6]{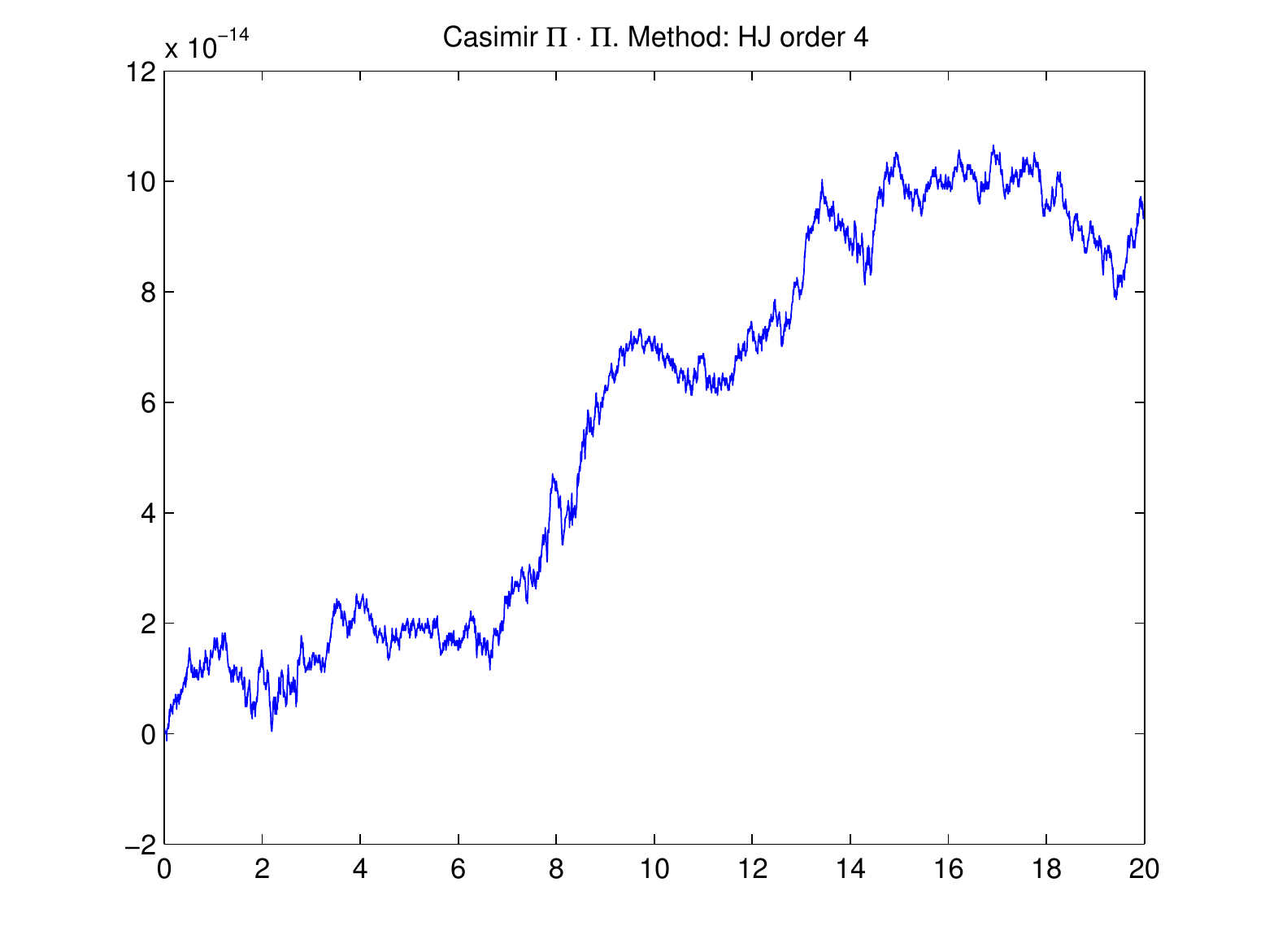}
%}
\end{center}
%\end{figure}
%
%\begin{figure}[H]
\begin{center}
%\setlength\fboxsep{0pt}
%\setlength\fboxrule{0.5pt}
%\fbox{
\includegraphics[trim = 10mm 5mm 10mm 0mm, clip, scale=.6]{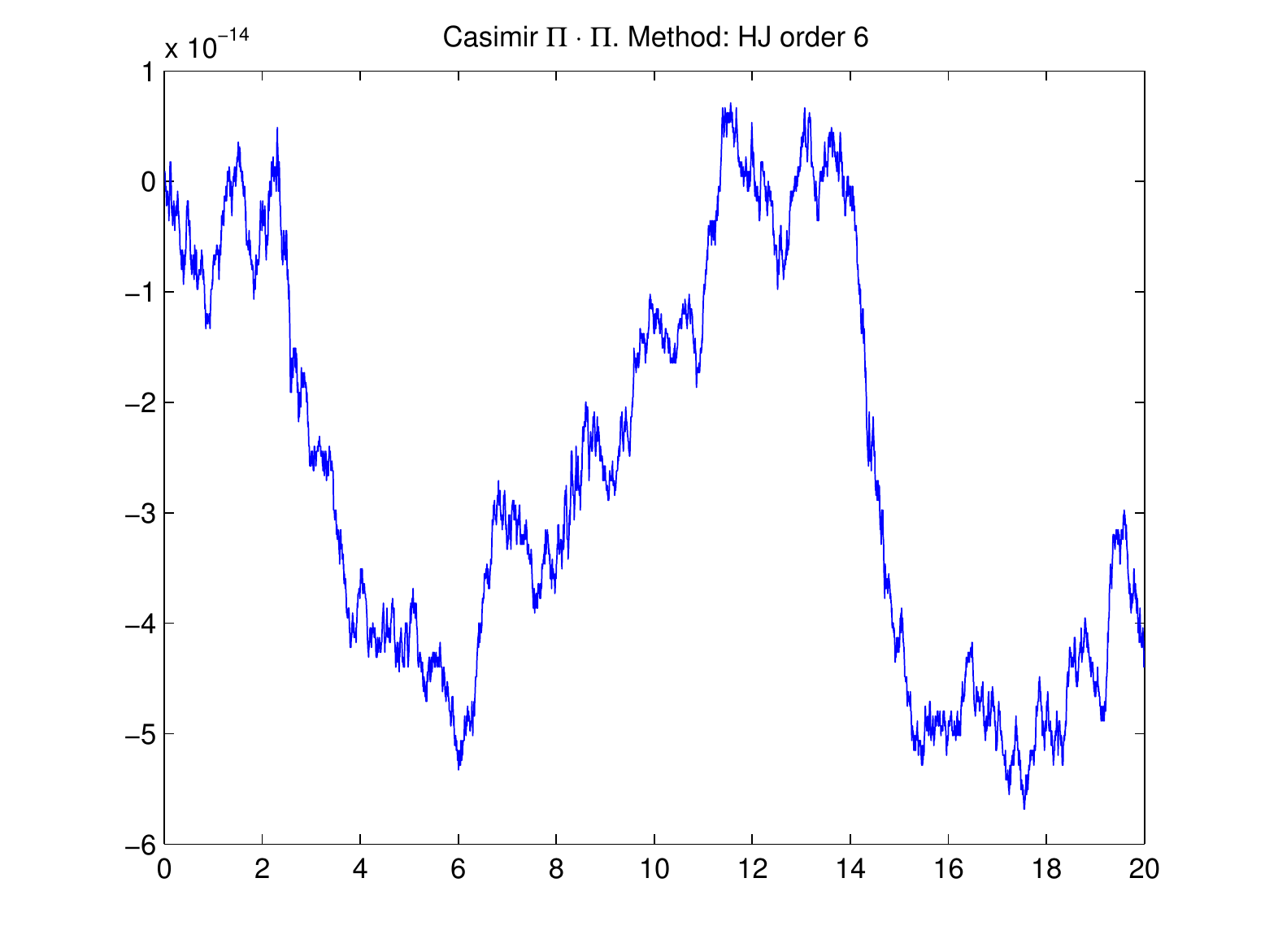}
\includegraphics[trim = 10mm 5mm 10mm 0mm, clip, scale=.6]{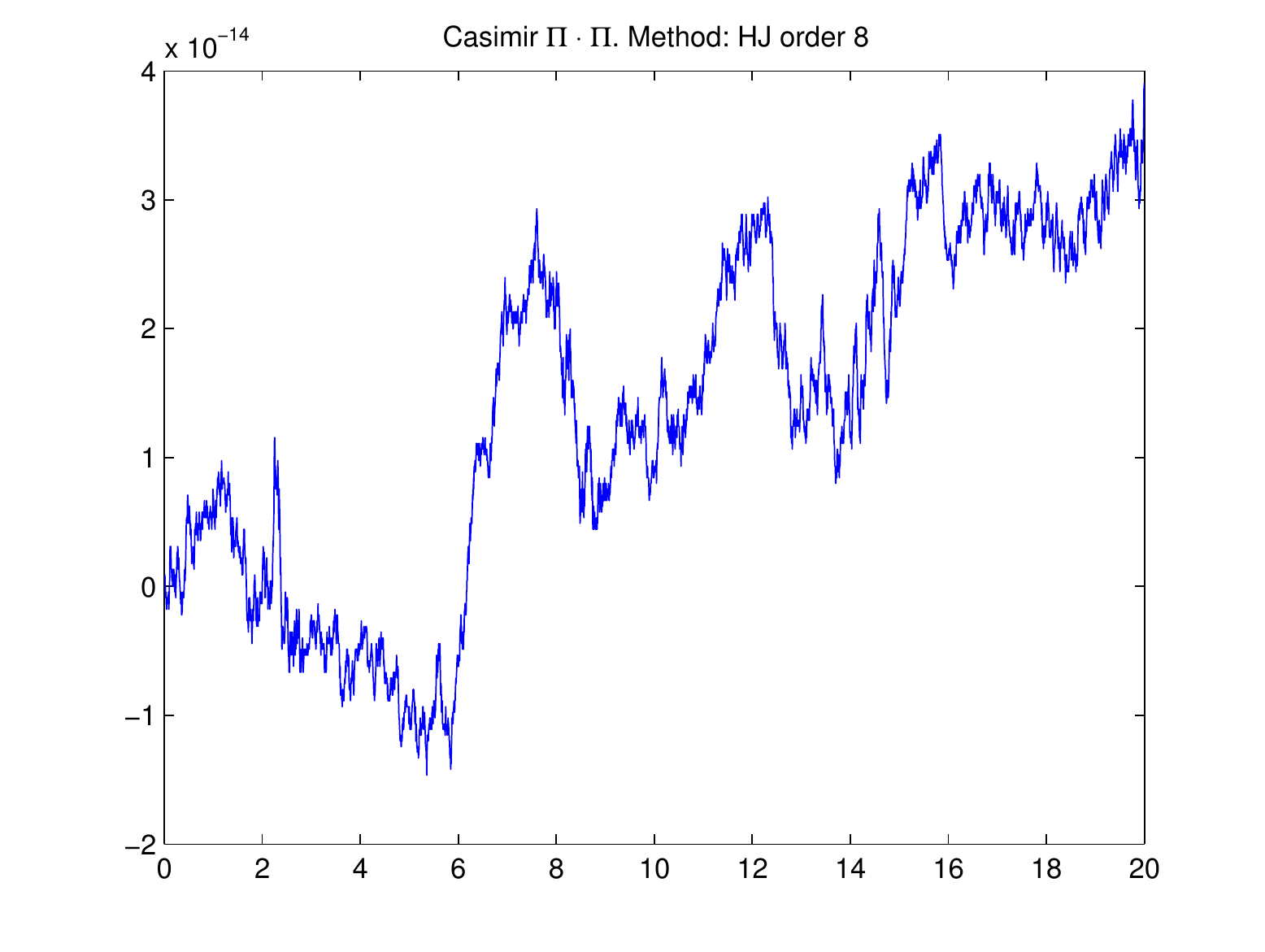}
%}
\end{center}
\setlength\abovecaptionskip{-10pt}
\caption{Casimir conservation, HJ method.}
\label{fig:HJ-Casimir}
\end{figure}

Compared to order $4$ Runge-Kutta, we observe that while the solution
of our method oscillates periodically 
the
Runge-Kutta method will get a bigger energy error after around $1000$ seconds.

 \begin{figure}[H]
\begin{center}
%\setlength\fboxsep{0pt}
%\setlength\fboxrule{0.5pt}
%\fbox{
\includegraphics[trim = 0mm 0mm 0mm 0mm, clip, scale=.45]{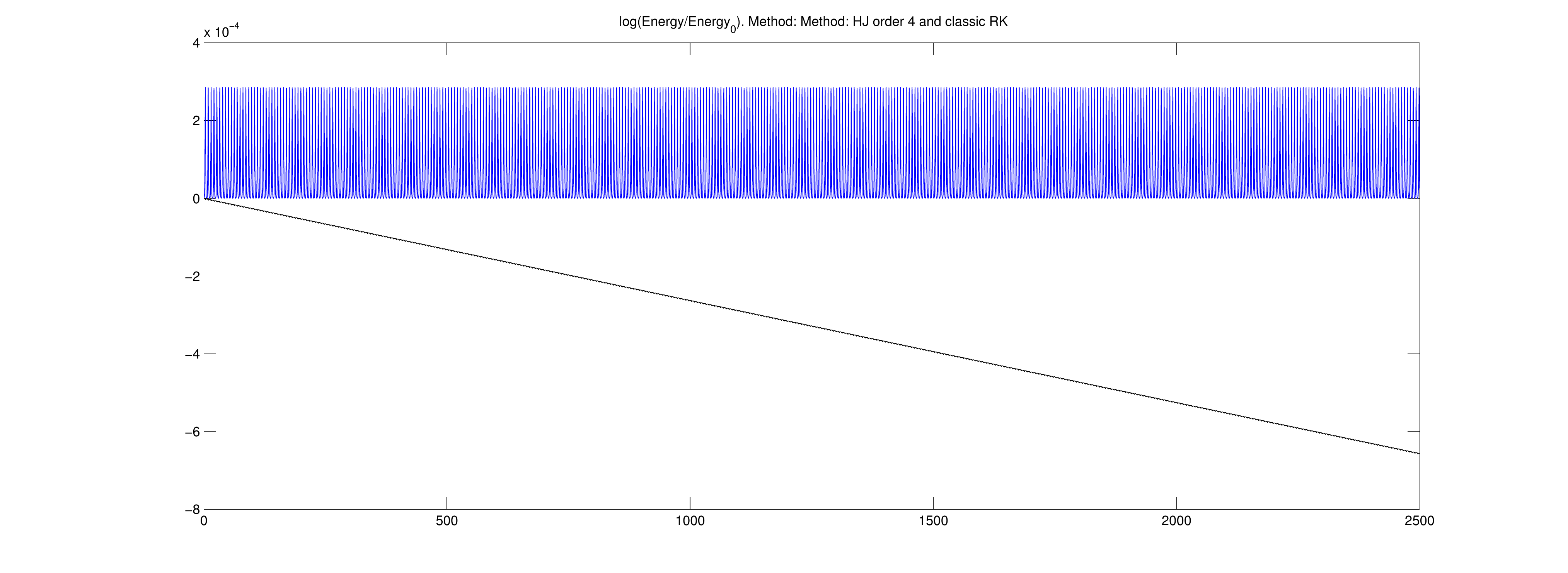}
%}
\end{center}
\setlength\abovecaptionskip{-10pt}
\caption{Energy conservation comparison. Order $4$ Runge-Kutta
  (below), and order $4$ HJ (above). $T=50$, $N=1000$, $h=T/N=0.05$.}
\end{figure}

We show the exceptional conservation of the Casimir for long time of our methods in the figure below, where we simulated the rigid body for $25000$ seconds.

\begin{figure}[H]
\begin{center}
%\setlength\fboxsep{0pt}
%\setlength\fboxrule{0.5pt}
%\fbox{
\includegraphics[trim = 0mm 0mm 0mm 0mm, clip, scale=.6]{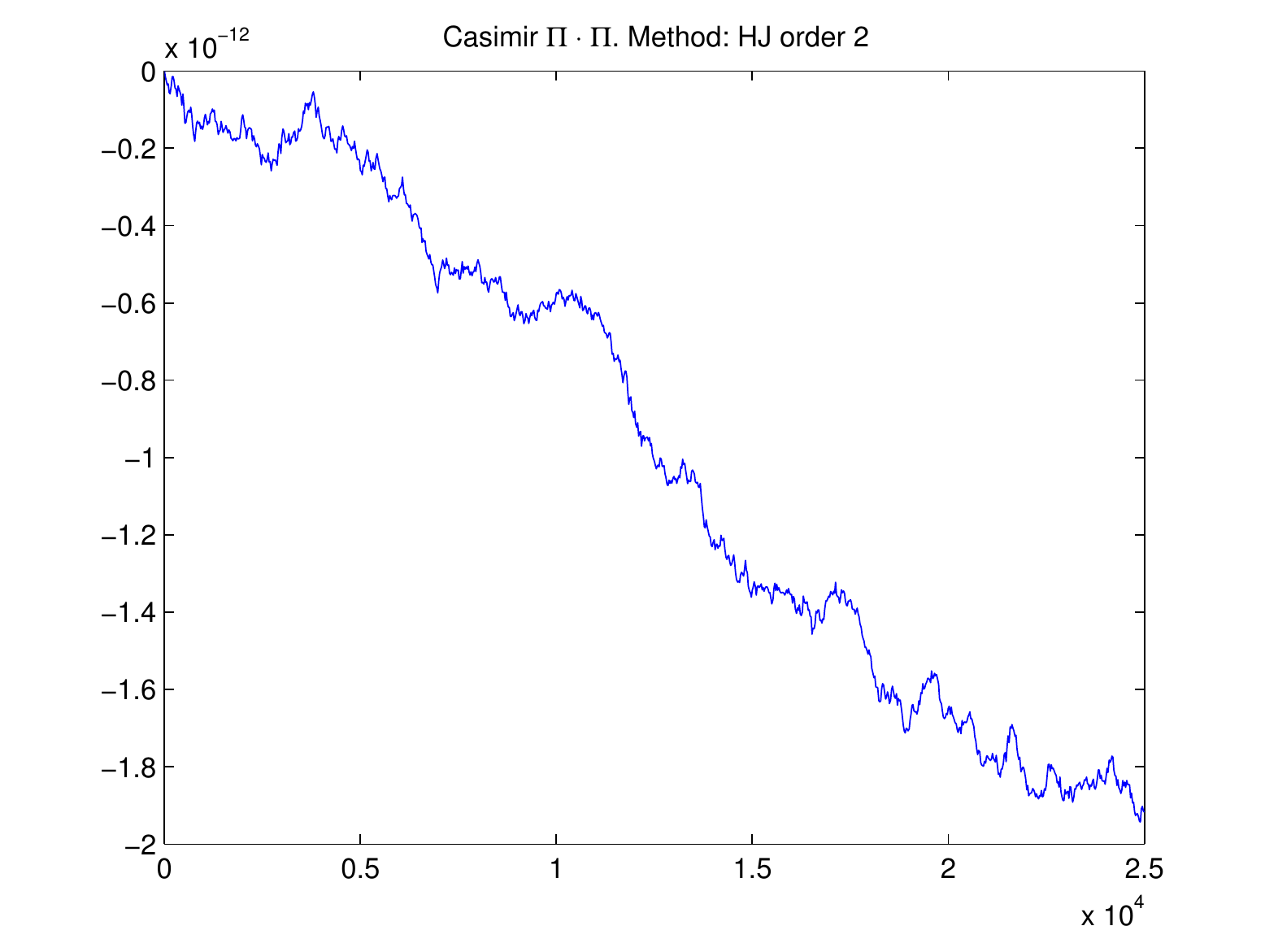}
%}
\end{center}
\setlength\abovecaptionskip{-10pt}
\caption{Casimir conservation, HJ order 2}
\end{figure}

%%%%%%%%%%%%%%%%%%%%%%%%%%%%%%%%%%%%%%%%%%%%%%%%%%
\subsubsection{Heavy Top}

As a concrete example of a system evolving on a transformation Lie algebroid
we consider  the heavy top. This example is modeled  on the action algebroid $\tau\ \colon\ S^2\times\mathfrak{so}(3)\to
S^2$ with Hamiltonian
\[
 H(\Gamma,\Pi)=\frac{1}{2}\Pi\cdot
 I^{-1}\Pi+mgl\Gamma\cdot \mathbf{e},
\]
where $\Pi\in \mathbb{R}^3\simeq\mathfrak{so}^*(3)$ is the angular
momentum, $\Gamma$ is the direction opposite to the gravity and
$\mathbf{e}$ is a unit vector in the direction from the fixed point to the
center of mass, all of them expressed in a frame fixed to the body.
The constants $m$, $g$ and $l$ are respectively the mass of the
body, the strength of the gravitational acceleration and the
distance from the fixed point to the center of mass. The matrix
$I$ is the inertia tensor of the body (see \cite{MMM,Ma}). In this case the Lie groupoid under consideration, $G$, is the action
groupoid $S^2\times SO(3)$, whose Lie algebroid is
$S^2\times \mathfrak{so}(3)$.
The source and target
maps are given by \[\begin{array}{c}\tilde\alpha_{\operatorname{cay}}(a_1,a_2,a_3,x,y,z,p_{a_1},p_{a_2},p_{a_3},p_x,p_y,p_z)
         \\ \noalign{\medskip}
=\left( \begin {array}{c} { a_1}\\{ a_2}
\\ { a_3}\\
  \left( \frac{{x}^{2}}{4}+1 \right) { p_x}+
 \left( \frac{xy}{4}-\frac{z}{2} \right) { p_y}+ \left( \frac{xz}{4}+\frac{y}{2}
 \right) { p_z}+a_2p_{a_3}-a_3p_{a_2}\\
 \left(\frac{xy}{4}+ \frac{z}{2} \right) {
 p_x}+ \left( \frac{y^2}{4}+1 \right) { p_y}+ \left( \frac{yz}{4}-\frac{x}{2} \right) { p_z}-a_1p_{a_3}+a_3p_{a_1}\\ 
\left( \frac{xz}{4}-\frac{y}{2}
 \right) { p_x}+ \left( \frac{yz}{4}+ \frac{x}{2} \right) { p_y}+ \left( \frac{z^2}{4} +1\right) { p_z}+a_1p_{a_2}-a_2p_{a_1}
\end {array} \right) \\ \noalign{\medskip}

\textrm{and} \\ \noalign{\medskip}\tilde\beta_{\operatorname{cay}}(a_1,a_2,a_3,
x,y,z,p_{a_1},p_{a_2},p_{a_3},p_x,p_y,p_z)\\ \noalign{\medskip}
=\left( \begin {array}{c}
{\frac { \left( {x}^{2}-{y}^{2}-{z}^{2}+4 \right) { a_1}}{{x}^{2}+{y
}^{2}+{z}^{2}+4}}+{\frac { \left( 2\,yx+4\,z \right) { a_2}}{{x}^{2}
+{y}^{2}+{z}^{2}+4}}+{\frac { \left( 2\,zx-4\,y \right) { a_3}}{{x}^
{2}+{y}^{2}+{z}^{2}+4}}\\
{\frac { \left( 2\,yx-4\,z \right) { a_1}}{{x}^{2}+{y}^{2}+{z}^{2}+4
}}+{\frac { \left( -{x}^{2}+{y}^{2}-{z}^{2}+4 \right) { a_2}}{{x}^{2}
+{y}^{2}+{z}^{2}+4}}+{\frac { \left( 2\,zy+4\,x \right) { a_3}}{{x}^
{2}+{y}^{2}+{z}^{2}+4}}
\\
{\frac { \left( 2\,zx+4\,y \right) { a_1}}{{x}^{2}+{y}^{2}+{z}^{2}+4
}}+{\frac { \left( 2\,zy-4\,x \right) { a_2}}{{x}^{2}+{y}^{2}+{z}^{2
}+4}}+{\frac { \left( -{x}^{2}-{y}^{2}+{z}^{2}+4 \right) { a_3}}{{x}^
{2}+{y}^{2}+{z}^{2}+4}}
\\
 \left( \frac{x^2}{4}+1 \right) { p_x}+
 \left(\frac{xy}{4}+ \frac{z}{2} \right) { p_y}+ \left( \frac{xz}{4}-\frac{y}{2}
 \right) { p_z}\\ \left( \frac{xy}{4}-\frac{z}{2} \right) {
 p_x}+ \left( \frac{y^2}{4}+1 \right) { p_y}+ \left( \frac{yz}{4}+\frac{x}{2}
 \right) { p_z}\\ \left( \frac{xz}{4}+\frac{y}{2} \right) {
 p_x}+ \left( \frac{yz}{4}-\frac{x}{2}
 \right) { p_y}+ \left( \frac{z^2}{4}+1 \right) { p_z}
\end {array} \right)
\end{array}\] We have run simulations using truncations of $S$ up to order 8 and
 the results are collected in Figure \ref{fig:heavytop-errors}. The
 initial condition was  $\left(\Gamma_0,\Pi_0)=((0.5, 0.5, -0.5)/ \| (0.5, 0.5, -0.5)
 \|,(0.1,-1,2)\right)$ and the parameters were 
$I = \operatorname{diag}(1,1.5,2)$,  $m = 0.1$, $g = 9.8$, $l = 0.2$ and $e = (0.1,0.2,0.5)/ \|(0.1,0.2,0.5)\|$.

\begin{figure}[H]
\begin{center}
%\setlength\fboxsep{0pt}
%\setlength\fboxrule{0.5pt}
%\fbox{
\includegraphics[trim = 25mm 7mm 37mm 10mm, clip, scale=.52]{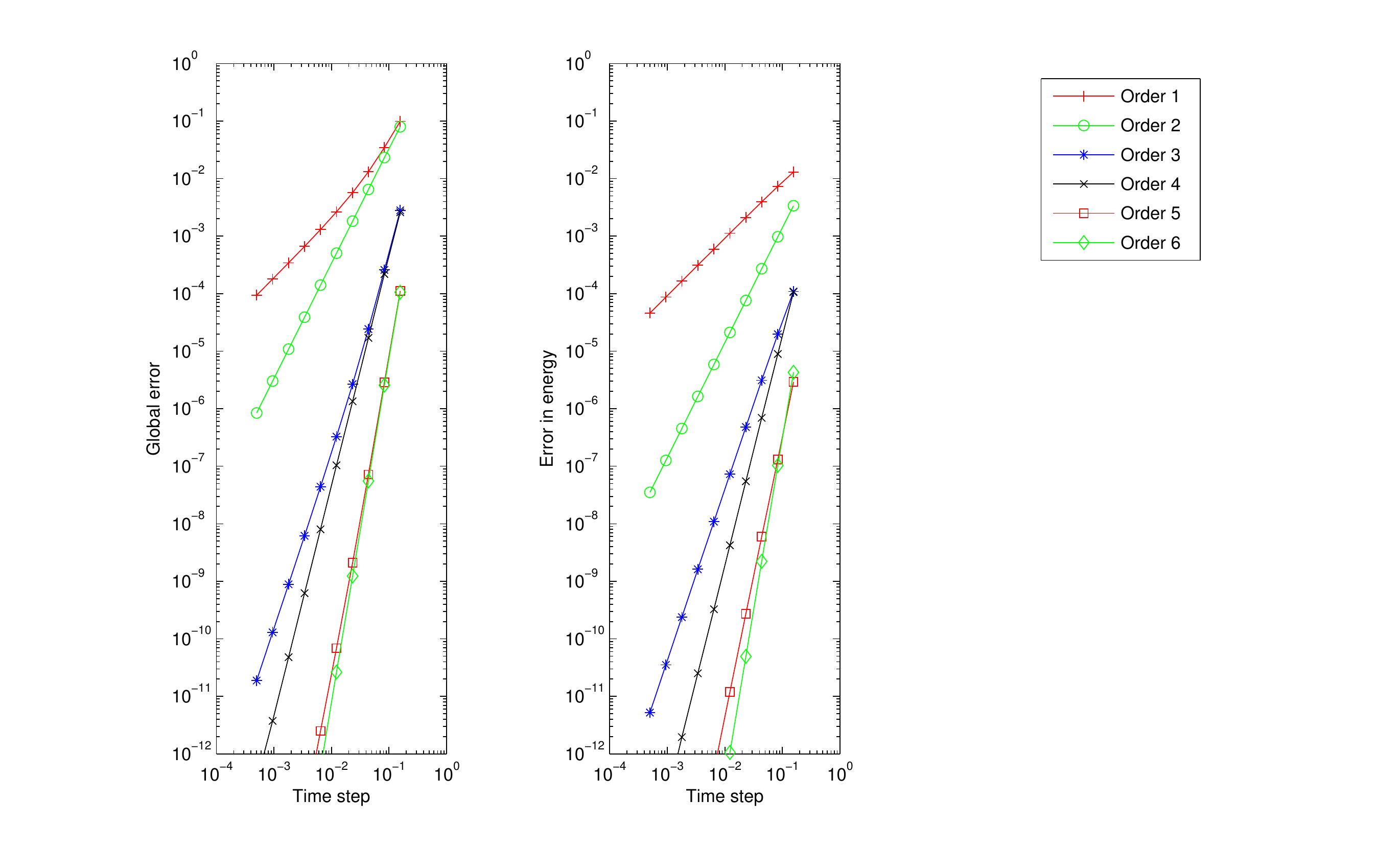}
%}
\end{center}
\setlength\abovecaptionskip{-10pt}
\caption{Global errors for the heavy top simulations, HJ method.}
\label{fig:heavytop-errors}
\end{figure}

A similar behavior in the conservation of energy is observed. 

\begin{figure}[H]
\begin{center}
%\setlength\fboxsep{0pt}
%\setlength\fboxrule{0.5pt}
%\fbox{
\includegraphics[trim = 0mm 0mm 0mm 0mm, clip, scale=.5]{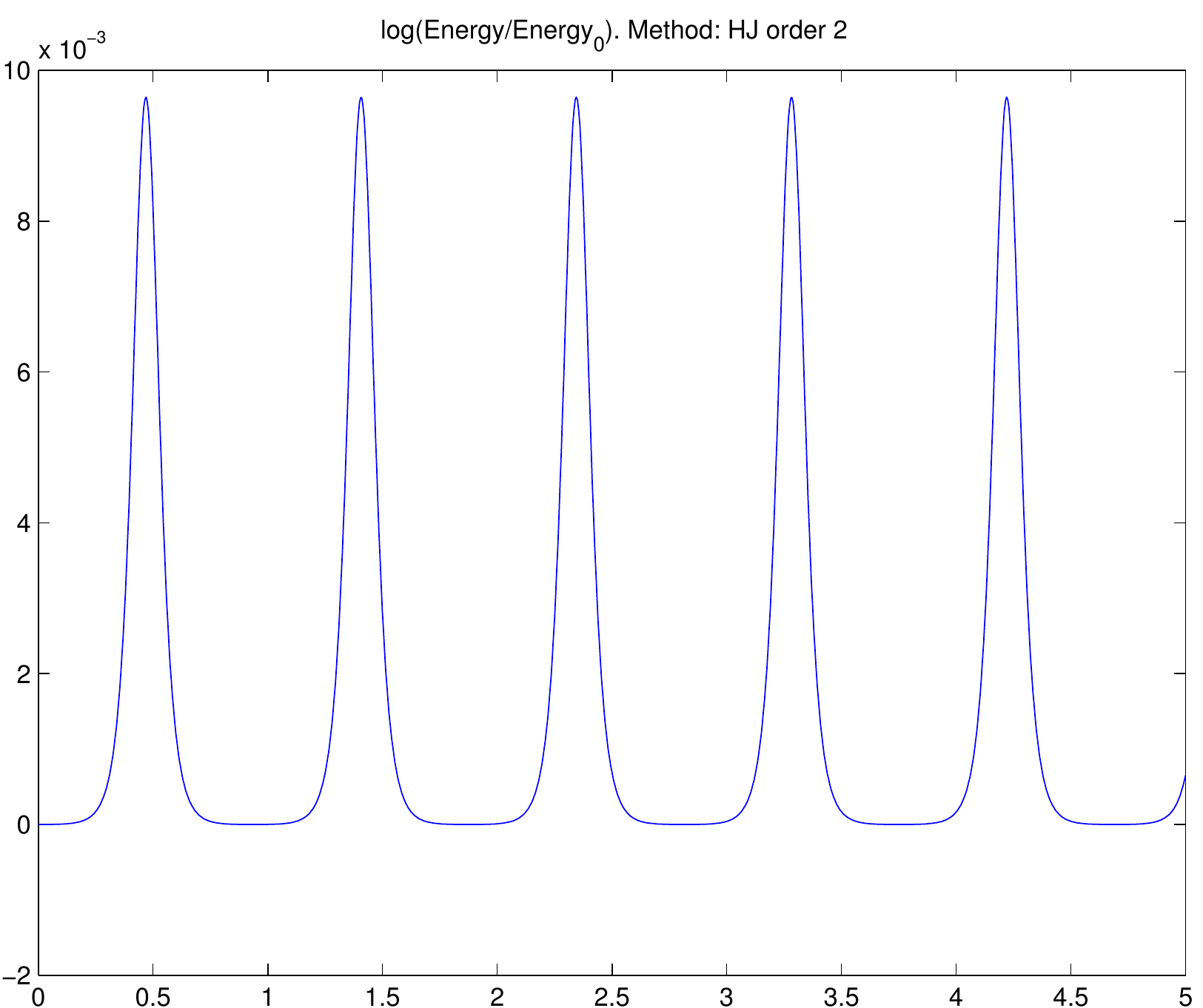}
\includegraphics[trim = 0mm 0mm 0mm 0mm, clip, scale=.5]{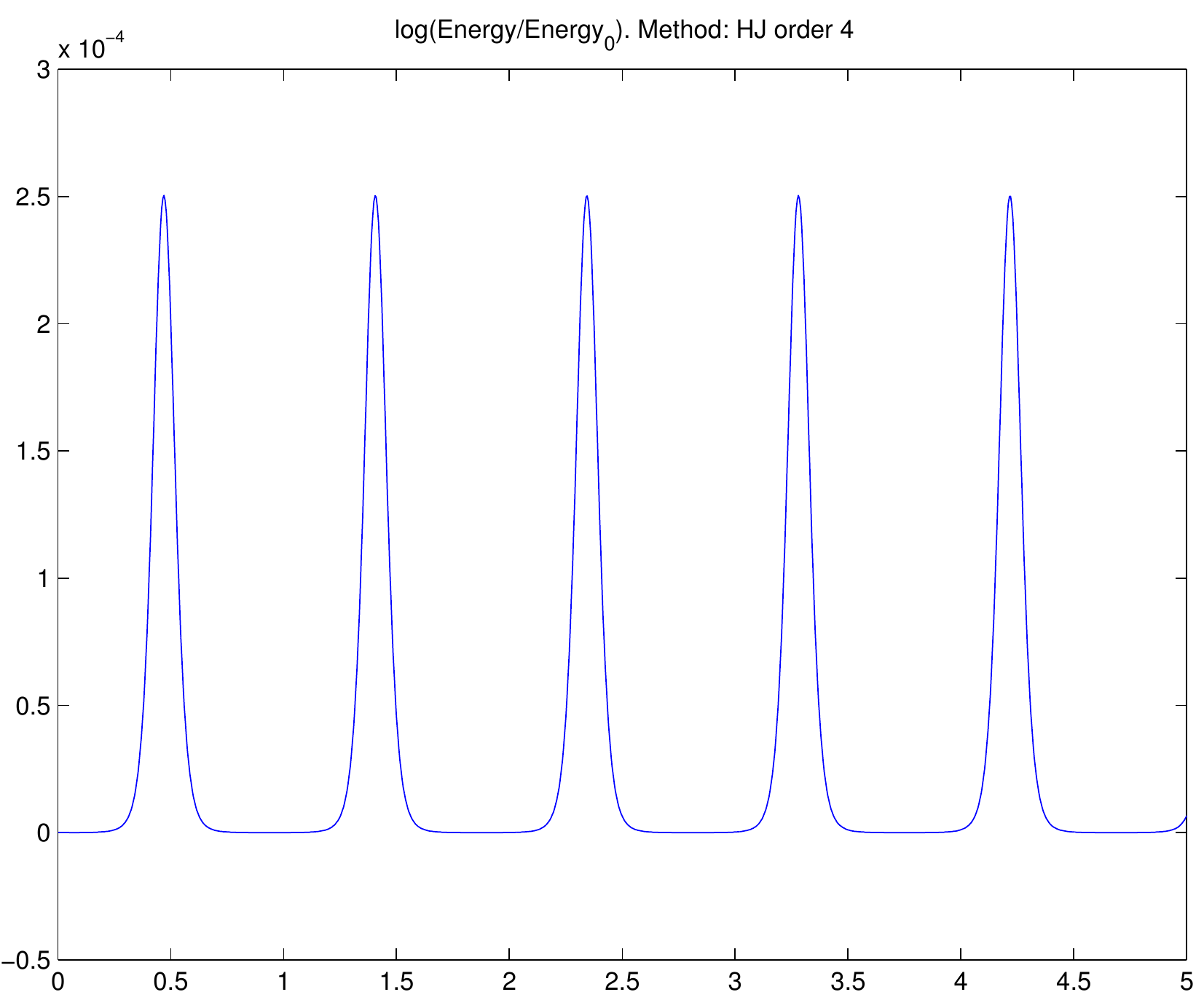}
%}
\end{center}
%\end{figure}
%
%\begin{figure}[H]
\begin{center}
%\setlength\fboxsep{0pt}
%\setlength\fboxrule{0.5pt}
%\fbox{
\includegraphics[trim = 0mm 0mm 0mm 0mm, clip, scale=.5]{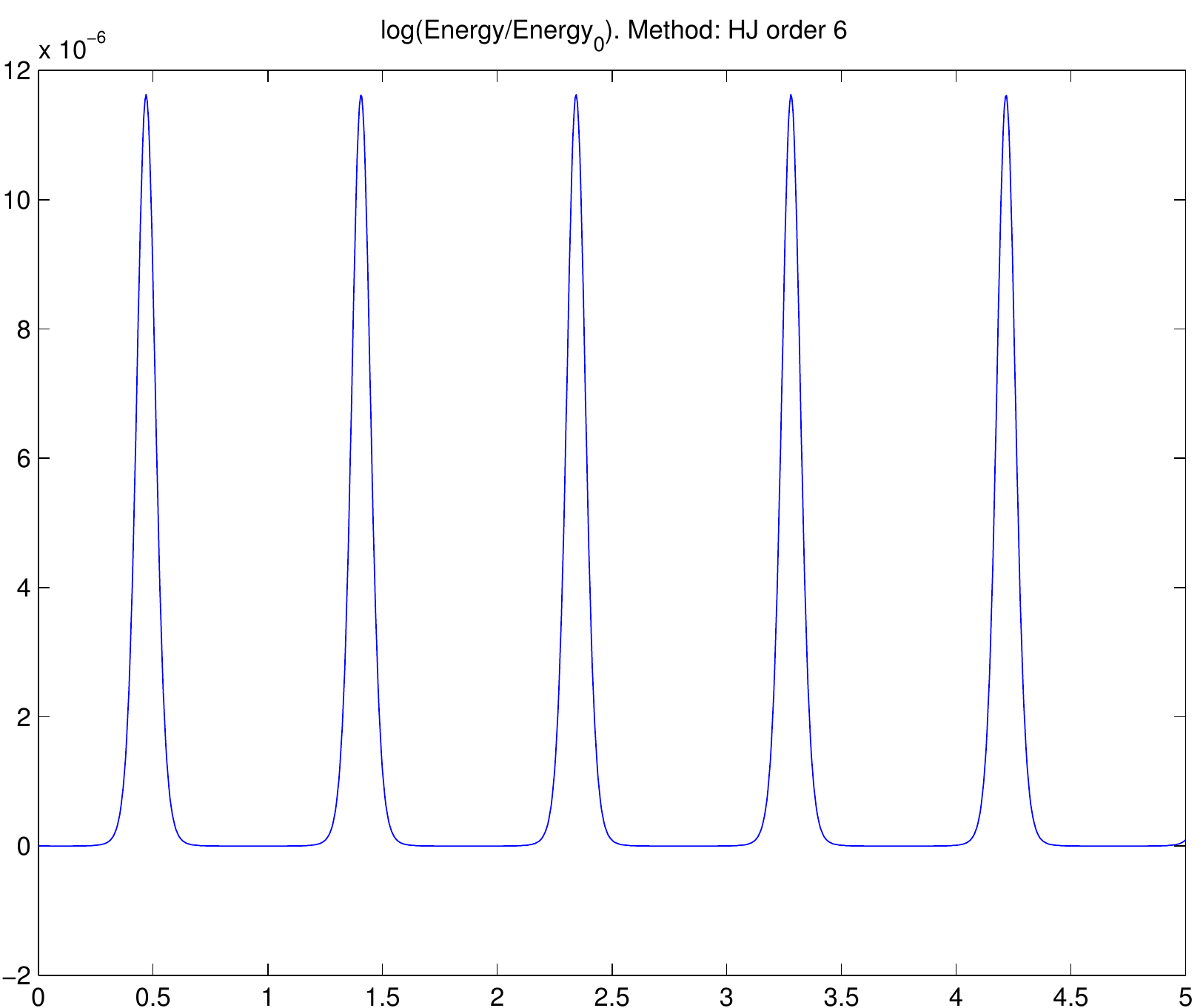}
\includegraphics[trim = 0mm 0mm 0mm 0mm, clip, scale=.5]{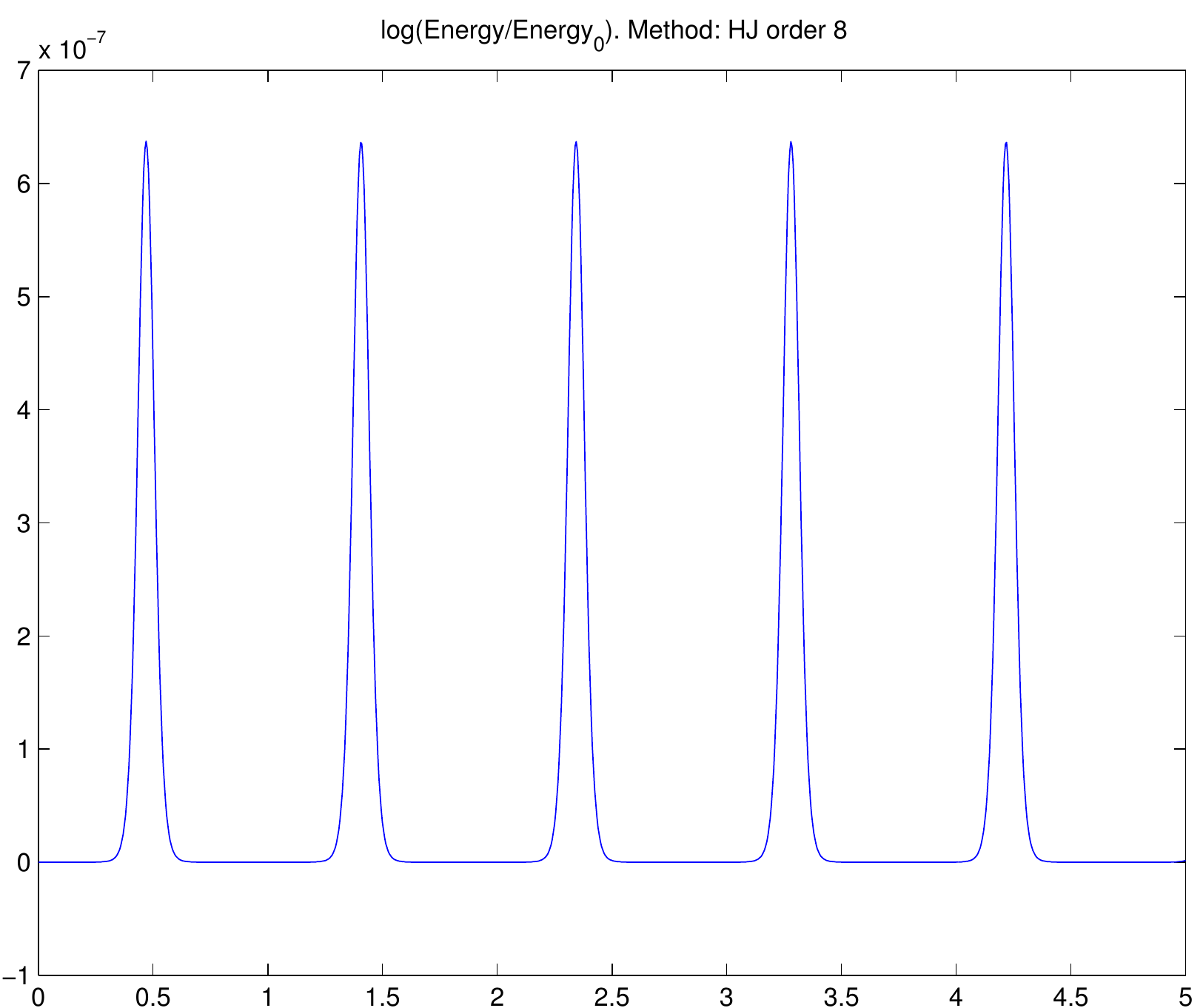}
%}
\end{center}
\setlength\abovecaptionskip{5pt}
\caption{Energy conservation, HJ method.}
\label{fig:HJ-energy-heavytop}
\end{figure}

The comparison with other non-geometric methods shows again a drift
for the conservation of Casimirs and energy. See the figure below for a comparison of
the conservation of Casimirs for the Hamilton-Jacobi and ode45 methods for the first $500$ seconds.
\begin{figure}[H]
\begin{center}
\includegraphics[trim = 0mm 0mm 0mm 0mm, clip, scale=.52]{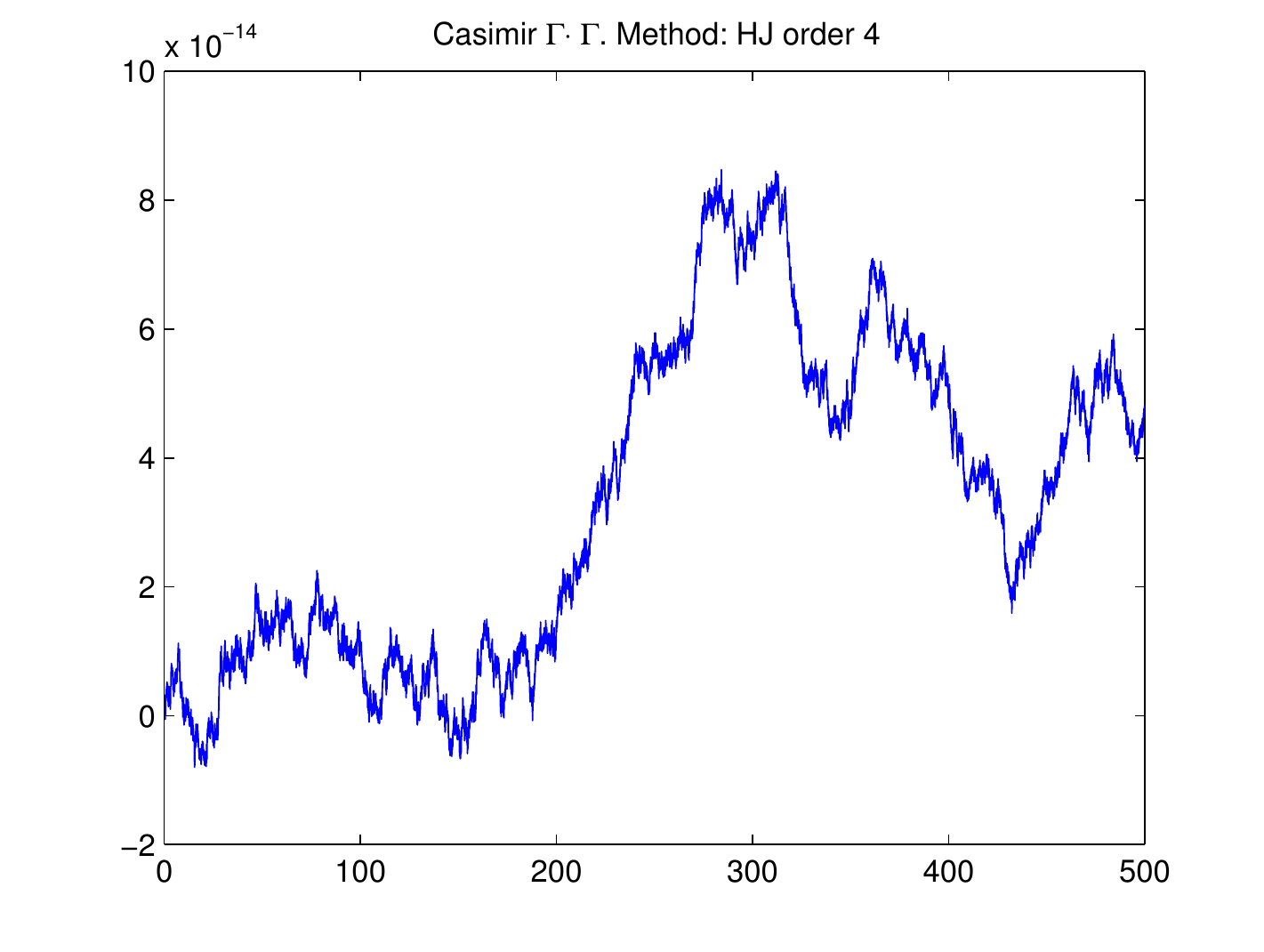}
\includegraphics[trim = 0mm 0mm 0mm 0mm, clip, scale=.52]{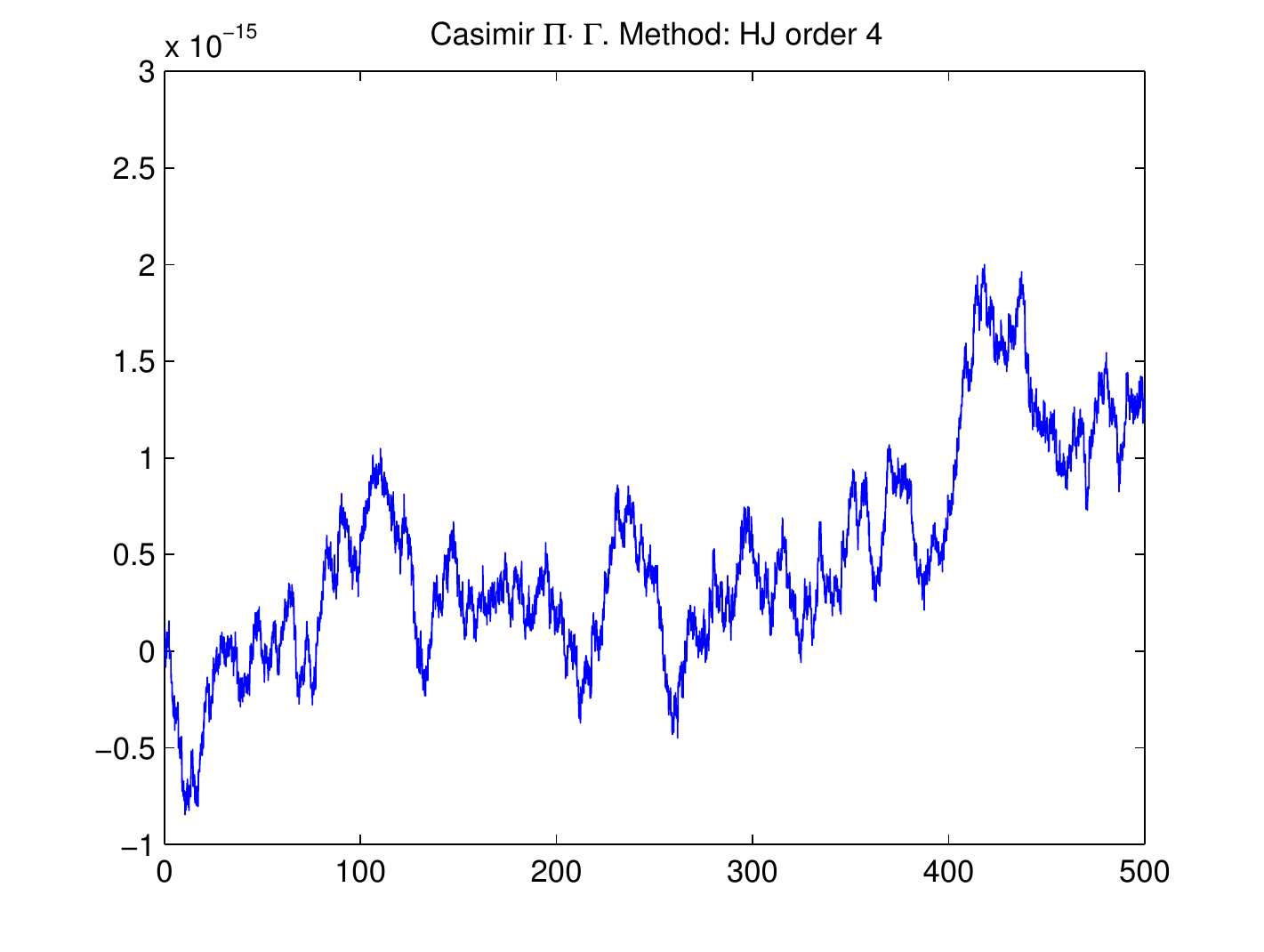}
\includegraphics[trim = 0mm 0mm 0mm 0mm, clip,scale=.52]{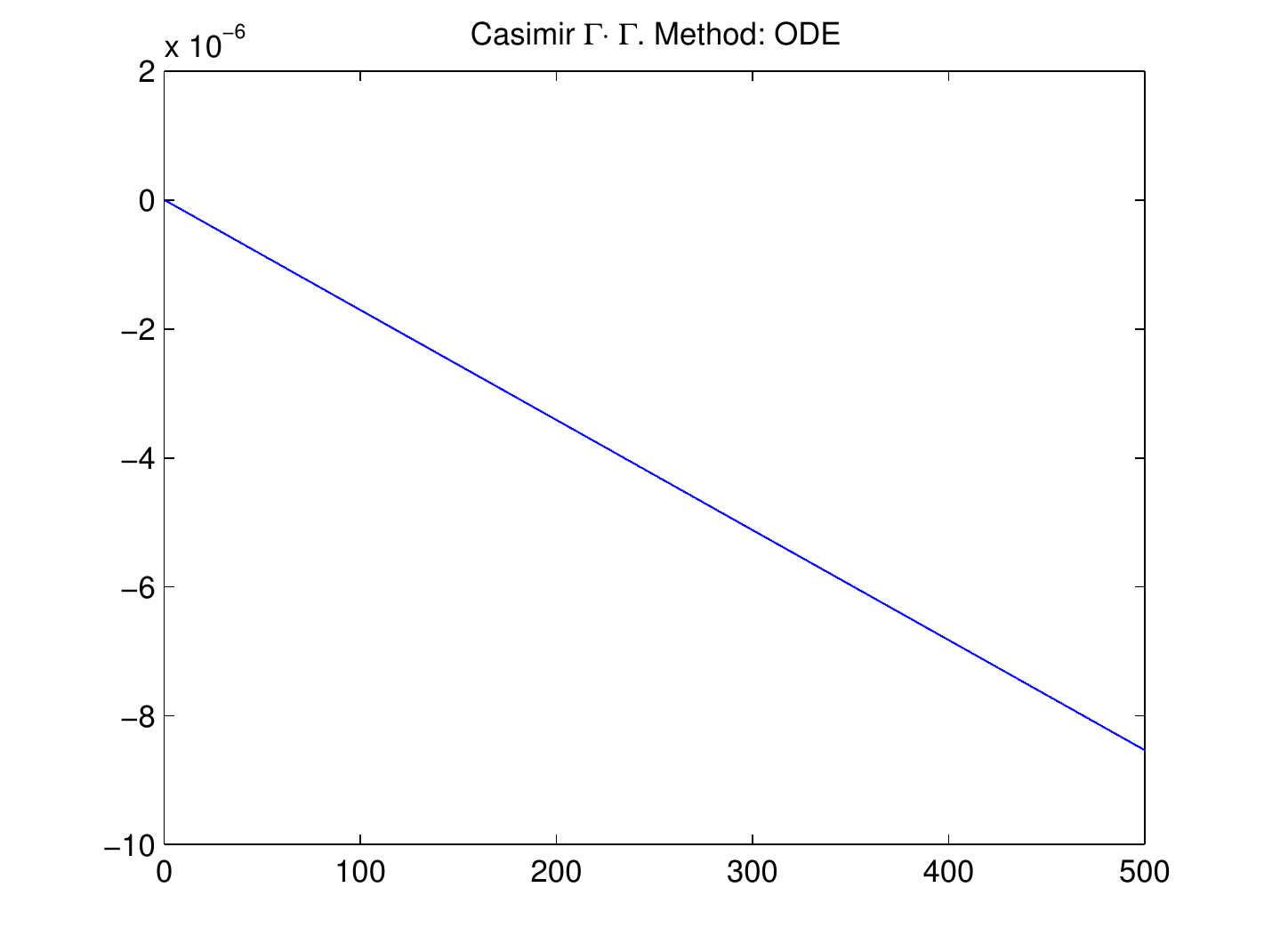}
\includegraphics[trim = 0mm 0mm 0mm 0mm, clip, scale=.52]{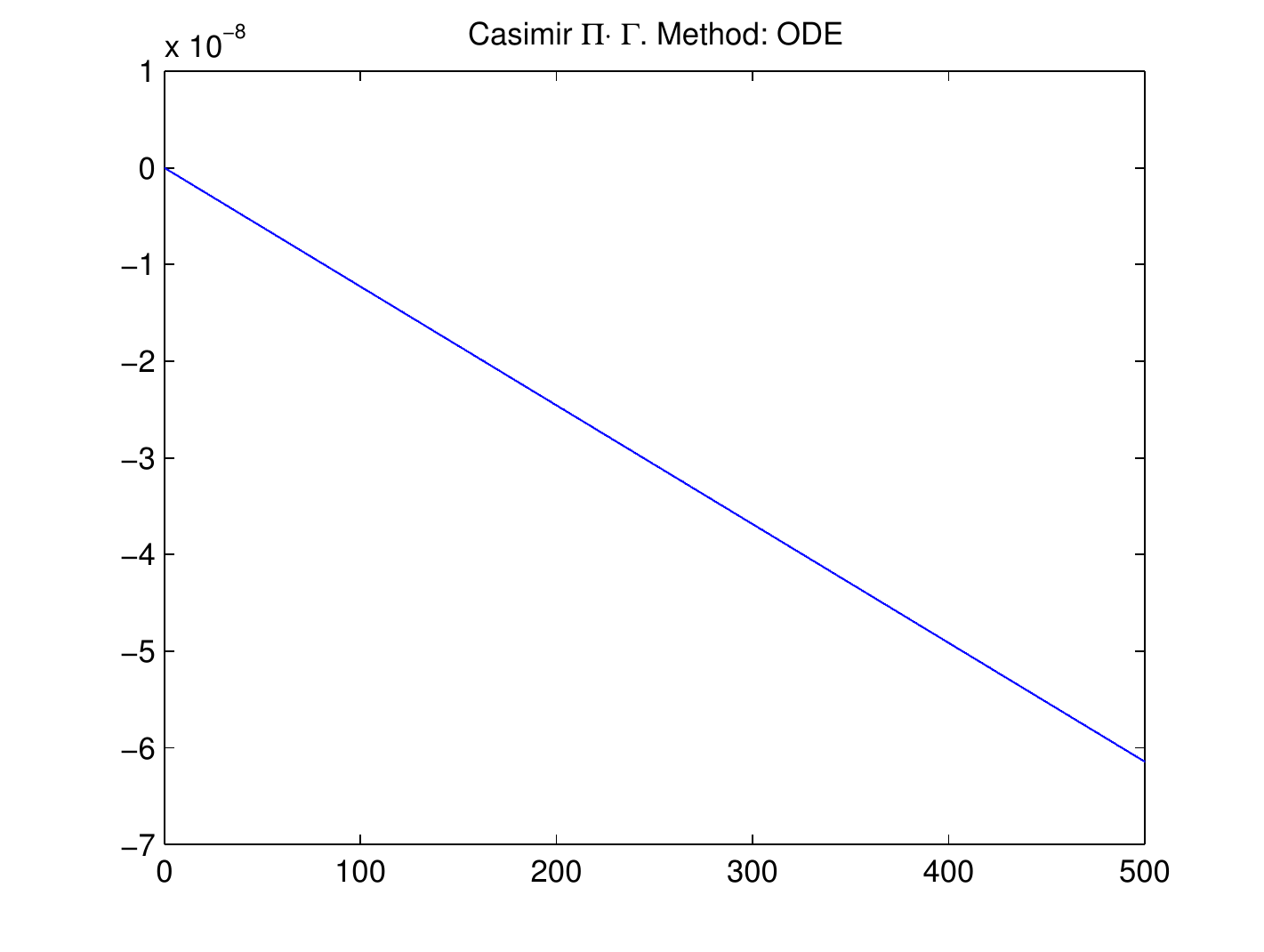}
\end{center}
\setlength\abovecaptionskip{-10pt}
\caption{Casimir conservation. Comparison of HJ and ode45 methods.}
\end{figure}
%%%%%%%%%%%%%%%%%%%%%%%%%%%%%%%%%%%%%%%%%%%%%%%%%%
\subsubsection{Elroy's Beanie}

This system evolves on an Atiyah algebroid.
It is probably the simplest example
of a dynamical system with a non-Abelian Lie
group of symmetries. It consists of two planar rigid
bodies attached at their centers of mass, moving
freely in the plane.
The configuration space is $Q=SE(2)\times S^1$ with coordinates
$(x, y, \theta, \psi)$, where the first three 
coordinates  describe the position and
orientation of the center of mass of the first
body. The last one is the relative orientation
between both bodies. The dynamics is determined by the Euler-Lagrange equations corresponding to the Lagrangian
\[
L(x, y, \theta, \psi, \dot{x}, \dot{y},
\dot{\theta}, \dot{\psi})= \frac{1}{2}m
(\dot{x}^2+\dot{y}^2)+\frac{1}{2}
I_1\dot{\theta}^2+\frac{1}{2}I_2
(\dot{\theta}+\dot{\psi})^2-V(\psi),
\]
where $m$ denotes the mass of the system and
$I_1$ and $I_2$ are the inertias of the first and
the second body, respectively; additionally, we
also consider a potential function of the form
$V(\psi)$. This Lagrangian is $SE(2)$-invariant  where the
group action is given by: \[ \Phi_g(q)=\left(
\begin{array}{c
}
z_1+x\cos\alpha-y\sin \alpha\\
z_2+x\sin\alpha+y\cos \alpha\\
\alpha+\theta\\
\psi
\end{array}
\right)
\]
with $g=(z_1, z_2, \alpha)$ and $q=(x, y, \theta, \psi)$.

The Lie algebra ${\mathfrak se}(2)$ is
generated by matrices of the form
\[
\hat{\xi}=\left(
\begin{array}{ccc}
0&\xi_3&\xi_1\\
-\xi_3&0&\xi_2\\
0&0&0
\end{array}
\right)
\]
with basis
\[
e_1=\left(
\begin{array}{ccc}
0&0&1\\
0&0&0\\
0&0&0
\end{array}
\right) ,\qquad e_2=\left(
\begin{array}{ccc}
0&0&0\\
0&0&1\\
0&0&0
\end{array}
\right) \qquad \textrm{ and  } e_3=\left(
\begin{array}{ccc}
0&1&0\\
-1&0&0\\
0&0&0
\end{array}
\right).
\]
Then, any $\xi\in {\mathfrak{se}(2)}$ is expressed as $\hat{\xi}=\xi_1 e_1+\xi_2 e_2+\xi_3 e_3$ and the Lie algebra structure on ${\mathfrak se}(2)$
is determined by 
\[
[e_1,e_2]=0,\qquad [e_1,e_3]=e_2,\qquad [e_2,
e_3]=-e_1.
\]

The quotient space $M=Q/G=(SE(2)\times S^1)/SE(2)\simeq S^1$ is
naturally parameterized by the coordinate $\psi$. The projection $\pi:
Q\longrightarrow M$ is given in coordinates by $\pi(x, y, \theta,
\psi)=\psi$, $TQ/G\simeq TS^1\times {\mathfrak{se}(2)}$ and  the reduced Lagrangian reads
\[
\mathfrak{l}(\psi, \dot{\psi},
\xi_1, \xi_2, \xi_3)=\frac{1}{2}m(\xi_1^2+\xi_2^2)+\frac{1}{2}I_1
\xi_3^2+\frac{1}{2}I_2(\xi_3+\dot{\psi})^2-V(\psi).
\]
The reduced equations (Lagrange-Poincar\'e equations)  are
\begin{align*}
\dot{\xi}_1&=-\xi_2\xi_3,\\
\dot{\xi}_2&=\xi_1\xi_3,\\
(I_1+I_2)\dot{\xi}_3+I_2\ddot{\psi}&=0,\\
I_2(\dot{\xi}_3+\ddot{\psi})&=-\displaystyle{\frac{\partial V}{\partial \psi}}.
\end{align*}

The Hamiltonian framework is defined on $T^*Q/G\equiv T^*S^1\times {\mathfrak se}^*(2)$ with coordinates
$(\psi, p_{\psi}; p_1, p_2, p_3)$. The linear Poisson bracket is given by: 
\[
\{p_1, p_2\}=0,\quad \{p_1, p_3\}=-p_2, \quad \{p_2, p_3\}=p_1, \quad \{p_i, p_{\psi}\}=\{p_i, \psi\}=0, \quad \{\psi, p_{\psi}\}=1
\]
and the corresponding Hamiltonian function is 
\[
H(\psi, p_{\psi}; p_1, p_2, p_3)=\frac{1}{2m}\left ({p_1}^2+p_2^2\right)+\frac{1}{2I_2}p_{\psi}^2+\frac{1}{2I_1}(p_3-p_\psi)^2+V(\psi),
\]
Using the Poisson structure and the Hamiltonian function, one can easily compute the Hamilton's equations 
\begin{align*}
\dot{p}_1&=-\frac{p_3-p_{\psi}}{I_1}p_2,\\
\dot{p}_2&=\frac{p_3-p_{\psi}}{I_1}p_1,\\
\dot{p}_3&=0,\\
\dot{p}_{\psi}&=-\displaystyle{\frac{\partial V}{\partial \psi}},\\
\dot{\psi}&=\frac{p_{\psi}}{I_2}-\frac{p_3-p_{\psi}}{I_1}.
\end{align*}
Now, to apply our Hamilton-Jacobi formalism, observe that $(Q\times Q)/G\simeq S^1\times S^1\times SE(2)$ and therefore $T^*(Q\times Q)/G)\simeq T^*S^1\times T^*S^1\times T^*SE(2)$. We compute the source and target map of the symplectic groupoid
$\tilde{\alpha}:  T^*S^1\times T^*S^1\times T^*SE(2)\rightarrow T^*S^1\times {\mathfrak se}^*(2)$ and  
$\tilde{\beta}:  T^*S^1\times T^*S^1\times T^*SE(2)\rightarrow T^*S^1\times {\mathfrak se}^*(2)$
\[
\begin{array}{l}
\tilde{\alpha}(\psi_1, p_{\psi}^1, \psi_2, p_{\psi}^2, (z_1, z_2, \theta),  P_1, P_2, P_3)=
(\psi_1, -p_{\psi}^1; P_1, P_2, P_3-P_1z_2+P_2z_1),
\\ \noalign{\medskip}
\tilde{\beta}
(\psi_1, p_{\psi}^1, \psi_2, p_{\psi}^2, (z_1, z_2, \theta),  P_1, P_2, P_3)=
(\psi_2, p_{\psi}^2; P_1\cos\theta+P_2\sin \theta, -P_1\sin\theta+P_2\cos \theta , P_3).
\end{array}
\]
 Therefore, the Hamilton-Jacobi equation is
 \[
 \frac{\partial S}{\partial t}(t,  p^1_{\psi}, \psi_2, p_x, p_y, p_z)+\hat{H}( \frac{\partial S}{\partial p_{\psi}}, -p_{\psi}^1, \psi_2, \frac{\partial S}{\partial \psi_2}; 
 \frac{\partial S}{\partial P_1},  \frac{\partial S}{\partial P_2}, \frac{\partial S}{\partial P_3}, P_1, P_2, P_3)=0
 \]
 where $\hat{H}=H\circ \tilde{\beta}$
 \begin{eqnarray*}
 \hat{H}(\psi_1, p_{\psi}^1, \psi_2, p_{\psi}^2, z_1, z_2, \theta,  P_1, P_2, P_3)&=&
 \frac{1}{2m}\left (({P_1\cos\theta+P_2\sin \theta})^2
 +(-P_1\sin\theta+P_2\cos \theta)^2\right)\\
 &&+\frac{1}{2I_2}(p^2_{\psi})^2+\frac{1}{2I_1}(P_3-p^2_\psi)^2.
\end{eqnarray*}

The method was implemented numerically using $m=3$, $I_1=5$, $I_2=1$, with $T=10$ and several time step values. The potential energy is  $V(\psi)=\cos(2\psi)$, and the initial conditions are $(\psi,p_\psi,p_1,p_2,p_3)=(1,-0.1,0.1,0.2,1)$. Figure~\ref{fig:elroy-errors} shows the global errors.
\begin{figure}[H]
\begin{center}
%\setlength\fboxsep{0pt}
%\setlength\fboxrule{0.5pt}
%\fbox{
\includegraphics[trim = 30mm 10mm 20mm 10mm, clip, scale=.52]{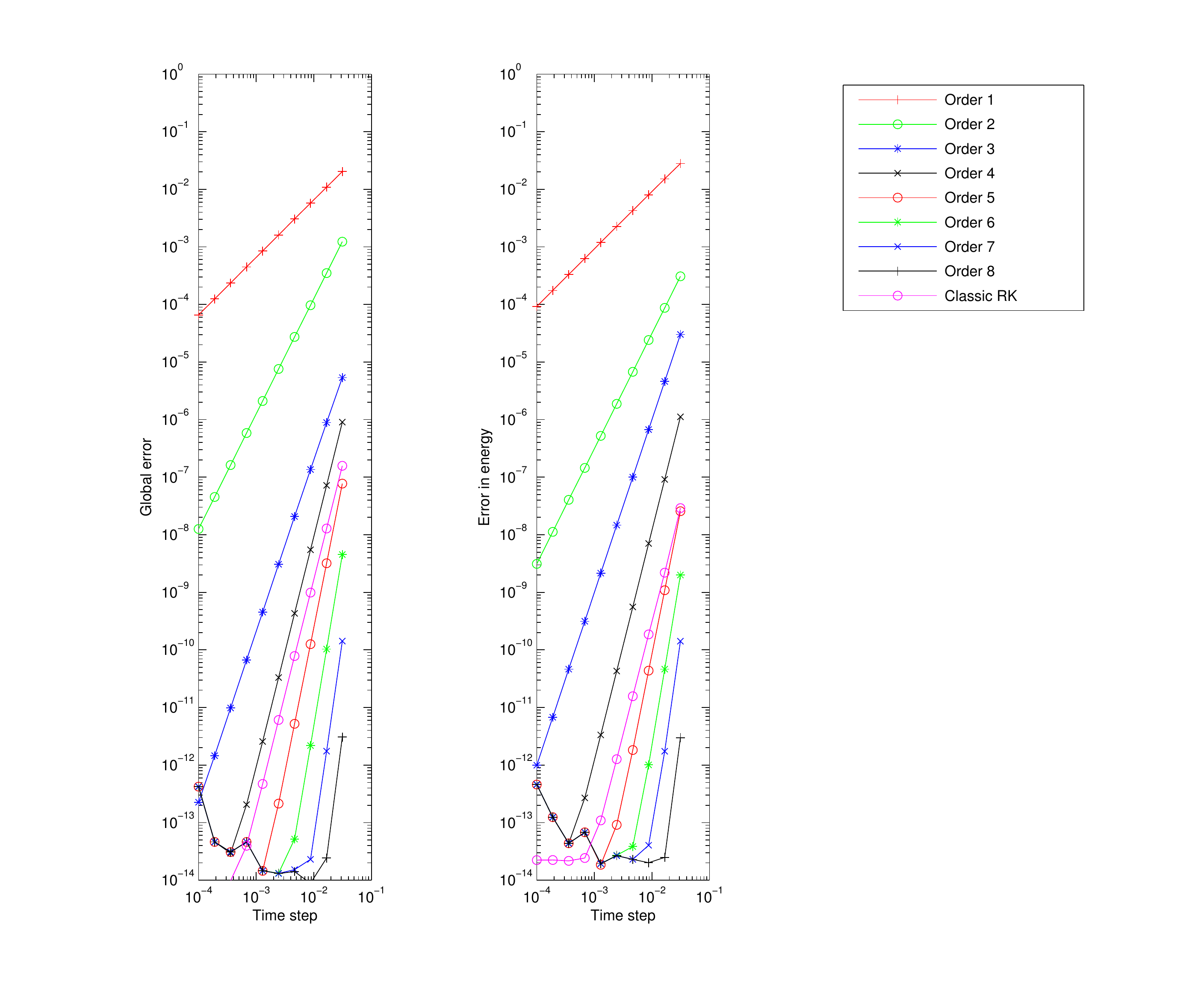}
%}
\end{center}
\setlength\abovecaptionskip{-10pt}
\caption{Global errors for the Elroy's beanie simulations, HJ method.}
\label{fig:elroy-errors}
\end{figure}

Next, we show the energy error for different orders of our Hamilton-Jacobi method, for $N=1000$. 

\begin{figure}[H]
\begin{center}
%\setlength\fboxsep{0pt}
%\setlength\fboxrule{0.5pt}
%\fbox{
\includegraphics[trim = 0mm 0mm 0mm 0mm, clip, scale=.5]{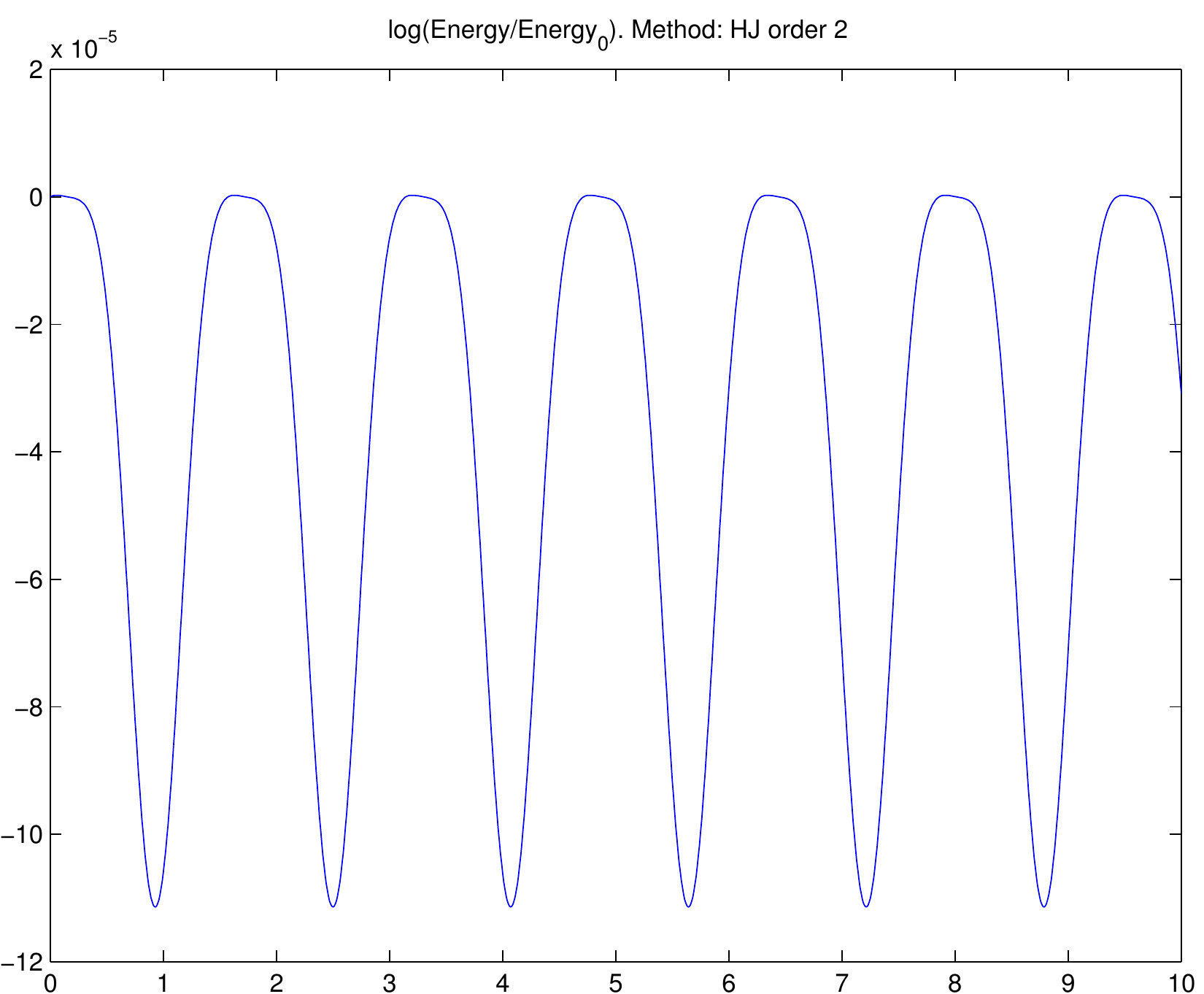}
\includegraphics[trim = 0mm 0mm 0mm 0mm, clip, scale=.5]{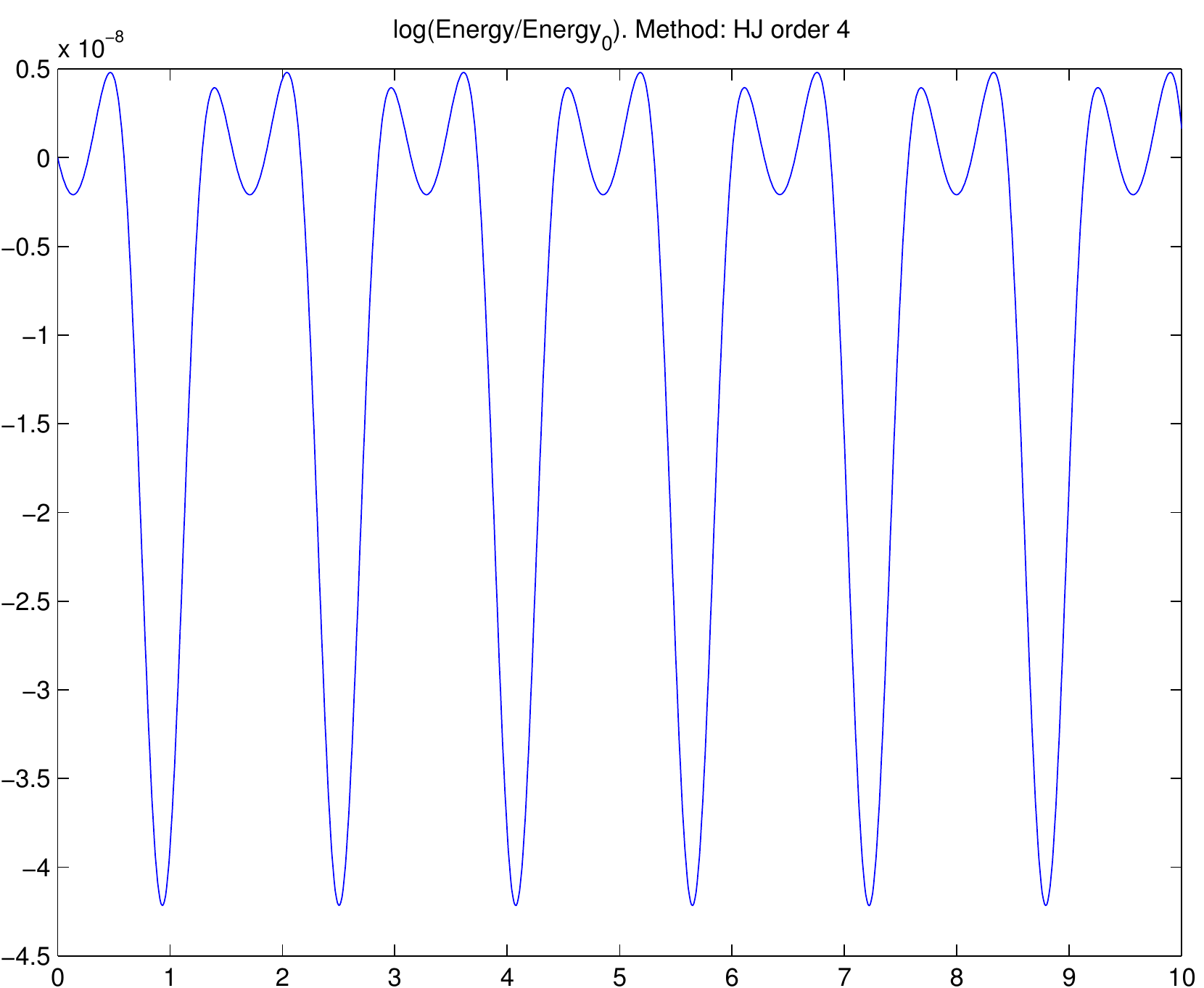}
%}
\end{center}
%\end{figure}
%
%\begin{figure}[H]
\begin{center}
%\setlength\fboxsep{0pt}
%\setlength\fboxrule{0.5pt}
%\fbox{
\includegraphics[trim = 0mm 0mm 0mm 0mm, clip, scale=.5]{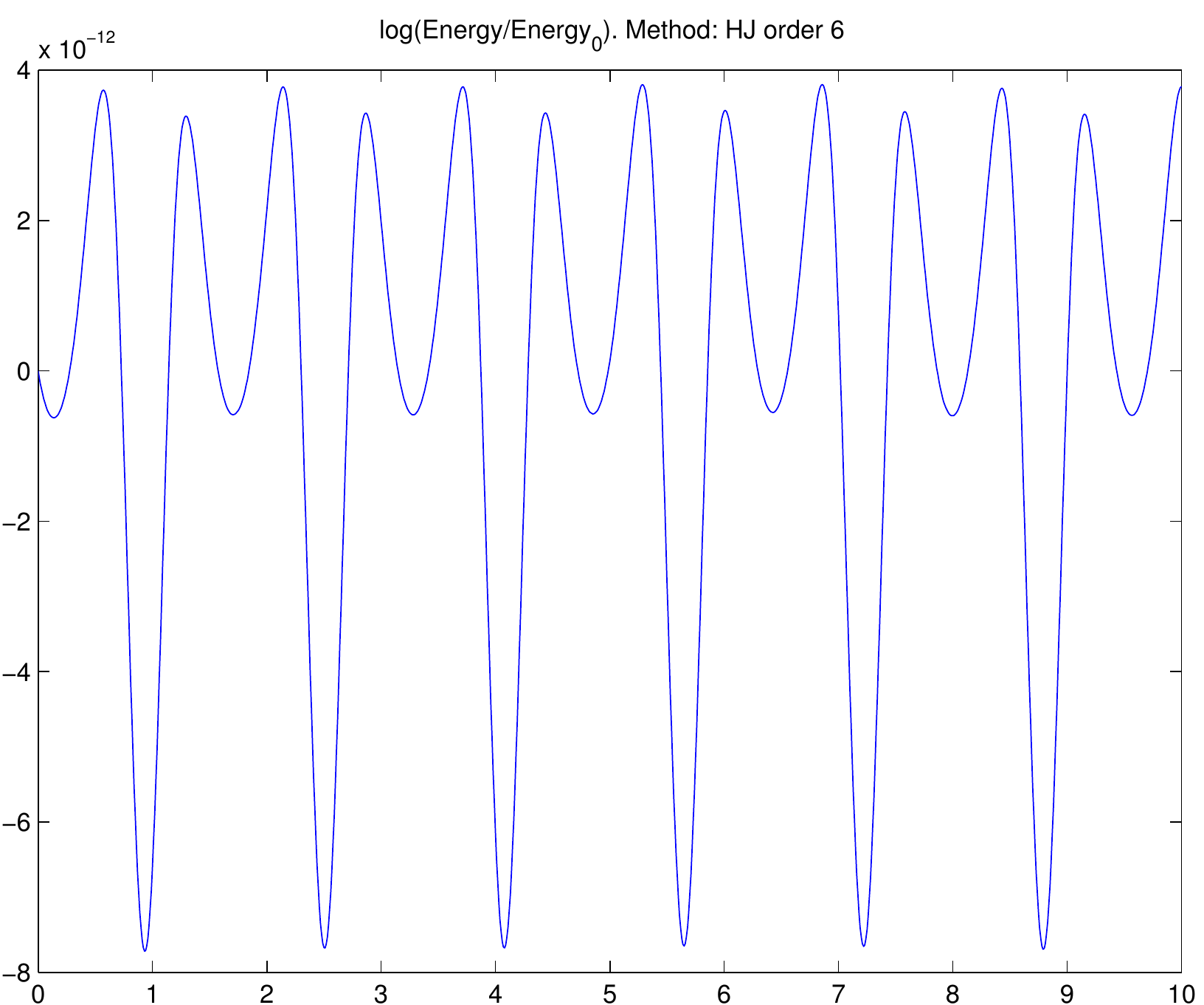}
\includegraphics[trim = 0mm 0mm 0mm 0mm, clip, scale=.5]{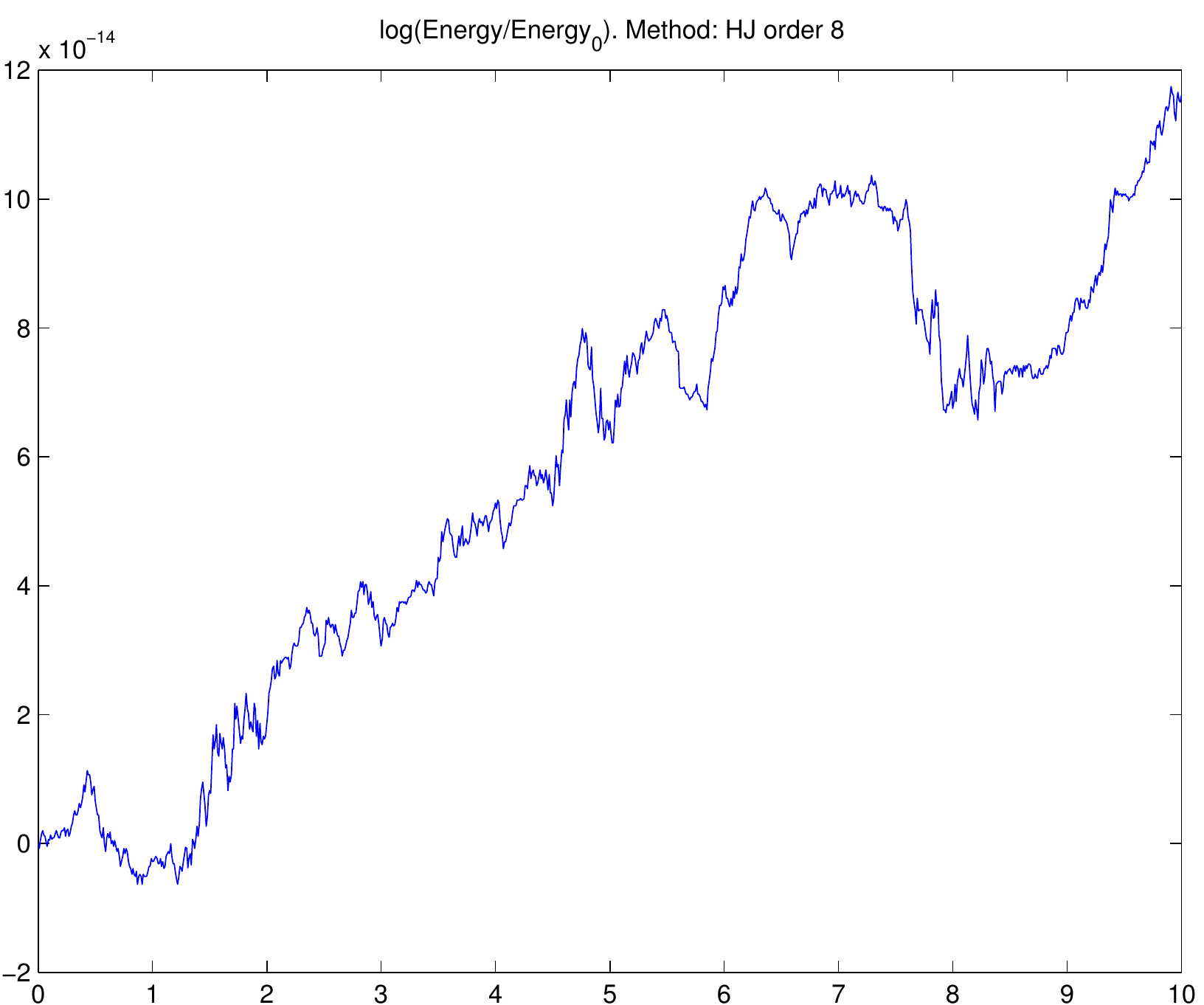}
%}
\end{center}
\setlength\abovecaptionskip{5pt}
\caption{Energy conservation, HJ method.}
\label{fig:HJ-energy-eb}
\end{figure}

Similar behavior to the previous cases can be observed. While for
``short'' times the behavior is similar to Runge-Kutta methods (of the
same order) the exceptional energy and Casimir conservation should
make our methods very well suited to study long term situations, which
are of huge importance to study, for instance, invariant submanifolds
(see below for the Runge Kutta's energy drift). Some improvements of these methods, see next section, should provide means to understand the dynamics of more complicated systems.

\begin{figure}[H]
\begin{center}
%\setlength\fboxsep{0pt}
%\setlength\fboxrule{0.5pt}
%\fbox{
\includegraphics[trim = 0mm 0mm 0mm 0mm, clip, scale=.5]{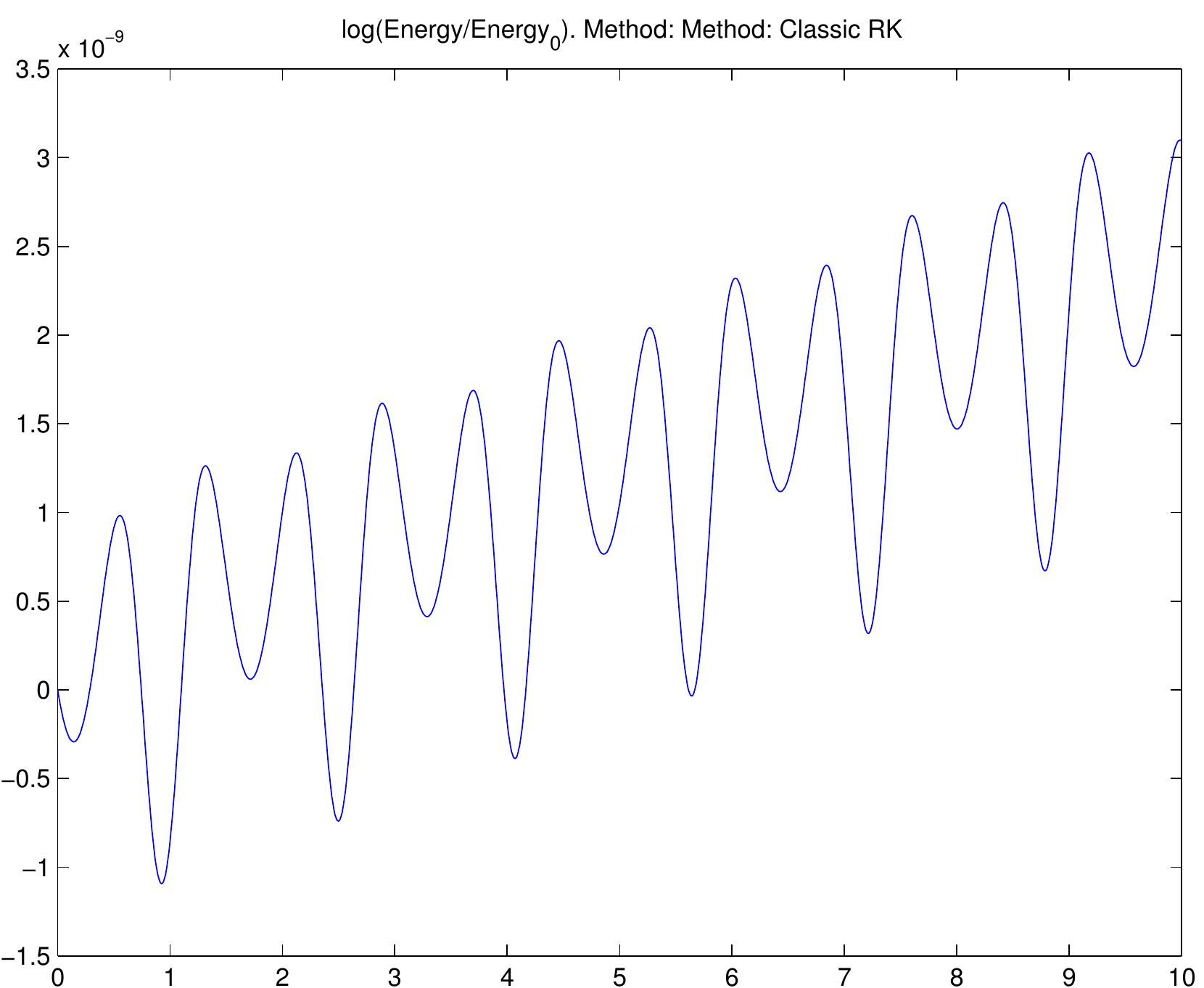}%}
\end{center}
\setlength\abovecaptionskip{5pt}
\caption{Energy conservation, Runge Kutta method.}
\label{fig:RKL-energy-eb}
\end{figure}

%%%%%%%%%%%%%%%%%%%%%%%%%%%%%%%%%%%%%%%%%%%%%%%%%%CONCLUSIONSANDOPENPROBLEMS%%%%%%%%%%%%%%%%%%%%%%%%%%%%%%%%%%%%%%%

\section{Conclusions and Remarks}\label{conclusions}

In this paper we developed a Hamilton--Jacobi theory for certain class
of linear Poisson structures which happen to be general enough to
include the systems important for classical mechanics. As a practical
application we present an improvement of some Poisson numerical
methods, previously introduced by Channell, Ge, Marsden and Scovel
among others. There are still several issues to be exploited, we list
some of them below.

\begin{enumerate}

\item {\it Simplification of the Hamilton--Jacobi equations using Casimirs:}\label{casimirs} In this paper we only treated some numerical methods as examples,  but it is our belief
  that applications of our results are very promising for analytic integration of Hamilton's
  equations. This should not be surprising, as the classical
  Hamilton--Jacobi theory has proved to be one of the most powerful
  tools for analytic integration, see Arnold's quote in Section \ref{introduction}. In this regard, Casimirs should play
  an important role, based on the following observation. If $H:A^*G\rightarrow
  \mathbb{R}$ is the Hamiltonian under consideration, and
  $C:A^*G\rightarrow \mathbb{R}$ is a Casimir, then
  $X_{H}=X_{H+\lambda C}$, where $\lambda$ is a constant. Nonetheless,
  the Hamilton--Jacobi equations for $H$ and for $H+\lambda C$ could be
  very different. As a simple but illustrative application of
  this fact, we present here an application to the computation of the
  rigid body when two moments of inertia are equal. It is remarkable
  that in \cite{McLachlan} the author uses a similar procedure to
  obtain explicit Lie--Poisson integrators. For more information about
  that relation see point \ref{Ruth} below.
	
	\begin{example}
	Let $G$ be the Lie group $SO(3)$ and $0<\varphi<2\pi, \ 0<\theta<\pi, \ 0<\psi<2\pi$ be the Euler angles, defined following \cite{MarsdenRatiu}. They form a coordinate chart, although not including the identity. The source and target of the cotangent groupoid read in the associated cotangent coordinates
	\[
	\begin{array}{l}
	\tilde\alpha(\varphi,\psi,\theta,p_\varphi,p_\psi,p_\theta)\rightarrow \left(\begin{array}{c}\Pi_1=(((p_\psi-\cos(\theta)p_\varphi)\sin(\varphi)+\cos(\varphi)\sin(\theta)p_\theta )/\sin(\theta)\\\Pi_2=((\cos(\theta)p_\varphi-p_\psi)\cos(\varphi)+\sin(\theta)\sin(\varphi)p_\theta )/\sin(\theta)\\\Pi_3=p_\varphi\end{array}\right),
	
	\\ \noalign{\bigskip}
	
\tilde\beta(\varphi,\psi,\theta,p_\varphi,p_\psi,p_\theta)\rightarrow \left(\begin{array}{c}\Pi_1=((p_\varphi-p_\psi \cos(\theta))\sin(\psi)+p_\theta \sin(\theta)\cos(\psi) )/\sin(\theta) \\\Pi_2=((p_\varphi-p_\psi\cos(\theta))\cos(\psi)-p_\theta \sin(\theta) \sin(\psi) )/\sin(\theta)\\\Pi_3=p_\psi\end{array}\right).
	
	\end{array}
	\]
	The Lagrangian submanifold
	\[
	\lag=\{(\pi,\pi/2,\pi/2,p_\varphi,p_\psi,p_\theta
)\textrm{ such that } p_\varphi,\ p_\psi, \ p_\theta \in\mathbb{R}\}
	\]
	generates the trasformation described by
        \[
\begin{array}{c}
\tilde\alpha(\pi,\pi/2,\pi/2,p_\varphi,p_\psi,p_\theta)=(-p_\theta,p_\psi,p_\varphi),
    
\\ \noalign{\medskip}
        \tilde\beta(\pi,\pi/2,\pi/2,p_\varphi,p_\psi,p_\theta)=(p_\varphi,-p_\theta,p_\psi),
\end{array}\]
        that is,
\[
\hat{\lag}(\Pi_1,\Pi_2,\Pi_3)=(\Pi_3,\Pi_1,\Pi_2).
\]
This transformation is not the identity, due to problems with the
parametrization of the identity using Euler's angles, but it is a very
easy transformation and invertible in a trivial way. Thus, it can be
used instead of the identity transfomation to achieve analogous results.
Following our previous construction we can give the generating
function \[S(p_\varphi,p_\psi ,p_\theta)=\pi p_\varphi+(1/2)\pi
p_\psi+(1/2)\pi p_\theta.\] The Hamiltonian of the rigid body dynamics,
when two moments of inertia are equal, and the Casimir are given by
	\[
	\begin{array}{l}
	H(\Pi_1,\Pi_2,\Pi_3)=\displaystyle\frac{1}{2}\left(\frac{\Pi_1^2}{I}+\frac{\Pi_2^2}{I}+\frac{\Pi_3^2}{I'}\right),\\ \noalign{\bigskip}
	C(\Pi_1,\Pi_2,\Pi_3)=\Pi_1^2+\Pi_2^2+\Pi_3^2.
	\end{array}
	\]

The dynamics of the Hamiltonians $H$ and
$H'=H+\displaystyle\frac{1}{2I}C$ are equal, since $C$ is a Casimir \[X_{H'}(\Pi_1,\Pi_2,\Pi_3)=\left(\displaystyle\frac{I-I'}{I I'}\Pi_2\Pi_3\right)\partial \Pi_1+\left(\displaystyle\frac{I'-I}{I I'}\Pi_1\Pi_3\right)\partial \Pi_2+0\partial \Pi_3.\]
Nonetheless, the Hamiltonians $H$ and $H'$ are very different as functions,
that implies that their $\tilde\beta$-pullback are quite disparate as well
\[
\begin{array}{rl}
H\circ \tilde{\beta}&=\displaystyle\frac{1}{2}\left(\frac{[(p_\varphi-p_\psi \cos(\theta))\sin(\psi)+p_\theta \sin(\theta) \cos(\psi)]^2}{I \sin^2(\theta)}\right.
\\ \noalign{\medskip}
&+
\displaystyle\frac{[(p_\varphi-p_\psi \cos(\theta))\cos(\psi)+p_\theta \sin(\theta) \sin(\psi)]^2}{I \sin^2(\theta)}
%\\ \noalign{\medskip}

+ \displaystyle\left.\frac{p_\psi^2}{I'}
\right)
\end{array}
\]
while
\[
\begin{array}{rl}
H'\circ \tilde{\beta}&= \left(\displaystyle\frac{I-I'}{I I'}\right){p_\psi^2}.
\end{array}
\]
This fact should have implications in our Hamilton--Jacobi theory, and this
is the case as we are going to show. A solution of the Hamilton--Jacobi
equation using the Hamiltonian $H'$ 
\[
\begin{array}{c}\displaystyle\frac{\partial S}{\partial
  t}(p_\varphi,p_\psi,p_\theta)+ (H'\circ \tilde{\beta})(\frac{\partial
  S}{\partial p_\varphi},\frac{\partial
  S}{\partial p_\psi},\frac{\partial
  S}{\partial p_\theta},p_\varphi,p_\psi,p_\theta)=0 \\
    \noalign{\medskip}

\Leftrightarrow \displaystyle\frac{\partial S}{\partial
  t}(p_\varphi,p_\psi,p_\theta)+  \left(\displaystyle\frac{I-I'}{I I'}\right){p_\psi^2}=0
\end{array}
\]
is given by direct inspection\[S
(t,p_\varphi,p_\psi,p_\theta)=- \left(\displaystyle\frac{I-I'}{I I'}\right){p_\psi^2}t+S_{0}(p_\varphi,p_\psi,p_\theta),\]where
$S_0$ is the initial condition, which is fixed to get the identity, so
finally
\[
S
(t,p_\varphi,p_\psi ,p_\theta)=-\left(\displaystyle\frac{I-I'}{I I'}\right){p_\psi^2}t+\pi p_\varphi+(1/2)\pi p_\psi +(1/2)\pi p_\theta.
\]
For the sake of simplicity we chose to integrate the easiest term of
the Hamiltonian, but the other terms are solvable in a similar
fashion. On the other hand, the Hamilton--Jacobi equation for the
Hamiltonian $H$ does not seem to be solvable in an obvious way. Although very elementary, this example shows that Casimirs can
be used to simplify the Hamilton--Jacobi equation. A systematic
development of these ideas should provide means to integrate the
Hamilton--Jacobi equation. These ideas are exclusive of the Poisson
setting, as the Casimirs are trivial for symplectic structures.
\end{example}

\item\label{improve} {\it Improvement of the numerical methods:} The numerical methods that
  we present here are very general and conserve the geometry very
  well, but they are not very efficient from the computational
  viewpoint. Our aim in this paper
  was to show how to use our results, rather than giving optimized
  numerical methods. Nonetheless, there is still a lot of room to improve
  these methods, as it has already been shown in
  \cite{BenzelGeScovel,ChannellScovel,Kang,MakazagaMurua,McLachlanScoveII}
  in the symplectic and Lie--Poisson case, where they give recipes to construct
  higher order approximations and reduce the computational cost. Those
  improvements can be applied in a straightforward fashion to our setting.

\item\label{Ruth} {\it Ruth type integrators:} Our methods are, generally,
  implicit. But in particular examples they can be
  made explicit sometimes. That is very important in order to develop
  very efficient methods, a nice exposition of these topics is given in \cite{McLachlan} . This is a classical fact, as it happens
  already in the symplectic case, classical references are
  \cite{ForestRuth,Ruth} where the authors develop fourth order
  explicit methods for mechanical Hamiltonians, that is, of the form
  kinetic plus potential energy. Our approach seems to be useful in
  that regard. For instance, consider the example introduced in
  \ref{casimirs} and assume that the moments of inertia of the rigid
  body are different, so the Hamiltonian reads,
\[
H(\Pi_1,\Pi_2,\Pi_3)=\displaystyle\frac{1}{2}\left(\frac{\Pi_1^2}{I^1}+\frac{\Pi_2^2}{I^2}+\frac{\Pi_3^2}{I^3}\right).
\]
We can use the Casimir to simplify one of the terms without changing
the dynamics, 
\[
H'=H-\displaystyle\frac{1}{2I^1}C=\frac{1}{2}\left(
  \frac{I^1-I^2}{I^1I^2}\Pi_2^2+\frac{I^1-I^3}{I^1I^3}\Pi_3^2\right)=C^1\Pi^2_2+C^2\Pi^2_3,
\]
where $C^1=\displaystyle\frac{I^1-I^2}{2I^1I^2}$ and
$C^2=\displaystyle\frac{I^1-I^3}{2I^1I^3}$. And using the previous expression for
$\tilde\beta$ we get
\[
H'\circ\tilde\beta=C^1\left(\displaystyle\frac{[(p_\varphi-p_\psi \cos(\theta))\cos(\psi)+p_\theta \sin(\theta) \sin(\psi)]^2}{ \sin^2(\theta)}
\right)+C^2p_{\psi}^2
\]
and the Hamilton--Jacobi equation becomes
\[
\displaystyle\frac{\partial S}{\partial
  t}+C^1\left(\displaystyle\frac{[(p_\varphi-p_\psi \cos(\frac{\partial S}{\partial
  p_\theta}))\cos(\frac{\partial S}{\partial
  p_\psi})+p_\theta \sin(\frac{\partial S}{\partial
  p_\theta}) \sin(\frac{\partial S}{\partial
  p_\psi})]^2}{ \sin^2(\frac{\partial S}{\partial
  p_\theta})}
\right)+C^2 (\frac{\partial S}{\partial
  p_\psi})^2=0.
\]
This equation is not easily integrable, but clearly $H=H^1+H^2$ where 
\[
\begin{array}{rl}
H^1&= C^1\left(\displaystyle\frac{[(p_\varphi-p_\psi \cos(\theta))\cos(\psi)+p_\theta \sin(\theta) \sin(\psi)]^2}{ \sin^2(\theta)}
\right)\\ \noalign{\bigskip}
H^2&=C^2p_{\psi}^2
\end{array}
\]
and the Hamilton--Jacobi equation for each of the Hamiltonians is
easily solvable by
\[
\begin{array}{rl}
S^1 (t,p_\varphi,p_\psi,p_\theta)&=-C^1{p_\theta^2}t+\pi p_\varphi+(1/2)\pi p_\psi+ (1/2)\pi p_\theta,
\\ \noalign{\medskip}
\textrm{and}
\\ \noalign{\medskip}
S^2 (t,p_\varphi,p_\psi,p_\theta)&=-C^2p_\psi^2t+\pi p_\varphi+(1/2)\pi p_\psi+(1/2)\pi p_\theta,
\end{array}
\]
respectively. Each of the solutions gives an explicit transformation which is a
rotation. The composition of the
explicit transformations induced by $S^1$ and $S^2$ gives the method
introduced by R.I. McLachlan in \cite{McLachlan}, but we obtained it
from the Hamilton--Jacobi theory. The development of  a rigorous Ruth
type integration techniques will provide very efficient numerical
methods which conserve the geometry. We want to stress here that all
the theoretical tools used in \cite{ForestRuth,Ruth}: the change of
the Hamiltonian under a canonical transformation, the different types
of generating functions,... were already introduced in this work. These
results applied to the linear Poisson setting were not present in the literature until now, as far as we know. The importance of the groupoid setting was already noticed by C. Scovel and A.D. Weinstein. We recall the quote from \cite{ScovelWeinstein}, {\it ``The groupoid aspect of the theory also provides natural Poisson
  maps, useful in the application of Ruth type
integration techniques, which do not seem easily derivable from the general theory of Poisson
reduction''.}

\item {\it Reduction of the Hamilton--Jacobi equation:} Some of the
  authors of this paper developed a reduction and reconstruction
  procedure for the Hamilton--Jacobi equation, see \cite{HamiltonJacobiSymmetries}, based on the following
  lemma (see \cite{HamiltonJacobiSymmetries,GaGuMaMe}).
\begin{lemma} Let $(M,\Omega)$ a symplectic manifold, $G$ a connected
  Lie group and $\Phi:G\times M\rightarrow M$ an action by
  symplectomorphisms. Assume that this action has an equivariant momentum mapping,
  say $J:M\rightarrow \mathfrak{g}^*$. Given a (connected) Lagrangian
  submanifold $\lag\subset M$, the following conditions are equivalent:
\begin{enumerate}
\item $\lag$ is a $G$-invariant Lagrangian submanifold,
\item $J_{|\lag}=\mu$, {\it i.e.}, $J$ is constant on $\lag$.
\end{enumerate}
\end{lemma}

The results introduced recently in \cite{LojaRatiuOrtega} should
permit the development of a reduction  theory applicable to our
framework. As an evidence of this fact, notice that in
\cite{LojaRatiuOrtega} the authors obtain a  momentum
mapping for an action which resembles the cotangent lifted
actions. That setting is very similar to the one used by the authors
in \cite{HamiltonJacobiSymmetries} and so it seems very likely to be
that reduction theory applies to our theory.

We also want to stress here that the reconstruction procedure introduced in the aforementioned work \cite{HamiltonJacobiSymmetries} can be combined with the Hamilton-Jacobi theory developed here to obtain symplectic integrators that conserve momentum mappings in the same way that Poisson integrators conserve the Casimirs.

\item{\it The reduced Hamilton-Jacobi theory: a discrete principal connection approach.} Let $H: T^*P \to \mathbb{R}$ be a Hamiltonian function which is invariant under the cotangent lift of a free and proper action of the Lie group $G$ on $P$. Denote by $M$ the space of orbits of the action of $G$ on $P$.  
Then, one may consider the reduced Hamiltonian function $H_{red}: T^*P/G \to \mathbb{R}$ and the corresponding Hamilton-Poincar\'e dynamics. 

Suppose that $S: \mathbb{R} \times (P \times P)/G \to \mathbb{R}$ is a solution of the Hamilton-Jacobi equation
\[
H_{red}^{ext} \circ dS = \beta_{\mathbb{R}}^*K, \; \; \mbox{ with } K: \mathbb{R} \times M \to \mathbb{R}.
\]
Note that the involved spaces in the previous equation are quotient manifolds. So, in order to write the Hamilton-Jacobi equation in a suitable way, we may use a discrete principal connection on the principal $G$-bundle $\pi: P \to M$ \cite{FeToZu,Leok,MMM}. In fact, in the presence of a discrete principal connection, one obtains a decomposition of the tangent bundle to $\mathbb{R} \times (P\times P)/G$ as follows
\[
T(\mathbb{R} \times (P \times P) / G) \simeq T\mathbb{R} \oplus T(M \times M) \oplus T((P \times G)/G)
\]
Thus, the Hamilton-Jacobi equation 
\[
d(H_{red}^{ext} \circ dS) = d(\beta_{\mathbb{R}}^*K)
\]
may be decomposed into the {\em horizontal Hamilton-Jacobi equation} and {\em the vertical Hamilton-Jacobi equation}. 
It would be interesting to check the efficiency of this method in some
explicit examples of symmetric Hamiltonian systems. This will be the
subject of a forthcoming paper. Anyway, the use of a (continuous)
principal connection has proved to be a very useful method in the
discussion of Hamilton-Poincar\'e (resp. Lagrange-Poincar\'e) equations
associated with a symmetric Hamiltonian (resp. Lagrangian) system (see
\cite{CeMaPeRa,CeMaRa,MaMiOrPeRa}). Even extensions of these methods
for higher-order mechanical systems and for classical field theories
have been also discussed in the literature (see \cite{EllGaHoRa,GaHoRa}).

\item {\it Truncation of infinite-dimensional Poisson systems using linear Poisson structures: } There are several
  examples of relevant physical importance which are
  infinite-dimensional Poisson systems: Euler equations of
  incompressible fluids, Vlasov--Maxwell and Vlasov--Poisson equations... Truncations
  of some of those systems
  conserving the geometry have already been carried out
  successfully \cite{McLachlan,ScovelWeinstein,FairlieZachos,Zeitlin}. In this
  regard, it
  seems that linear Poisson structures, {\it i.e.} dual bundles of Lie
  algebroids, should be the natural setting for that. After that
  truncation is done, our methods could be applied in order to understand the qualitative behavior of those infinite-dimensional systems.

\item {\it Poincar\'e's generating function:} Poincar\'e's generating functions have been used in dynamical systems in order to relate critical points of a function to periodic orbits. Our setting admits an analogous theory using the coordinates introduced previously. We describe these statements in the Lie algebra case below.

\begin{example}
Let $G$ be a Lie group with Lie algebra $\mathfrak{g}$. Consider $(g^i,p_i)$ a set of coordinates introduced following Section \ref{generatingalgebras}. Then the following lemma holds.
\begin{lemma}
Let $S(p_i)$ a function such that the corresponding Lagrangian submanifold $\{(\frac{\partial S}{\partial p_i},p_i) / p_i\in\mathbb{R}\}$ is a bisection, and let $\hat{S}:\mathfrak{g}^*\rightarrow \mathfrak{g}^*$ be the induced Poisson transformation. Then the critical points of $S$ correspond to fixed points of $\hat{S}$.
\end{lemma}

\begin{proof}
The Lagrangian submanifold that generates the identity is given in those coordinates by 
\[
\lag_{\id}=\{(0^i,p_i)\textrm{ such that }p_i\in\mathbb{R}\}
\]
and so $\lag_{\id}\cap \textrm{graph}(dS)=\{\textrm{critical points of } S\}$ which concludes the proof.
\end{proof}
\end{example}

Similar results can be obtained in the general situation.

\item {\it Other geometric settings: Non-holonomic mechanics.} In this
  paper we showed that, with the appropriate geometric tools, the
  classical complete solutions of the Hamilton--Jacobi theory can be
  extended to the Poisson case. This extension, far from being just a
  theoretical question,  gives means to study Hamiltonian dynamical
  systems, numerically and analytically. Recently, the Hamilton-Jacobi
  theory has been extended to other frameworks, like the non-holonomic
  mechanics, see
  \cite{CaGrMaMaMuRo2,HamiltonJacobiAlgebroids,Bloch,Bloch2}. Unfortunately,
  to the best of our knowledge, in the mentioned works there is not
  described a way to generate transformations from complete solutions
  of the Hamilton-Jacobi equation, analogous to the results described
  here. To obtain such a procedure will be extremely illuminating in
  order to construct non-holonomic integrators. Anyway,
  Hamilton--Jacobi theory for non-holonomic mechanical systems may be a
  useful method to obtain first integrals of the system
  \cite{CaGrMaMaMuRo2,GrilloPadron} which eventually will facilitate the
  integration of the systems. So, in conclusion, Hamilton--Jacobi
  theory could be used in the integration of some interesting
  symmetric non-holonomic systems which have been discussed, very
  recently, using new geometric tools (see \cite{Balseiro,Balseiro2,BaSa}).
\end{enumerate}

%%%%%%%%%%%%%%%%%%%%%%%%%%%%%%%%%%%%%%%%%%%%%%%%%%APPENDICES%%%%%%%%%%%%%%%%%%%%%%%%%%%%%%%%%%%%%%%
\begin{appendices}
\section{Lie Algebroids}\label{liealgebroids}
In this section we will recall the definition of  a Lie
algebroid (see \cite{CannasDaSilvaWeinstein,Mackenzie}). Associated to every Lie groupoid there is a Lie algebroid,
as can be seen in Appendix \ref{liegroupoids},
although the converse is not true. 

\subsection{Definition}
A \emph{Lie algebroid} is a vector bundle $\tau:A\rightarrow M$
endowed with the following data:
\begin{itemize}
\item A bundle map $\rho: A\rightarrow TM$ called the \emph{anchor}.

\item A \emph{Lie bracket on the space of sections} $\Gamma(\tau)$
  satisfying the Leibniz identity, i.e.
\[
\lcf X, fY\rcf=f\lcf X, Y\rcf+\mathcal{L}_{\rho(X)}(f)Y
\]
for all $X, \ Y\in \Gamma(\tau)$ and any $f\in C^\infty(M)$.
\end{itemize}

\subsection{Examples}

In this section, we introduce some examples of Lie algebroids. More
examples, maybe the most natural ones, will be sketched in Appendix
\ref{liegroupoids} when we talk about the Lie algebroid associated with
a Lie groupoid.

\subsubsection{Vector Fields}

It can be seen that given a manifold $M$ there is a one-to-one
correspondence between Lie algebroid structures on the trivial bundle
$M\times \mathbb{R}$ and vector fields $X$ on $M$. The vector bundle
$\tau:A=M\times \mathbb{R}\rightarrow M$ is given just by the projection
onto $M$.
\begin{enumerate}
\item The anchor map $\rho:A\rightarrow TM$ is given by
  $\rho(m,t)=tX(m)\in TM$.
\item Given $f, \ g\in C^\infty(M)$ then $\lcf f, g\rcf=fX(g)-gX(f)$.
\end{enumerate}

\subsubsection{$2$-forms}

Let $M$ be a manifold and consider now the bundle $A=TM\times
\mathbb{R}$ over $M$, where the vector bundle structure $\tau:A\rightarrow M$
is the obvious one. Any closed $2$-form $\omega$ on $M$ defines a Lie algebroid
estructure on that bundle:
\begin{enumerate}
\item The anchor $\rho:A\rightarrow TM$ is given by the projection
  onto the first factor.

\item The Lie bracket is given by $\lcf (X,f), (Y,g)\rcf=([X,Y], \ 
  X(g)-Y(f)+\omega(X,Y))$
\end{enumerate}

\begin{remark}
This case can be used to find examples of non-integrable Lie algebroids.
\end{remark}

\subsubsection{Poisson Manifolds}
One really important example of Lie algebroids is given by Poisson
manifolds. Let $(P,\Lambda)$ be a Poisson manifold, then we consider the cotangent bundle of $P$, $\pi_P:T^*P\rightarrow P$
and where the anchor and the Lie bracket are given by:
\begin{enumerate}
\item The anchor map is just the induced morphism $\Lambda^{\sharp}:T^*P\rightarrow TP$ by the Poisson tensor $\Lambda$.
\item The Lie bracket is given by $\lcf \alpha,\beta
  \rcf=\mathcal{L}_{\Lambda^\sharp(\alpha)}\beta-\mathcal{L}_{\Lambda^\sharp(\beta)}\alpha-d(\Lambda(\alpha,\beta))$
  for all $\alpha$ and $\beta$ $1$-forms on $P$.
\end{enumerate}

\subsection{The Poisson Structure of the Dual of a Lie Algebroid}\label{poisson}

Given a Lie algebroid $(\tau:A\rightarrow M, \ \rho,\ \lcf\cdot{}, \
\cdot{} \rcf)$, its dual $A^*$ has a natural linear Poisson structure
that we proceed to describe now. 

Given $X$ and $Y$ sections of $\tau$, then they determine linear
functions on $A^*$ that we denote by $\hat{X}$ and $\hat{Y}$. On the
other hand each $f\in C^\infty(M)$ determines a function $f\circ \tau$ which is
constant on the fibers. It can be seen that there exists
a unique Poisson structure on $A^*$ which satisfies
\[
\begin{array}{l}
\{\hat{X},\hat{Y}\}=\widehat{\lcf X, Y\rcf}, \\ \noalign{\medskip}
\{\hat{X},f\circ\tau\}=(\rho(X)(f))\circ\tau, \\ \noalign{\medskip}
\{f\circ\tau,g\circ\tau\}=0,
\end{array}
\]
for all $X,Y\in\Gamma(\tau)$ and $f,g\in C^\infty(M)$.

Once we choose a local basis of sections $e_j$ $j=1,\ldots,m$ and local coordinates
$x^i$ $i=1,\ldots,n$ on $M$. This system induces local coordinates in $A^*$ by
$(x^i,\mu_i)$, where $i=1,\ldots,n$, $j=1,\ldots,m$ and where
$\mu_i=\hat{e}_i$. If the structure functions and the
anchor map read locally 
\[
\lcf e_i, e_j\rcf=c^k_{ij}e_k\quad \textrm{ and } \quad
\rho(e_i)=\rho^j_i \displaystyle\frac{\partial}{\partial x^j},
\]
then the Poisson bracket on the coordinates $(x^i,\mu_j)$ reads
\[
\begin{array}{c}
\{x^i,x^j\}=0, \\ \noalign{\medskip} \{\mu_i,x^j\}=\rho_i^j, \\ \noalign{\medskip}\{\mu_i,\mu_j\}=c_{ij}^k\mu_k.
\end{array}
\]
Due to our conventions, this Poisson structure is the opposite to the
one that makes $\tilde\beta$ a Poisson mapping. So in our work we are
considering the Poisson structure $(A^*G,-\{\cdot{},\ \cdot{} \})$.
\section{Lie Groupoids}\label{liegroupoids}

We recall here the definition of a (Lie) groupoid. For more
information about this concept, we refer the reader to the monograph \cite{lecturesontheintegrabilityofliebrackets} and K.\ Mackenzie's
book \cite{Mackenzie}.

\subsection{Definition}
\paragraph{Groupoids:}
A \emph{groupoid} is a set $G$ equipped with the following  data:
\begin{enumerate}
\item another set $M$, called the \emph{base};
\item two surjective maps $\alpha\colon G\to M$ and $\beta\colon G\to M$, called, respectively,  the \emph{source} and \emph{target} projections; we visualize an element $g\in G$ as an arrow from $\alpha(g)$ to $\beta(g)$:
$$
\xymatrix{*=0{\stackrel{\bullet}{\mbox{\tiny
 $x=\alpha(g)$}}}{\ar@/^1pc/@<1ex>[rrr]_g}&&&*=0{\stackrel{\bullet}{\mbox{\tiny
$y=\beta(g)$}}}}
$$
\item A partial multiplication, or composition map, $m\colon G_2\to G$ defined on the subset $G_2$ of $G\times G$:
\[
G_2=\left\{ (g,h)\in G\times G\mid \beta(g)=\alpha(h) \right\}.
\]
The multiplication will be denoted for simplicity by $m(g, h)=gh$.  It verifies the following properties:
\begin{enumerate}
\item $\alpha(gh)=\alpha(g)$ and $\beta(gh)=\beta(h)$.
\item $(gh)k=g(hk)$.
$$\xymatrix{*=0{\stackrel{\bullet}{\mbox{\tiny
 $\alpha(g)=\alpha(gh)$}}}{\ar@/^2pc/@<2ex>[rrrrrr]_{gh}}{\ar@/^1pc/@<2ex>[rrr]_g}&&&*=0{\stackrel{\bullet}{\mbox{\tiny
 $\beta(g)=\alpha(h)$}}}{\ar@/^1pc/@<2ex>[rrr]_h}&&&*=0{\stackrel{\bullet}{\mbox{\tiny
 $\beta(h)=\beta(gh)$}}}}$$
\end{enumerate}
\item An \emph{identity section} $\epsilon\colon  M \to G$ such that
\begin{enumerate}
\item $\epsilon(\alpha(g))g=g$ and $g\epsilon(\beta(g))=g$ for all $g\in G$,
\item $\alpha(\epsilon(x))=\beta(\epsilon(x))=x$ for all $x\in M$.
\end{enumerate}

\item An \emph{inversion map} $\iota\colon  G \to G$, to be denoted simply by $\iota(g)=g^{-1}$, such that
\begin{enumerate}
\item $g^{-1}g=\epsilon(\beta(g))$ and $gg^{-1}=\epsilon(\alpha(g))$.
$$\xymatrix{*=0{\stackrel{\bullet}{\mbox{\tiny
 $\alpha(g)=\beta(g^{-1})$}}}{\ar@/^1pc/@<2ex>[rrr]_g}&&&*=0{\stackrel{\bullet}{\mbox{\tiny
 $\beta(g)=\alpha(g^{-1})$}}}{\ar@/^1pc/@<2ex>[lll]_{g^{-1}}}}$$
\end{enumerate}
\end{enumerate}

We will denote a groupoid $G$ over a base $M$ by $\xymatrix@1@C=1.5em{
G\ar@<.3ex>[r]^{\alpha}\ar@<-.3ex>[r]_{\beta}&M
}$ or simply $G\rightrightarrows M$.

It is easy to see that $\epsilon$ must be injective, so there is a
natural  identification between $M$ and $\epsilon(M)$. However, we will keep a distinction between the two sets.

\paragraph{Lie Groupoids:}

A groupoid, $G \rightrightarrows M$, is said to be a \emph{Lie
groupoid} if $G$ and $M$ are differentiable manifolds, all the structural maps
are differentiable and besides, $\alpha$ and $\beta$ differentiable
submersions. If $G \rightrightarrows M$ is a Lie groupoid then $m$
is a submersion, $\epsilon$ is an embedding and $\iota$ is a
diffeomorphism. Notice that since $\alpha$ and $\beta$ are
submersions, the $\alpha$ and $\beta$-fibers are submanifolds. The same properties
imply that $G_2$ is a submanifold. We will use $G^x=\alpha^{-1}(x)$,
$G_y=\beta^{-1}(y)$ and $G^x_y=\alpha^{-1}(x)\cap\beta^{-1}(y)$.

\paragraph{Left and Right multiplication:}
Given $g\in G_y^x$, so $g:x\to y$, we can define two (bijective) mappings
$l_{g}\colon G^y\to G^x$ and and $r_{g}\colon G_x\to G_y$, which are the
\emph{left translation by $g$} and the \emph{right translation by $g$}
respectively. These diffeomorphisms are given by
\begin{equation}\label{leftrightmultiplication}
\begin{array}{lcr}
\begin{array}{rccl}
l_{g}\colon &G^y&\longrightarrow &G^x \\ 
&h &\mapsto&l_{g}(h) = gh
\end{array}
&; &
\begin{array}{rccl}
r_{g}\colon &G_x&\longrightarrow& G_y \\ 
&h& \mapsto &r_{g}(h) = hg,
\end{array}
\end{array}
\end{equation}
where we have that $(l_{g})^{-1} = l_{g^{-1}}$ and $(r_{g})^{-1} = r_{g^{-1}}$.

\paragraph{Bisections:}
A submanifold $\mathcal{B}\subset G$ is called a bisection of $G$ if the
restricted maps, $\alpha_{|\mathcal{B}}: \ \mathcal{B}\rightarrow M$ and $\beta_{|\mathcal{B}}: \
\mathcal{B}\rightarrow M$ are both diffeomorhisms. Consequently, for any
bisection $\mathcal{B}\subset G$, there is a  corresponding $\alpha$-section
$\mathcal{B}_\alpha =(\alpha_{|\mathcal{B}})^{-1}: \ M\rightarrow \mathcal{B}$, where $\beta\circ
\mathcal{B}_\alpha:\ M\rightarrow M$ is a diffeomorphism. Likewise, there is a
$\beta$-section $\mathcal{B}_\beta=(\beta_{|\mathcal{B}})^{-1}: \ M\rightarrow G$, where
$\alpha\circ \mathcal{B}_{\beta}=(\beta\circ \mathcal{B}_\alpha)^{-1}:\ M \rightarrow M $
is a diffeomorphism. More generally, $\mathcal{B}\subset G$ is called a local
bisection if the restricted maps $\alpha_{|\mathcal{B}}$ and $\beta_{|\mathcal{B}}$ are
local diffeomorphisms onto open sets, $U,\ V\subset M$,
respectively. Local bisections on a Lie groupoid always exist.

\paragraph{Invariant Vector Fields:}
A vector field $X$ on $G$ is said to be
\emph{left-invariant} (resp., \emph{right-invariant}) if it is
tangent to the fibers of $\alpha$ (resp., $\beta$) and
\[X(gh) = (T_{h}l_{g})(X(h))  \quad
\Big(\textrm{ resp. }X(gh)= (T_{g}r_{h})(X(g)) \Big)\]for all $(g,h) \in G_{2}$.

\paragraph{Morphisms:}
Given two Lie groupoids $G \rightrightarrows M$ and $G'
\rightrightarrows M'$, a \emph{morphism of Lie groupoids} is a
smooth map $\Phi\colon  G \to G'$ such that
\begin{enumerate}
\item If $(g, h) \in G_{2} $ then $ (\Phi(g), \Phi(h)) \in (G')_{2}$ and
\item $\Phi(gh) = \Phi(g)\Phi(h)$.
\end{enumerate}

A morphism of Lie groupoids $\Phi\colon  G \to G'$ induces a smooth map
$\Phi_{0}\colon  M \to M'$ in such a way that the source, the
target and the identity section commute with the morphism, i.e.
\[
\alpha' \circ \Phi = \Phi_{0} \circ \alpha, \makebox[.3cm]{}
\beta' \circ \Phi = \Phi_{0} \circ \beta, \makebox[.3cm]{} \Phi
\circ \epsilon = \epsilon' \circ \Phi_{0},
\]
$\alpha$, $\beta$ and $\epsilon$ (resp., $\alpha'$, $\beta'$ and
$\epsilon'$) being the source, the target and the identity sections
of $G$ (resp., $G'$).

\subsection{Lie Algebroid Associated to a Lie Groupoid}

Given a Lie groupoid $G$ we denote by
$AG=\ker(T\alpha)_{|\epsilon(M)}$, i.e., the set of vectors tangent to
the $\alpha$-fibers restricted to the units of the groupoid. Since the
units, $\im(\epsilon)$, are diffeomorphic to the base manifold, $M$, we
will consider  the set $AG$ as a vector bundle $\tau\colon  AG \to M$. The
reader should keep this identification in mind, because it is going to be
used implicitly in some places ($M\equiv \im(\epsilon)\subset G$).

It is easy to prove that there exists a
bijection between the space of sections $\Gamma(\tau)$ and the set of
left-invariant (resp., right-invariant) vector fields on $G$. If $X$
is a section of $\tau\colon  AG \to M$, the corresponding left-invariant
(resp., right-invariant) vector field on $G$ will be denoted
$\lvec{X}$ (resp., $\rvec{X}$), where
\begin{equation}\label{linv}
\lvec{X}(g) = (T_{\epsilon(\beta(g))}l_{g})(X(\beta(g))),
\end{equation}
\begin{equation}\label{rinv}
\left(\textrm{resp.,\ } \rvec{X}(g) = -(T_{\epsilon
(\alpha(g))}r_{g}\circ \iota)( X(\alpha(g)))\right),
\end{equation}
for $g \in G$.

Using the above facts, we may introduce a Lie algebroid structure
$(\lcf\cdot , \cdot\rcf, \rho)$ on $AG$:
\begin{enumerate}

\item The \emph{anchor map} $\rho\colon AG\to TM$ is
\[
\rho(X)(x) = (T_{\epsilon(x)}\beta)(X(x))
\]
for $X\in \Gamma(\tau)$ and $x \in M$.

\item The \emph{Lie bracket} on the space of sections   $\Gamma(\tau)$,
  denoted by $\lcf \cdot, \cdot\rcf$ is defined by
\[
\lvec{\lcf X, Y\rcf} = [\lvec{X}, \lvec{Y}],
\]
for $X, Y \in \Gamma(\tau)$ and $x \in M$. 

Note that
\[
\rvec{\lcf X, Y\rcf} = -[\rvec{X}, \rvec{Y}], \makebox[.3cm]{}
[\rvec{X}, \lvec{Y}] = 0,
\]
\[
T\iota\circ \rvec{X}=-\lvec{X}\circ \iota,\;\;\;\; T\iota\circ
\lvec{X}=-\rvec{X}\circ \iota,
\]
(for more details, see \cite{Mackenzie}). The dual bundle of $AG$ will be denoted by $A^*G$.
\end{enumerate}

In addition, define the vector bundle $V\alpha$ as the sub-bundle of
$TG$ consisting of $\alpha$\emph{-vertical vectors}, that is, vectors tangent to the $\alpha$-fibers. $V\beta$ is defined analogously. Thus $AG$ is the restriction of $V\alpha$ to $\epsilon(M)$.

\subsection{Examples of Lie groupoids}\label{Subsec-2.3}

Next, we will present some examples of Lie groupoids. The
corresponding associated Lie algebroid is pointed out in each case.

\subsubsection{Lie Groups}

 Any {\it Lie group} $G$ is a Lie groupoid over
$\{ \ide \}$, the identity element of $G$. 
\begin{enumerate}
 \item The source, $\alpha$, is the constant map $\alpha(g)=\ide$.
\item The target, $\beta$, is the constant map $\beta(g)=\ide$.
\item The identity map is $\epsilon(\ide)=\ide$.
\item The inversion map is $\iota(g) =g^{-1}$.
\item The multiplication is $m(g,h) =g\cdot h$, for
any $g$ and $h$ in $G$.
\end{enumerate}

\paragraph{Associated Lie algebroid:}
The Lie algebroid
associated with $G$ is just the {\it Lie algebra} ${\frak g}$ of $G$ in a
straightforward way.

\subsubsection{The Pair or Banal Groupoid}

 Let $M$ be a manifold. The
product manifold $M \times M$ is a Lie groupoid over $M$ called the \emph{pair or banal
groupoid}.  Its structure mappings are:
\begin{enumerate}
 \item The source, $\alpha$, is the projection onto the first
   factor.
\item The target, $\beta$, is the projection onto the second factor.
\item The identity map is $\epsilon(x)
= (x, x)$, for all $x \in M$.
\item The inversion map is $\iota(x, y) = (y, x)$.
\item The multiplication is $m((x, y), (y, z)) = (x, z)$, for
$(x, y), (y, z) \in M \times M$.
\end{enumerate}

\paragraph{Associated Lie Algebroid:} 
If $x$ is a point of $M$, it follows that
\[
A=\ker(T\alpha)_{\epsilon(x)}=\{0_x\}\times T_xM
\]
which gives the vector bundle structure and given $(0_x,X_x)\in A$,
then \[\tau (0_x,X_x)=x.\]
\begin{enumerate}
\item The anchor is given by the projection over the second factor
  $\rho(0_x,X_x)=X_x\in T_xM$
\item The Lie bracket on the space of sections, $\Gamma(\tau)$, is the
  Lie bracket of vector fields on the second factor $\lcf (0,X),(0,Y)\rcf=(0,[X,Y])$.
\end{enumerate}
In conclusion, we have that the Lie algebroid $A(Q \times Q) \to Q$ may be identified 
with the standard Lie algebroid $\tau_Q: TQ \to Q$.

\subsubsection{Atiyah or Gauge Groupoids}\label{atiyah}

Let $\pi\colon  P \rightarrow M$ be a
principal $G$-bundle. Then the free action $\Phi\colon  G \times P \to
P$ induces the diagonal action  $\Phi'\colon  G \times (P \times
P) \to P\times P$ by $\Phi'(g, (q, q')) = (gq, gq')$. Moreover, one may consider the quotient manifold $(P
\times P) / G$ and it admits a Lie groupoid structure over $M$,  called the
\emph{Atiyah or Gauge groupoid} (see, for instance, \cite{Mackenzie,
MMM}). We describe now the structural mappings.
\begin{enumerate}
\item The source, ${\alpha}\colon  (P \times P) / G \to M $ is
  given by $[(q,
q')] \mapsto \pi(q)$.
\item The target, ${\beta}\colon  (P \times P) / G \to M$  is given by
  $[(q,
q')] \mapsto \pi(q')$.
\item The identity map, ${\epsilon}\colon  M \to (P \times P) / G$ is
$x\mapsto [(q, q)], \; \mbox{ if } \pi(q) = x.$
\item The inversion map, ${\iota}\colon  (P \times P) / G \to (P
  \times P) / G$ is $ [(q, q')] \mapsto [(q', q)]$.
\item The multiplication map ${m}\colon  ((P \times P) / G)_{2} \to (P
  \times P) / G$ is $ ([(q, q')], [(gq', q'')]) \mapsto [(gq, q'')]$.
\end{enumerate}

\paragraph{Associated Lie Algebroid:}
It easily follows that $A=\ker{T\alpha}_{\epsilon(M)}$ can be
identified with $TP/G$. Then the associated Lie algebroid is just
$\tau: TP/G\rightarrow M$, where $\tau$ is the obvious
projection and the Lie algebroid structure is provided by
\begin{enumerate}
\item The anchor, $\rho:TP/G\rightarrow TM$, is given by the quotient
  of the natural projection map $T\pi:TP\rightarrow TM$. That is,
  $\rho=\tilde{T\pi}:TP/G\rightarrow TM$.
\item 
The space of sections of the vector bundle $\tau: TP/G \to P/G = M$ may be identified with the set of $G$-invariant vector fields on $P$. Under this identification, the Lie bracket on the space of sections is given by the standard Lie
  bracket of vector fields. We remark that it is easy to see that the Lie bracket of
  two $G$-invariant vector fields is another $G$-invariant vector field.
\end{enumerate}

\subsubsection{Action Lie groupoids}\label{action}

 Let $G$ be a Lie group and let ${\Phi}\colon M\times G\to M$, $(x,g)\mapsto xg$, be a right action of $G$  on $M$.  Consider the \emph{action Lie groupoid} $M\times G$ over $M$ with
structural maps given by
\begin{enumerate}
\item The source is ${\alpha}(x,g)=x$.
\item The target is ${\beta}(x,g)=xg$.
\item The identity map is ${\epsilon }(x)=(x,{e})$.
\item The inversion map is ${\iota}(x,g)=(xg, g^{-1})$.
\item The multiplication is $m((x,g),(xg,g'))=(x,gg')$.
\end{enumerate}
See, for instance, \cite{Mackenzie,
MMM} for the details.

\paragraph{Associated Lie Algebroid:}
Now, let ${\mathfrak g}=T_{{ e}}G$ be the Lie algebra of $G$. Given
$\xi\in\mathfrak{g}$ we will denote by $\xi_M$ the infinitesimal generator of the action
$\Phi\colon M\times G\to M$. Consider now the
vector bundle $\tau:M\times \mathfrak{g}\rightarrow M$ where $\tau$ is
the projection over the first factor, endowed with the following
structures:
\begin{enumerate}
\item The anchor is $\rho(x,\xi)=\xi_M(x)$.
\item The Lie bracket on the space of sections is given by $\lcf
\widetilde{\xi},\widetilde{\eta}\rcf(x)=[\widetilde{\xi}(x),\widetilde{\eta}(x)]
+
(\widetilde{\xi}(x))_M(x)(\widetilde{\eta})-(\widetilde{\eta}(x))_M(x)(\widetilde{\xi})$
for $\widetilde\xi$, $\widetilde\eta\in \Gamma(\tau)$.
\end{enumerate}
The resultant Lie algebroid is just the Lie algebroid of the action Lie groupoid $M \times G \rightrightarrows M$.
\end{appendices}

\vspace{0.2in}

\section*{Acknowledgments}
This work has been partially supported by MINEICO, MTM 2013-42
870-P, MTM 2015-64166-C2-2-P, MTM 2016-76702-P,  the European project IRSES-project ``Geomech-246981'' and the ICMAT Severo Ochoa project SEV-2011-0087 and 
SEV-2015-0554.
M. Vaquero wishes to thank MINEICO for a FPI-PhD Position,
BES-2011-045780,  and David Iglesias-Ponte and Luis Garc\'ia-Naranjo
for useful discussions. The research of S. Ferraro has been supported
by CONICET (PIP 2013--2015 GI 11220120100532CO), ANPCyT (PICT
2013-1302) and SGCyT UNS. We would like to thank the referees for their comments and insight, which helped us to improve our manuscript.

\nocite{*}
\bibliography{hj}
\bibliographystyle{acm}
\end{document}